\documentclass[pdflatex,sn-mathphys-num]{sn-jnl}

\usepackage{anyfontsize}
\usepackage[T1]{fontenc}
\usepackage{graphicx}%
\usepackage{multirow}%
\usepackage{amsmath,amssymb,amsfonts}%
\usepackage{amsthm}%
\usepackage{mathrsfs}%
\usepackage[title]{appendix}%
\usepackage{xcolor}%
\usepackage{textcomp}%
\usepackage{manyfoot}%
\usepackage{booktabs}%
\usepackage{algorithm}%
\usepackage{algorithmicx}%
\usepackage{algpseudocode}%
\usepackage{listings}%
\usepackage{hhline}
\usepackage{multicol, multirow}

\usepackage{framed}
\usepackage{eurosym}
\newcommand{\Euro}{\textup{\euro}}
\usepackage{amsthm}
\usepackage{graphicx}
\usepackage{color}
\usepackage{longtable}
\usepackage[normalem]{ulem}
\usepackage[caption=false]{subfig}
\usepackage{tikz}
\usetikzlibrary{calc}
\usetikzlibrary{shapes.geometric}
\usetikzlibrary{shapes,decorations,arrows,calc,arrows.meta,fit,positioning}
\tikzset{
	-Latex,auto,node distance =1 cm and 1 cm,semithick,
	state/.style ={ellipse, draw, minimum width = 0.7 cm},
	point/.style = {circle, draw, inner sep=0.04cm,fill,node contents={}},
	bidirected/.style={Latex-Latex,dashed},
	el/.style = {inner sep=2pt, align=left, sloped},
	square/.style={regular polygon,regular polygon sides=4}
}
\newcommand{\myassetnode}[5]{%
	\node[state,circle] (#1) at (#2,#3) {#4};
	\node[state,square,scale=0.5, fill=white!100] (#1asset) at (#2+0.3, #3-0.3) {#5};
}
\newcommand{\myassetnodenew}[5]{%
	\node[state,circle] (#1) at (#2,#3) {#4};
	\node[state,rectangle, minimum width=30pt, minimum height=12pt, scale=0.8, fill=white!100] (" ") at (#2+0.42, #3-0.36) {};
	\node[scale=0.5] at (#2+0.42, #3-0.36) {#5};
}
\newcommand{\mynode}[4]{%
	\node[state,circle] (#1) at (#2,#3) {#4};%
}

\usepackage{array}
\theoremstyle{definition}
\newtheorem{thm}{Theorem}
\newtheorem{lem}{Lemma}

\newtheorem{clm}{Claim}
\newtheorem{definition}{Definition}
\def\qed{\hfill\hbox{${\vcenter{\vbox{
	\hrule height 0.4pt\hbox{\vrule width 0.4pt height 6pt						
	\kern5pt\vrule width 0.4pt}\hrule height 0.4pt}}}$}}

\usepackage{vmargin}
\setmarginsrb{3.5cm}{2.25cm}{3.5cm}{3.05cm}{0.3cm}{0.3cm}{-0.3cm}{1.0cm}
\usepackage{wrapfig}

\DeclareMathOperator{\sumop}{sum}

\begin{document}

\title[Payment Scheduling in the Interval Debt Model]{Payment Scheduling in the Interval Debt Model}

\author[1]{\fnm{Tom} \sur{Friedetzky}}\email{tom.friedetzky@durham.ac.uk}

\author*[1]{\fnm{David C.} \sur{Kutner}}\email{david.c.kutner@durham.ac.uk}

\author[1]{George B. Mertzios }\email{george.mertzios@durham.ac.uk}
\equalcont{Partially supported by the EPSRC grant EP/P020372/1.}

\author[1]{\fnm{Iain A.} \sur{Stewart}}\email{i.a.stewart@durham.ac.uk}

\author[1]{\fnm{Amitabh} \sur{Trehan}}\email{amitabh.trehan@durham.ac.uk}

\affil[1]{\orgdiv{Department of Computer Science}, \orgname{Durham University}, \orgaddress{Upper Mountjoy Campus, Stockton Road, Durham DH1 3LE, UK}}

\abstract{
    The network-based study of financial systems has received considerable attention in recent years but has seldom explicitly incorporated the dynamic aspects of such systems. We consider this problem setting from the temporal point of view and introduce the Interval Debt Model (IDM) and some scheduling problems based on it, namely: \textsc{Bankruptcy Minimization/Maximization}, in which the aim is to produce a payment schedule with at most/at least a given number of bankruptcies; \textsc{Perfect Scheduling}, the special case of the minimization variant where the aim is to produce a schedule with no bankruptcies (that is, a perfect schedule); and \textsc{Bailout Minimization}, in which a financial authority must allocate a smallest possible bailout package to enable a perfect schedule. We show that each of these problems is NP-complete, in many cases even on very restricted input instances. On the positive side, we provide for \textsc{Perfect Scheduling} a polynomial-time algorithm on (rooted) out-trees although in contrast we prove NP-completeness on directed acyclic graphs, as well as on instances with a constant number of nodes (and hence also constant treewidth). When we allow non-integer payments, we show by a linear programming argument that the problem \textsc{Bailout Minimization} can be solved in polynomial time.
}


\maketitle

	\section{Introduction}
	
	 A natural problem in the study of financial networks is that of whether and where a failure will occur if no preventative action is taken. We focus specifically on the flexibility that financial entities are afforded as regards the precise timing of their outgoings and for this purpose introduce the \emph{Interval Debt Model (IDM)} in which a set of financial entities is interconnected by debts due within specific time intervals. In the IDM, a \emph{payment schedule} specifies timings of payments to serve the debts. We examine the computational hardness of determining the existence of a schedule of payments with ``good'' properties, e.g., no or few bankruptcies, or minimizing the scale of remedial action. In particular, we establish how hardness depends on variations in the exact formalism of the model (to allow some small number of bankruptcies or insist on none at all) and on restrictions on the structure or lifetime of the input instance. A unique and novel feature of the IDM is its capacity to capture the temporal aspects of real-world financial systems; previous work has seldom explicitly dealt with this intrinsic facet of real-world debt. 

\paragraph*{Financial Networks}
    Graph theory provides models for many problems of practical interest for analyzing (or administering) financial systems. For example, Eisenberg and Noe's work \cite{eisenberg2001systemic} abstracts a financial system to be a weighted digraph (in which each node is additionally labeled according to the corresponding entity's assets). The authors of that work are focused on the existence and computation of a \emph{clearing vector}, which is essentially a set of payments among nodes of the graph which can be executed synchronously without violating some validity constraints. Their model provides the basis for much subsequent work in the network-based analysis of financial systems: it has been adapted to incorporate default costs \cite{rogers2013failure}, Credit Default Swaps \cite{schuldenzucker2017finding} (CDSs) (that is, derivatives through which banks can bet on the default of another bank in the system) and the sequential behavior of bank defaulting in real-world financial networks \cite{papp2020sequential}.     
    
    An axiomatic aspect of Eisenberg and Noe's model is the so-called \emph{principle of proportionality}: that a defaulting bank pays off each of its creditors proportionally to the amount it is owed. 
    Some recent work has considered alternative payment schemes, which allow, for example, paying some debts in full and others not at all (so-called \emph{non-proportional payments}).
    For example, Bertschinger, Hoefer and Schmand \cite{BHS24} study financial networks in a setting where each node is a rational agent which aims to maximize flow through itself by allocating its income to its debts. The focus of that work is on game-theoretic questions, such as the price of anarchy, or the existence, properties, and computability of equilibria. Papp and Wattenhoffer \cite{PW20} also study a non-proportional setting, additionally incorporating CDSs.

    Complementing the decentralized, game-theoretic approach is the question of the (centralized) computability of a \emph{globally} ``good'' outcome through bailout allocation \cite{EW21} (also called cash injection \cite{KKZ22}), or timing default announcements \cite{papp2020sequential}, among other operations. In such works, the prototypical objective is to minimize the number of bankruptcies; related measures include total market value \cite{EW21}, or systemic liquidity \cite{KKZ22}. Egressy and Wattenhoffer \cite{EW21} focus solely on computational complexity of leveraging bailouts to optimize a range of objectives, in a setting which incorporates proportional payments and default costs.     
    Kanellopoulos, Kyropoulou and Zhou \cite{KKZ22,KKZ24} apply both a game-theoretic and a classical complexity perspective to two mechanisms: debt forgiveness (deletion of edges in the financial network) and cash injection (bailouts). Notably, in this work the central authority may remove debts in a way which may be detrimental to certain individuals, but beneficial to the total systemic liquidity.


    
    
    

    Previous research on financial networks has also drawn from ecology \cite{haldane2011systemic}, statistical physics \cite{bardoscia2021physics} and Boolean networks \cite{eisenberg1996boolean}. 

    A central motivation of financial network analysis is to inform central banks' and regulators' policies. The concepts of \emph{solvency} and \emph{liquidity} are core to this task: a bank is said to be \emph{solvent} if it has enough assets (including, e.g., debts owed to it) to meet all its obligations; and it is said to be \emph{liquid} if it has enough liquid assets (that is, cash) to meet its obligations on time. An illiquid but solvent bank may exist even in modern interbank markets \cite{rochet2004coordination}. In such cases, a central bank may act as a \emph{lender of last resort} and extend loans to such banks to prevent them defaulting on debts \cite{bagehot1873lombard, rochet2004coordination}. The optimal allocation of bailouts to a system in order to minimize damage has also been studied as an extension of Eisenberg and Noe's model \cite{papachristou2022allocating}. Here, bailouts refer to funds provided by a third party (such as a government) to entities to help them avoid bankruptcy.

\paragraph*{Temporal Graphs} 
    Temporal graphs are graphs whose underlying connectivity structure changes over time. Such graphs allow us to model real-world networks which have inherent dynamic properties, such as transportation networks \cite{10.1007/s10796-021-10164-2}, contact networks in an epidemic \cite{enright2021deleting, kutner2023tardis} and communication networks; for an overview see~\cite{Holme-Saramaki-book-13,michail2016introduction}. Most commonly, following the formulation introduced by Kempe, Kleinberg and Kumar~\cite{KKK00}, a temporal graph has a fixed set of vertices together with edges that appear and disappear at integer times up to an (integer) lifetime. Often, a natural extension of a problem on static graphs to the temporal setting yields a computationally harder problem; for example, finding node-disjoint paths in a temporal graph remains NP-complete even when the underlying graph is itself a path \cite{klobas2021interference}, and finding a temporal vertex cover remains NP-complete even on star temporal graphs \cite{akrida2020temporal}. 

 \paragraph*{Contributions}
In this paper we present a novel framework, the \emph{Interval Debt Model} (IDM), for considering problems of bailout allocation and payment scheduling in financial networks by using temporal graphs to account for the isochronal aspect of debts between financial entities (previous work has almost exclusively focused on static financial networks). In particular, the IDM offers the flexibility that entities can pay debts earlier or later, within some agreed interval. We introduce several natural problems and problem variants in this model and show that the tractability of such problems depends greatly on the network topology and on the restrictions on payments (i.e., the admission or exclusion of partial and fractional payments on debts). 

Our work explores the natural question of whether and how payments can be scheduled to avert large-scale failures in financial networks. Broadly, we establish that computing a zero-failure schedule (a \emph{perfect schedule}) is NP-complete even when the network topology is highly restricted, unless we admit fractional payments, in which case determining the existence of a perfect schedule is tractable in general. Interestingly, if we allow a small number $k$ of bankruptcies to occur then every problem variant is computationally hard even on inputs with $O(1)$ nodes. 
This can be thought of analogously to \textsc{Max 2SAT} being strictly harder than \textsc{2SAT} (unless P=NP) despite being a ``relaxation'': in \textsc{Max 2SAT} we allow up to $k$ clauses to not be satisfied. 
Furthermore, in the setting where we insist not only on payments being for integer amounts but more strongly that any payment is for the full amount of the corresponding debt, finding a perfect schedule is NP-complete even if there are only four nodes. 

We begin by introducing, first by example and then formally, the Interval Debt Model in Section 2. In Section 3 we present our results: in Section 3.1--3.3 we establish some sufficent criteria for NP-hardness for each of the problems we consider; and in Section 3.4 we present two polynomial-time algorithms. Our conclusions and directions for further research are given in Section 4.
	
	\section{The Interval Debt Model}\label{sec-IDM-model}
	
	In this section, we introduce (first by example and then formally) the \emph{Interval Debt Model}, a framework in which temporal graphs are used to represent the collection of debts in a financial system. 
	
	\subsection{An illustrative example} 
	
	As an example, consider a tiny financial network consisting of the 3 banks $u$, $v$ and $w$ with \Euro30, \Euro20 and \Euro10, respectively, in initial external assets and where there are the following inter-bank financial obligations:
	\begin{itemize}
		\item bank $u$ owes bank $v$ \Euro20 which it must pay by time 3 and \Euro15 which it must pay at time 4 or time 5 (note that all payments must be made at integer times)
		\item bank $v$ must pay bank $w$ a debt of \Euro25 at time 2 exactly
		\item bank $w$  must pay \Euro25 to bank $v$ between times 4 and 6 (that is, at time 4, 5 or 6).
	\end{itemize}
	A graphical representation of this system is shown in Fig.~\ref{fig:intro example} (the descriptive notation used should be obvious and is retained throughout the paper).
	
	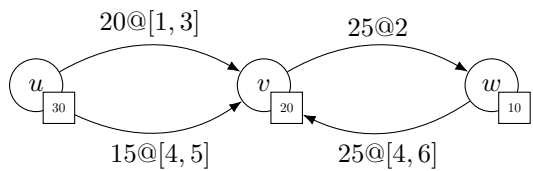
\begin{figure}[!ht]
		\centering
		\begin{tikzpicture}
			\myassetnode{u}{0}{0}{$u$}{$30$}
			\myassetnode{v}{3}{0}{$v$}{$20$}
			\myassetnode{w}{6}{0}{$w$}{$10$}
			\path (u) edge[bend left] node[above] {$20@[1,3]$} (v);
			\path (uasset) edge[bend right] node[below] {$15@[4,5]$} (v);
			\path (v) edge[bend left] node[above] {$25@2$} (w);
			\path (w) edge[bend left] node[below] {$25@[4,6]$} (vasset);
		\end{tikzpicture}
		\caption{A simple instance of the Interval Debt Model (IDM). Numbers in square boxes represent the initial external assets of the node (for example, \Euro30 for node $u$), directed edges represent debts, and the label on an edge represents the terms of the associated debt (for example, $u$ must pay $v$ \Euro20 between time 1 and time 3).}
		\label{fig:intro example}
	\end{figure}
	
	Several points can be made about this system: node $u$ is \emph{insolvent} as its \Euro30 in initial external assets are insufficient to pay all its debts; node $v$ may be \emph{illiquid} for it may default on part of its debt to $w$, e.g., if $u$ pays all of its first debt at time 3, or may remain liquid, e.g., if it receives at least \Euro5 from $u$ by time 2; and node $w$ is solvent and certain to remain liquid in any case. The choices made by the various banks, in the form of a (payment) schedule, clearly affect the status of the overall financial system. Note that solvency is determined solely by whether sufficient funds exist whereas liquidity depends upon when debts are paid and owed.
	
	One may ask several questions about our toy financial system such as: Are partial payments allowed (e.g., $u$ paying \Euro18 of the \Euro20 debt at time 1, and the rest later)? If so, are non-integer payments allowed? Can money received be immediately forwarded (e.g., $u$ paying $v$ \Euro20 at time 2 and $v$ paying $w$ \Euro25 at time 2)? Does $v$ necessarily have to pay its debt to $w$ at time 2 if it has the liquid assets to do so? We now expand upon these questions and specify in detail the setting we consider in the remainder of the paper. Note that throughout the paper we use the euro \Euro~as our monetary unit of resource even though, as we will see, we have a variant of the Interval Debt Model within which payments can be made for any rational fraction of a euro. We often prefix monetary payments with the symbol \Euro~to make our proofs more readable.
	
	\subsection{Formal setting}
	
	Formally, an \emph{Interval Debt Model} (\emph{IDM}) instance is a 3-tuple $(G, D, A^0)$ as follows.
	\begin{itemize}
		\item $G=(V, E)$ is a finite digraph with the set of $n$ nodes (or, alternatively, \emph{banks}) $V=\{v_i:i=1,~2,~\ldots,~n\}$ and the set of $m$ directed labelled edges $E\subseteq V\times V \times \mathbb{N}$, with the edge $(u,v,id)\in E$ denoting that there is an edge, or \emph{debt}, whose label is $id$, from the \emph{debtor} $u$ to the \emph{creditor} $v$. We can have multi-edges but the labels of the edges from some node $u$ to some node $v$ must be distinct and form a contiguous integer sequence $0,1,2,\ldots$. We refer to the subset of edges directed out of or in to some specific node $v$ by $E_{\text{out}}(v)$ and $E_{\text{in}}(v)$, respectively. We also refer to the undirected graph obtained from $G$ by ignoring the orientations on directed edges as the \emph{footprint} of $G$.
		\item $D:E\to\{(a,t_1,t_2):a,t_1,t_2\in\mathbb{N}\setminus\{0\}, t_1\leq t_2\}$ is the \emph{debt function} which associates \emph{terms} to every debt (ordinarily, we abbreviate $D((u,v,id))$ as $D(u,v,id)$). Here, if $e$ is a debt with terms $D(e)=(a,t_1,t_2)$ then $a$ is the \emph{monetary amount} (or \emph{monetary debt}) to be paid and $t_1$ (resp.~$t_2$) is the first (resp.~last) time at which (any portion of) this amount can be paid. For any debt $e\in E$, we also write $D(e) = (D_a(e),D_{t_1}(e), D_{t_2}(e))$. For simplicity of notation, we sometimes denote the terms $D(e)=(a,t_1,t_2)$ by $a@[t_1,t_2]$ or by $a@t_1$ when $t_1=t_2$ (as we did in Fig.~\ref{fig:intro example}); also, for simplicity, we sometimes just refer to $a@[t_1,t_2]$ as the debt.
		\item $A^0=(c_{v_1}^0,c_{v_2}^0,...c_{v_{n}}^0)\in \mathbb N^{n}$ is a tuple with $c_{v_i}^0$ denoting the \emph{initial external assets} (i.e. starting cash) of bank $v_i$.
	\end{itemize}
	
	We refer to the greatest time-stamp $T$ that appears in any debt for a given instance as the \emph{lifetime} and assume that all network activity ceases after time $T$. The instance shown in Fig.~\ref{fig:intro example}, which has lifetime $T=6$, is formally given by: $V=\{u,v,w\}$, $E=\{(u,v,0), (u,v,1), (v,w,0), (w,v,0)\}$, $D(u,v,0)=(20,1,3)$, $D(u,v,1)=(15,4,5)$, $D(v,w,0)=(25,2,2)$, $D(w,v,0)=(25,4,6)$ and $A^0 = (c_u^0, c_v^0, c_w^0)$, where $c_u^0=30$, $c_v^0=20$ and $c_w^0=10$. Similarly, the instance shown in Fig.~\ref{fig:p3} has lifetime $T=2$ and is given by $V=\{u,v,w\}$, $E=\{(u,v,0), (v,w,0)\}$, $D(u,v,0)=(1,1,2)$, $D(v,w,0)=(1,1,1)$ and $A^0 = (c_u^0, c_v^0, c_w^0)$, where $c_u^0=1$, $c_v^0=0$ and $c_w^0=0$.

	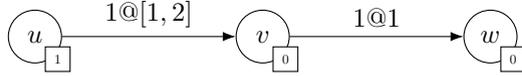
\begin{figure}[!ht]
		\centering
		\begin{tikzpicture}
			\myassetnode{u}{0}{0}{$u$}{$1$}
			\myassetnode{v}{3}{0}{$v$}{$0$}
			\myassetnode{w}{6}{0}{$w$}{$0$}
			\path (u) edge node[above] {$1@[1,2]$} (v);	
			\path (v) edge node[above] {$1@1$} (w);	
		\end{tikzpicture}
		\caption{An IDM instance for which every schedule is described by four payment values $p_{(u,v,0)}^1$, $p_{(v,w,0)}^1$, $p_{(u,v,0)}^2$ and $p_{(v,w,0)}^2$.}
		\label{fig:p3}
	\end{figure}
	
	The \emph{size} of the instance $(G, D, A^0)$ is defined as $n + m + \log(T) + b$, where $b$ is the maximum number of bits needed to encode any of the (integer) numeric values appearing as the monetary amounts in the debts. Note that in what follows, we usually do not mention the label $id$ of a debt $(u,v,id)$ but just refer to the debt as $(u,v)$ when this causes no confusion.
	
	\subsection{Schedules}\label{subsec:schedules}
	
	Given an IDM instance $(G,D,A^0)$, a (\emph{payment}) \emph{schedule} $\sigma$ describes the times at which the banks transfer \emph{assets} to one another via \emph{payments}. Formally, a schedule $\sigma$ is a set of $|E|T$ \emph{payment values} $p_e^t\geq 0$, one for each edge-time pair $(e,t)$ (note that no payments are made at time 0). Equivalently, a schedule can be expressed as an $|E|\times T$ matrix $S$ with the payment values $p_e^t$ the entries of the matrix. The value $p_e^t$ is the monetary amount of the debt $e$ paid at time $t$. Our intention is that at any time $1\leq t \leq T$, every payment value $p_e^t>0$ of a schedule $\sigma$ is paid by the debtor of $e$ to the creditor of $e$, not necessarily for the full monetary amount $D_a(e)$ but for the amount $p_e^t$. A schedule for the instance of Fig.~\ref{fig:p3} consists of the four payments values $p_{(u,v,0)}^1$, $p_{(v,w,0)}^1$, $p_{(u,v,0)}^2$ and $p_{(v,w,0)}^2$. Note that, using the above representation of a schedule $\sigma$, we might have a large number of zero payments. Therefore, for simplicity of presentation, in the remainder of the paper we specify schedules by only detailing the non-zero payments. An example schedule for the IDM instance in Fig. \ref{fig:p3} is then $p_{(u,v,0)}^1=1$, $p_{(v,w,0)}^1=1$.
	
	We now introduce some auxiliary variables which are not strictly necessary but help us to concisely express constraints on and properties of schedules. For nodes $u,v \in V$ and time $0\leq  t \leq T$, the following values are with respect to some specific schedule.
	\begin{itemize}
		\item Denote by $I_v^t$ the total monetary amount of incoming payments of node $v$ at time $t$.
		\item Denote by $O_v^t$ the total monetary amount of outgoing payments (expenses) of node $v$ at time $t$. 
		\item We write $p_{u,v}^t$ to denote the total amount of all payments made from debtor $u$ to creditor $v$ at time $t$ in reference to \emph{all} debts from $u$ to $v$; that is, $p_{u,v}^t=\sum_{i}p_{(u,v,i)}^t$. 
		\item The vector $A^0 = (c_{v_1}^0, c_{v_2}^0, \ldots, c_{v_n}^0)$ specifies the initial external assets (cash) of each node at time $0$. For $t>0$, we denote by $c_v^t$ node $v$'s cash assets at time $t$; that is, $c_v^t=c_v^{t-1}+I_v^t-O_v^{t}$. 
	\end{itemize}
	For clarity, we refer to the starting cash of banks as ``initial external assets'' and to liquid assets in general as cash assets. By cash assets `at time $t$' (resp. `prior to time $t$') we mean after all (resp. before any) of the payments associated with time $t$ have been executed. Cash assets at time $0$ are then precisely the initial external assets at time $0$ (possibly supplemented by some bailout, as we shall see later). 
	
	We have been a little vague so far as regards the form of the payment values in any schedule and have not specified whether these values are integral, rational or do not necessarily equal the full monetary amount of the debt. As we detail below, we have variants of the model covering different circumstances (with perhaps the standard version being when payment values are integral but do not necessarily equal the full monetary amount of the debt).
	
	Recall the example schedule from Fig. \ref{fig:p3}, which we can represent as $p_{u,v}^1=1$, $p_{v,w}^1=1$. As we shall soon see, the payments in this schedule can be legitimately discharged in order to satisfy the terms of all debts but in general this need not be the case. However, there might be schedules that are not \emph{valid}, as well as valid schedules in which banks default on debts (that is, go \emph{bankrupt}). We deal with the key notions of validity and bankruptcy now.
	
	\begin{definition}\label{def:validschedule}
		A schedule is \emph{valid} if it satisfies the following properties (for any debt $e$, terms $D(e)=(a,t_1,t_2)$ and node $v$): 
		\begin{itemize}
			\item all payment values are non-negative; that is, $p_e^t\geq 0$, for $1\leq t\leq T$
			\item all cash asset values (as derived from payment values and initial external assets) are non-negative; that is, $ c_v^t\geq 0$, for $0\leq t\leq T$
			\item no debts are overpaid; that is, $\sum_{t=1}^T p_e^t \leq a$
			\item no debts are paid too early; that is, $\sum_{t=1}^{t_1-1} p_e^t = 0$.
		\end{itemize}
		Given some IDM instance, some schedule and some debt $e$ with terms $D(e) = (a,t_1,t_2)$, the debt $e$ is said to be \emph{payable} at any time in the interval $[t_1,t_2-1]$. At time $t_2$, $e$ is said to be \emph{due}. At time $t_2\leq t\leq T$, if the full amount $a$ has not yet been paid (including payments made at time $t_2$) then $e$ is said to be \emph{overdue} at time $t$. A debt is \emph{active} whenever it is payable, due or overdue.
		However, a bank is said to be \emph{withholding} if, at some time $1\leq t\leq T$, it has an overdue debt and sufficient cash assets to pay (part of, where fractional or partial payments are permitted; see below) the debt. If any bank is withholding (at any time) in the schedule then the schedule is not valid.
	\end{definition}
	So, for example and with reference to the IDM instance in Fig.~\ref{fig:p3}, if, according to some schedule, bank $u$ pays $1$ to bank $v$ at time $1$ but $v$ makes no payment to $w$ at time 1 then $v$ is withholding and the schedule is not valid.
	
	\begin{definition}
		With reference to some schedule, a bank is said to be \emph{bankrupt} (at time $t$) if it is the debtor of an overdue debt (at time $t$). We say that a schedule has $k$ \emph{bankruptcies} if $k$ distinct banks are bankrupt at some time in the schedule (the times at which these banks are bankrupt might vary). A bank may \emph{recover} from bankruptcy if it subsequently receives sufficient income to pay off all its overdue debts. 
	\end{definition}
	
	\begin{definition}
		A bank $v$ is said to be \emph{insolvent} if all its assets (that is, the sum of all debts due to $v$ and of $v$'s initial external assets) are insufficient to cover all its obligations (that is, the sum of all debts $v$ owes). Formally, $v$ is insolvent if 
		\[
		c_v^0 + \sum_{e\in E_{\text{in}}(v)}D_a(e)~~< \sum_{e\in E_{\text{out}}(v)}D_a(e).
		\]
		A bank which is insolvent will necessarily be bankrupt in any schedule.
	\end{definition}
 
	We will not be concerned with the precise timing of bankruptcy or the recovery or not of any bank in this paper.
	
	We now detail three variants of the model (alluded to earlier) in which different natural constraints are imposed on the payment values.
	\begin{definition}
		In what follows, $e$ is an arbitrary debt and $1\leq t\leq T$ some time.
		\begin{itemize}
			\item In the \emph{Fractional Payments} (\emph{FP}) variant, the payment values may take rational values; that is, $p_e^t \in \mathbb{Q}$ and we allow payments for a smaller amount than the full monetary amount of $e$. 
			\item In the \emph{Partial Payments} (\emph{PP}) variant, the payment values may take only integer values; that is, $p_e^t \in \mathbb{N}$ and we allow payments for a smaller amount than the full monetary amount of $e$.
			\item In the \emph{All-or-Nothing} (\emph{AoN}) variant, every payment value must fully cover the relevant monetary amount of $e$; that is, every payment value must be for the full monetary amount of $e$ or zero. So, $p_e^t \in \{D_a(e), 0\}$.
		\end{itemize}
	\end{definition}
	For example, the instance of Fig.~\ref{fig:p3} has the following valid schedules:
	\begin{itemize}
		\item (in all variants) the schedule above in which $p_{u,v}^1=p_{v,w}^1=\Euro 1$ (all debts are paid in full at time 1)
		\item (in all variants) the schedule in which $p_{u,v}^2=p_{v,w}^2=\Euro 1$ (all debts are paid in full at time 2)
		\begin{itemize}
			\item under this schedule, node $v$ is bankrupt at time 1 as \Euro1 of the debt $(v,w,0)$ is unpaid and that debt is overdue
		\end{itemize} 
		\item (in the FP variant only) for every $a\in \mathbb Q$, where $0<a<1$, the schedule in which $p_{u,v}^1=p_{v,w}^1=\Euro a$ and $p_{u,v}^2=p_{v,w}^2=\Euro 1-a$
		\begin{itemize}
			\item under each of these schedules, node $v$ is bankrupt at time 1 as \Euro$1-a$ of the debt $(v,w,0)$ is unpaid and that debt is overdue.
		\end{itemize} 
	\end{itemize}
	
	It is worthwhile clarifying the concepts of instant forwarding and payment-cycles. We emphasize that we allow a bank to instantly spend income received. Note that in any valid schedule for the instance in Fig.~\ref{fig:p3}, $v$ \emph{instantly forwards} money received from $u$ to $w$ (so as not to be withholding); so, the cash assets of $v$ never exceed 0 in any valid schedule. This behaviour is consistent with the Eisenberg and Noe model \cite{eisenberg2001systemic} in which financial entities operate under a single clearing authority which synchronously executes payments. Indeed, in such cases a payment-chain of any length is permitted and the payment takes place instantaneously regardless of chain length.
	
	Furthermore, and still consistent with the Eisenberg and Noe model, there is the possibility of a \emph{payment-cycle} which is a set of banks $\{u_1,u_2,\ldots,u_c\}$, for some $c\geq 2$, with a set of debts $\{e_i=(u_i,u_{i+1},l_i): 1\leq i\leq c-1\}\cup\{e_c = (u_{c},u_0,l_{c})\}$ so that at some time $t$, all debts are active yet none has been fully paid and where each bank makes a payment, at time $t$, of the same value $a$ towards its debt. As an illustration, Fig. \ref{fig:cycle} shows three `cyclic' IDM instances, all with lifetime $T=2$. By our definition of a valid schedule, the schedule $p_{u,v}^1=p_{v,w}^1=p_{w,x}^1=p_{x,u}^1=\Euro 1$, forming a payment-cycle, is valid in all three instances. This is intuitive for Fig.~\ref{fig:cycleloaded}, where each node has sufficient initial external assets available to pay all its debts in full at any time, irrespective of income. In Fig. \ref{fig:cycleroulette}, we may imagine that the \Euro1 moves from node $u$ along the cycle, satisfying every debt at time 1. This is a useful abstraction but not strictly accurate: rather, we should imagine that all four banks simultaneously order payments forward under a single clearing system. The clearing system calculates the balances that each bank would have with those payments executed, ensures they are all non-negative (one of our criteria for schedule validity) and then executes the payments by updating all accounts simultaneously. This distinction is significant when we consider Fig.~\ref{fig:cycleempty} in which no node has any initial external assets. A clearing system ordered to simultaneously pay all debts would have no problem doing so in the Eisenberg and Noe model and in our model this constitutes a valid schedule. We highlight that there also exist valid schedules for the instance in Fig.~\ref{fig:cycleempty} in which all four banks go bankrupt, one schedule being where all payments at any time are 0: here, no bank is withholding (they all have zero cash assets), so the schedule is valid, but every bank has an overdue debt and so is bankrupt. 

	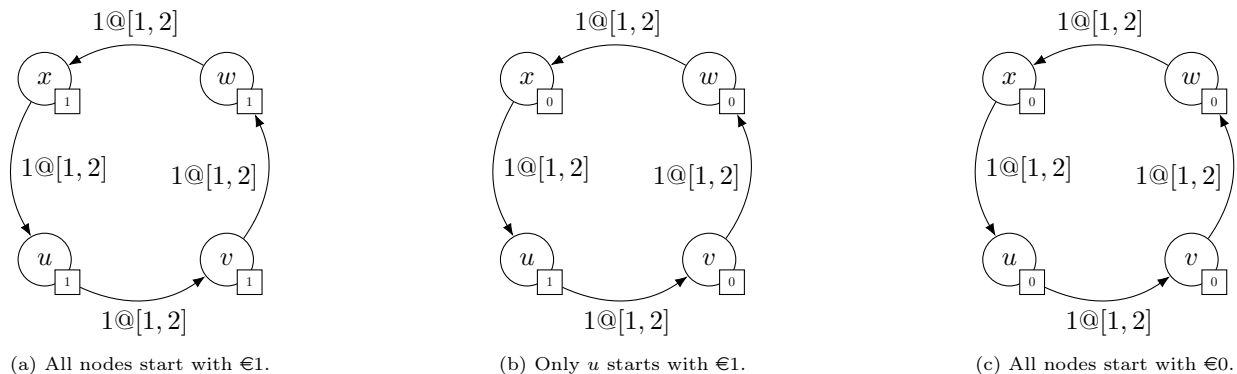
\begin{figure*}[!ht]
		\centering
		\subfloat[All nodes start with \Euro1.\label{fig:cycleloaded}]{
			\begin{tikzpicture}
				\myassetnode{u}{0}{0}{$u$}{$1$}
				\myassetnode{v}{2.4}{0}{$v$}{$1$}
				\myassetnode{w}{2.4}{2.4}{$w$}{$1$}
				\myassetnode{x}{0}{2.4}{$x$}{$1$}
				\path (uasset) edge[bend right] node[below] {$1@[1,2]$} (v);	
				\path (v) edge[bend right] node[left] {$1@[1,2]$} (wasset);	
				\path (w) edge[bend right] node[above] {$1@[1,2]$} (x);	
				\path (x) edge[bend right] node[right] {$1@[1,2]$} (u);	
			\end{tikzpicture}
		}
		\hfill
		\subfloat[Only $u$ starts with \Euro1.\label{fig:cycleroulette}]{
			\begin{tikzpicture}
				\myassetnode{u}{0}{0}{$u$}{$1$}
				\myassetnode{v}{2.4}{0}{$v$}{$0$}
				\myassetnode{w}{2.4}{2.4}{$w$}{$0$}
				\myassetnode{x}{0}{2.4}{$x$}{$0$}
				\path (uasset) edge[bend right] node[below] {$1@[1,2]$} (v);	
				\path (v) edge[bend right] node[left] {$1@[1,2]$} (wasset);	
				\path (w) edge[bend right] node[above] {$1@[1,2]$} (x);	
				\path (x) edge[bend right] node[right] {$1@[1,2]$} (u);	
			\end{tikzpicture}
		}
		\hfill
		\subfloat[All nodes start with \Euro0.\label{fig:cycleempty}]{
			\begin{tikzpicture}
				\myassetnode{u}{0}{0}{$u$}{$0$}
				\myassetnode{v}{2.4}{0}{$v$}{$0$}
				\myassetnode{w}{2.4}{2.4}{$w$}{$0$}
				\myassetnode{x}{0}{2.4}{$x$}{$0$}
				\path (uasset) edge[bend right] node[below] {$1@[1,2]$} (v);	
				\path (v) edge[bend right] node[left] {$1@[1,2]$} (wasset);	
				\path (w) edge[bend right] node[above] {$1@[1,2]$} (x);	
				\path (x) edge[bend right] node[right] {$1@[1,2]$} (u);	
			\end{tikzpicture}
		}
		\caption{Examples illustrating the behaviour of cycles in the IDM. In all instances shown the schedule in which all nodes pay their debts in full at time 1 is valid.}
		\label{fig:cycle}
		
	\end{figure*}

	We use payment-cycles throughout our constructions in a context such as that in Fig.~\ref{fig:paymentcycle}. Here, a valid schedule is where all nodes pay their corresponding debts in full at time $t$. The effect is that the \Euro1 of cash assets at node $u$ is `transferred' to \Euro1 of cash assets at node $v$.

	\begin{figure*}[!ht]
	\centering
	\scalebox{1}{
		\begin{tikzpicture}
			\myassetnode{u}{0}{2}{$u$}{$1$}
			\myassetnode{w}{2}{2}{$w$}{$0$}
			\myassetnode{x}{4}{2}{$x$}{$0$}
			\myassetnode{v}{6}{2}{$v$}{$0$}
			\myassetnode{y}{3}{0}{$y$}{$0$}
			\path (u) edge node[above] {$1@t$} (w);	
			\path (x) edge node[above] {$1@t$} (w);	
			\path (x) edge node[above] {$1@t$} (v);	
			\path (wasset) edge node[left] {$2@t$} (y);	
			\path (y) edge node[right] {$2@t$} (x);	
		\end{tikzpicture}}

	\caption{Using a payment-cycle to effectively transfer \Euro1 of assets from node $u$ to node $v$.}
	\label{fig:paymentcycle}

	\end{figure*}
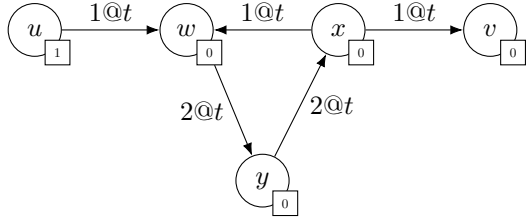

	\subsection{Canonical instances}
	
	We wish to replace certain IDM instances with equivalent yet simpler ones. For example, consider the instance given in Fig.~\ref{fig:intro example} but where every time-stamp in the instance is multiplied by a factor of $100$ (so that, for example, the debt from $w$ to $v$ becomes $25@[400,600]$). This `inflated' instance is in essence equivalent to the original one but has a lifetime of $600$.
	
	\begin{definition}
		Let $(G, D, A^0)$ be an instance. Then the set of time-stamps $\{t: D_{t_1}(e)=t \text{ or } D_{t_2}(e)=t \text{, for some edge } e\}$ is the set of \emph{extremal} time-stamps.
	\end{definition}
	
	There is a simple preprocessing step such that we can assume that the lifetime $T$ of any IDM instance is polynomially bounded in $n$ and $m$ (that is, the numbers of banks and debts, respectively). This preprocessing step modifies the instance such that every $1\leq t \leq T$ is an extremal time-stamp with the process being to simply omit non-extremal time-stamps and then compact the remaining time-stamps. Observe that this procedure is such that any valid schedule in the original IDM instance can be transformed into a valid schedule in the compacted instance so that these schedules have the same number of bankruptcies, eventual assets and so forth, and vice versa when the compacted instance is expanded into the original instance. Hence, we need not consider pathological cases in which the lifetime is, say, exponential in the number of nodes and debts. Given this restriction, we can now revise the notion of the size of an IDM instance to say that it is $n+m+b$ where $n$ is the number of banks, $m$ is the number of debts and $b$ is the maximum number of bits needed to encode any of the numeric values appearing as monetary amounts of debts.
	
	\begin{lem}\label{lem:polyverif}
		For any given IDM instance and any schedule, in any of the FP, PP or AoN variants, it is possible in polynomial-time both to check whether the schedule is valid and to compute the number of bankruptcies under the schedule.
	\end{lem}
	
	\begin{proof}[Proof sketch] It is possible to iterate over the schedule once and calculate: the cash assets of every node, and which debts are overdue at each time-stamp. 
    Computing the set $\{v | v $ has some overdue debt under $\sigma\}$ is then straightforward, and the number of bankruptcies is the cardinality of that set.

    It remains to check the validity of the schedule. We can efficiently verify that there are:
    \begin{description}
        \item[No withholding banks:] iterate once over the debts overdue at each time. If the debt $e=(u,v,i)$ is overdue at time $t$, verify that $c_u^t$ is insufficient to make a payment toward $e$ (i.e. $c_u^t= 0$ in the FP or PP model, or $c_u^t<D_a(e)$ in the AoN model).
        \item[No overpaid debts:] iterate over all debts and ensure payments made with reference to each are no more than the debt amount.
        \item[No debts paid too early:] ensure $p_e^t = 0$ for any $t < D_{t_1}(e)$, for each debt $e$.
    \end{description}
	\end{proof}

	\subsection{Problem definitions} \label{subsec:problems}
	
	We now define some decision problems with natural real-world applications.

    \begin{framed}
        \noindent\underline{\textsc{Bankruptcy Minimization}}
    	\begin{description}
    		\item Instance: an IDM instance $(G,D,A^0)$ and an integer $k$
    		\item Yes-instance: an instance for which there exists a valid schedule $\sigma$ such that at most $k$ banks go bankrupt at some time in the schedule $\sigma$.
    	\end{description}
    \end{framed}
     
    \begin{framed}
    	\noindent\underline{\textsc{Perfect Scheduling}}
    	\begin{description}
    		\item Instance: an IDM instance $(G,D,A^0)$
    		\item Yes-instance: an instance for which there exists a valid schedule $\sigma$ such that no debt is ever overdue in $\sigma$; that is, a \emph{perfect schedule}.
    	\end{description}
    \end{framed}	
    
    \begin{framed}
         
    	\noindent\underline{\textsc{Bailout Minimization}}
    	\begin{description}
    		\item Instance: an IDM instance $(G,D,A^0)$ (with $n$ banks) and an integer $b$
    		\item Yes-instance: an instance for which there exists a positive bailout vector $B=(b_{1}, b_{2}, \ldots,  b_{n})$ with $\sum_{i=1}^n b_{i} \leq b$ and valid schedule $\sigma$ such that $\sigma$ is a perfect schedule for the instance $(G,D,A^0+B)$.
    	\end{description}
    \end{framed}
 
	The problem \textsc{Perfect Scheduling} is equivalent to the \textsc{Bailout Minimization} problem where $b=0$ and to the \textsc{Bankruptcy Minimization} problem where $k=0$. 
 
    \begin{framed}
    	\noindent\underline{\textsc{Bankruptcy Maximization}}
    	\begin{description}
    		\item Instance: an IDM instance $(G,D,A^0)$ and an integer $k$
    		\item Yes-instance: an instance for which there exists a valid schedule $\sigma$ such that at least $k$ banks go bankrupt at some time in the schedule $\sigma$. 
    	\end{description}
     \end{framed}
 
	The problem \textsc{Bankruptcy Maximization} is interesting to consider for quantifying a `worst-case' schedule where banks' behavior is unconstrained beyond the terms of their debts.
	
	All of the problems above exist in the AoN, PP and FP variants and are in NP: for every yes-instance, there exists a witness schedule, polynomial in the size of the input, the validity of which can be verified in polynomial time (see Lemma \ref{lem:polyverif}). 
	
	Every valid PP schedule is a valid FP schedule whereas not every valid AoN schedule is a valid PP schedule. In an AoN schedule, a bank may go bankrupt while still having assets (insufficient to pay off any of its debts) whereas this is prohibited in any PP schedule as that bank would be withholding. If we restrict the instances to only those in which for every debt $e$, $D_a(e)=1$ then every valid AoN schedule for that instance is a valid PP schedule and a valid FP schedule.
	
	We call a digraph $G$ from some IDM instance a \emph{multiditree} whenever the footprint of $G$ is a tree. We call a multiditree in which every edge is directed away from the root a \emph{rooted out-tree} (or just \emph{out-tree}). By an \emph{out-path} we mean an out-tree where the footprint is a path and the root is either of the endpoints. We take this opportunity to note that an out-tree is both a directed acyclic graph (DAG) and a multiditree, but that not every multiditree DAG is an out-tree. 
	
	A summary of some of our upcoming results is given in Table~\ref{table:results} (with NP-c denoting `NP-complete' and P denoting `polynomial-time'). 
    However, note that there are other, more nuanced results in what follows that do not feature in Table~\ref{table:results}. Also, even though hardness for out-trees entails hardness for multiditrees and DAGs, we reference a separate result for the latter two settings where that is proven under different (stronger) constraints for the more general graph class. For example, our proof that \textsc{AoN Bankruptcy Minimization} is NP-complete on out-trees uses a construction requiring $T \ge 2$, but our proof of Theorem \ref{thm:bankmin} has $T=1$. 

	\begin{table}[!ht]
		\centering
		\setlength{\tabcolsep}{1pt}
		\begin{tabular}{
				| m{4.7cm}
				| >{\centering\arraybackslash}m{1.6cm} 
				| >{\centering\arraybackslash}m{1.6cm} 
				| >{\centering\arraybackslash}m{1.6cm} 
				| >{\centering\arraybackslash}m{1.8cm} |}
			\hline
			\footnotesize problem 
			& \footnotesize out-tree 
			& \footnotesize multiditree 
			& \footnotesize DAG  
			& \footnotesize general case \\
			\hline
			\footnotesize\textsc{FP Bankruptcy Minimization}  
			& \footnotesize ?                       
			& \footnotesize ?        
			& \footnotesize\hspace{.5em}NP-c \newline (Thm~\ref{thm:bankmin})      
			& \footnotesize\hspace{.5em}NP-c \newline (Thm~\ref{thm:bankmin}) \\ 
			\hline
			\footnotesize\textsc{PP Bankruptcy Minimization}                                 
			& \footnotesize ?
			& \footnotesize\hspace{.5em}NP-c \newline (Thm~\ref{thm:perfsched multiditree})
			& \footnotesize\hspace{.5em}NP-c \newline (Thm~\ref{thm:bankmin})
			& \footnotesize\hspace{.5em}NP-c \newline (Thm~\ref{thm:bankmin}) \\ 
			\hline
			\footnotesize\textsc{AoN Bankruptcy Minimization}
			& \footnotesize\hspace{.5em}NP-c \newline (Thm~\ref{thm:AoN3path})
			& \footnotesize\hspace{.5em}NP-c \newline (Thm~\ref{thm:AoN3path})
			& \footnotesize\hspace{.5em}NP-c \newline (Thm~\ref{thm:bankmin})
			& \footnotesize\hspace{.5em}NP-c \newline (Thm~\ref{thm:bankmin}) \\ 
			\hline
			\footnotesize\textsc{FP Perfect Scheduling}
			& \footnotesize \hspace{1em}P \newline (Thm~\ref{thm:ptime fractional})
			& \footnotesize \hspace{1em}P \newline (Thm~\ref{thm:ptime fractional})
			& \footnotesize \hspace{1em}P \newline (Thm~\ref{thm:ptime fractional})          
			& \footnotesize \hspace{1em}P \newline (Thm~\ref{thm:ptime fractional}) \\ 
			\hline
			\footnotesize\textsc{PP Perfect Scheduling} 
			& \footnotesize\hspace{1em}P \newline (Thm~\ref{thm:ptime out-tree})                       
			& \footnotesize\hspace{.5em}NP-c \newline (Thm~\ref{thm:perfsched multiditree})
			& \footnotesize\hspace{.5em}NP-c \newline (Thm~\ref{thm:perfschedDAG})
			& \footnotesize\hspace{.5em}NP-c \newline (Thm~\ref{thm:perfschedDAG}) \\ 
			\hline
			\footnotesize\textsc{AoN Perfect Scheduling}                                 
			& \footnotesize\hspace{.5em}NP-c \newline (Thm~\ref{thm:AoN3path})
			& \footnotesize\hspace{.5em}NP-c \newline (Thm~\ref{thm:AoN3path})
			& \footnotesize\hspace{.5em}NP-c \newline (Thm~\ref{thm:perfschedDAG}) 
			& \footnotesize\hspace{.5em}NP-c \newline (Thm~\ref{thm:perfschedDAG})  \\
			\hline
			\footnotesize\textsc{FP Bailout Minimization}
			& \footnotesize \hspace{1em}P \newline (Thm~\ref{thm:ptime fractional})
			& \footnotesize \hspace{1em}P \newline (Thm~\ref{thm:ptime fractional})
			& \footnotesize \hspace{1em}P \newline (Thm~\ref{thm:ptime fractional})          
			& \footnotesize \hspace{1em}P \newline (Thm~\ref{thm:ptime fractional}) \\ 
			\hline
			\footnotesize\textsc{PP Bailout Minimization} 
			& \footnotesize\hspace{1em}P \newline (Thm~\ref{thm:ptime out-tree})                       
			& \footnotesize\hspace{.5em}NP-c \newline (Thm~\ref{thm:perfsched multiditree})
			& \footnotesize\hspace{.5em}NP-c \newline (Thm~\ref{thm:perfschedDAG})
			& \footnotesize\hspace{.5em}NP-c \newline (Thm~\ref{thm:perfschedDAG}) \\ 
			\hline
			\footnotesize\textsc{AoN Bailout Minimization}                                 
			& \footnotesize\hspace{.5em}NP-c \newline (Thm~\ref{thm:AoN3path})
			& \footnotesize\hspace{.5em}NP-c \newline (Thm~\ref{thm:AoN3path})
			& \footnotesize\hspace{.5em}NP-c \newline (Thm~\ref{thm:perfschedDAG}) 
			& \footnotesize\hspace{.5em}NP-c \newline (Thm~\ref{thm:perfschedDAG})  \\
			\hline
			\footnotesize\textsc{FP Bankruptcy Maximization}                                
			& \footnotesize ?
			& \footnotesize ?
			& \footnotesize\hspace{.5em}NP-c \newline (Thm \ref{thm:bankmax})   
			& \footnotesize\hspace{.5em}NP-c \newline (Thm \ref{thm:bankmax}) \\ 
			\hline
			\footnotesize\textsc{PP Bankruptcy Maximization}                                
			& \footnotesize ?
			& \footnotesize ?
			& \footnotesize\hspace{.5em}NP-c \newline (Thm \ref{thm:bankmax})   
			& \footnotesize\hspace{.5em}NP-c \newline (Thm \ref{thm:bankmax}) \\ 
			\hline
			\footnotesize\textsc{AoN Bankruptcy Maximization}
			& \footnotesize\hspace{.5em}NP-c \newline (Thm~\ref{thm:AoN3pathBankMax})
			& \footnotesize\hspace{.5em}NP-c \newline (Thm~\ref{thm:AoN3pathBankMax})
			& \footnotesize\hspace{.5em}NP-c \newline (Thm~\ref{thm:AoN3pathBankMax}) 
			& \footnotesize\hspace{.5em}NP-c \newline (Thm~\ref{thm:AoN3pathBankMax})  \\
			\hline
		\end{tabular}\medskip
		\caption{Summary of results.}\label{table:results}
	\end{table}

    \subsection{Discussion of the model}
    We describe here some notable differences (and similarities) of the IDM as compared with other studied models. 
    
    First and foremost, the IDM is a temporal model; the eponymous ``interval debts'' are its principal distinguishing feature when contrasted with other financial network models.
    The \emph{timing} of payments, not their allocation to one payee or another, is the principal question. In fact, in \textsc{Perfect Scheduling} this is the \emph{only} question. As we shall see, under the restriction $D_{t_1}=D_{t_2}$ that problem (and its superproblem \textsc{Bailout Minimization}) become straightforwardly solvable in all variants.  All of our hardness results arise from the expressivity of that degree of freedom (the scheduling of payments sooner or later). Indeed, in \textsc{Bankruptcy Minimization} and \textsc{Bankruptcy Maximization} that freedom remains, and the problems remain NP-complete under the same restriction $D_{t_1}=D_{t_2}$. Consequently, the results of other works which do not have a temporal component \cite{KKZ24,EW21,PW20} do not straightforwardly carry over to the IDM. 
    
    We take this opportunity to emphasize that interval debts are practically motivated; in particular, some real-world debts may be paid neither early nor late (see, e.g., Figure \ref{fig:bond}). 

    
    \begin{wrapfigure}{l}{0.5\textwidth}        
        \begin{center}
        \includegraphics[width=0.9\linewidth]{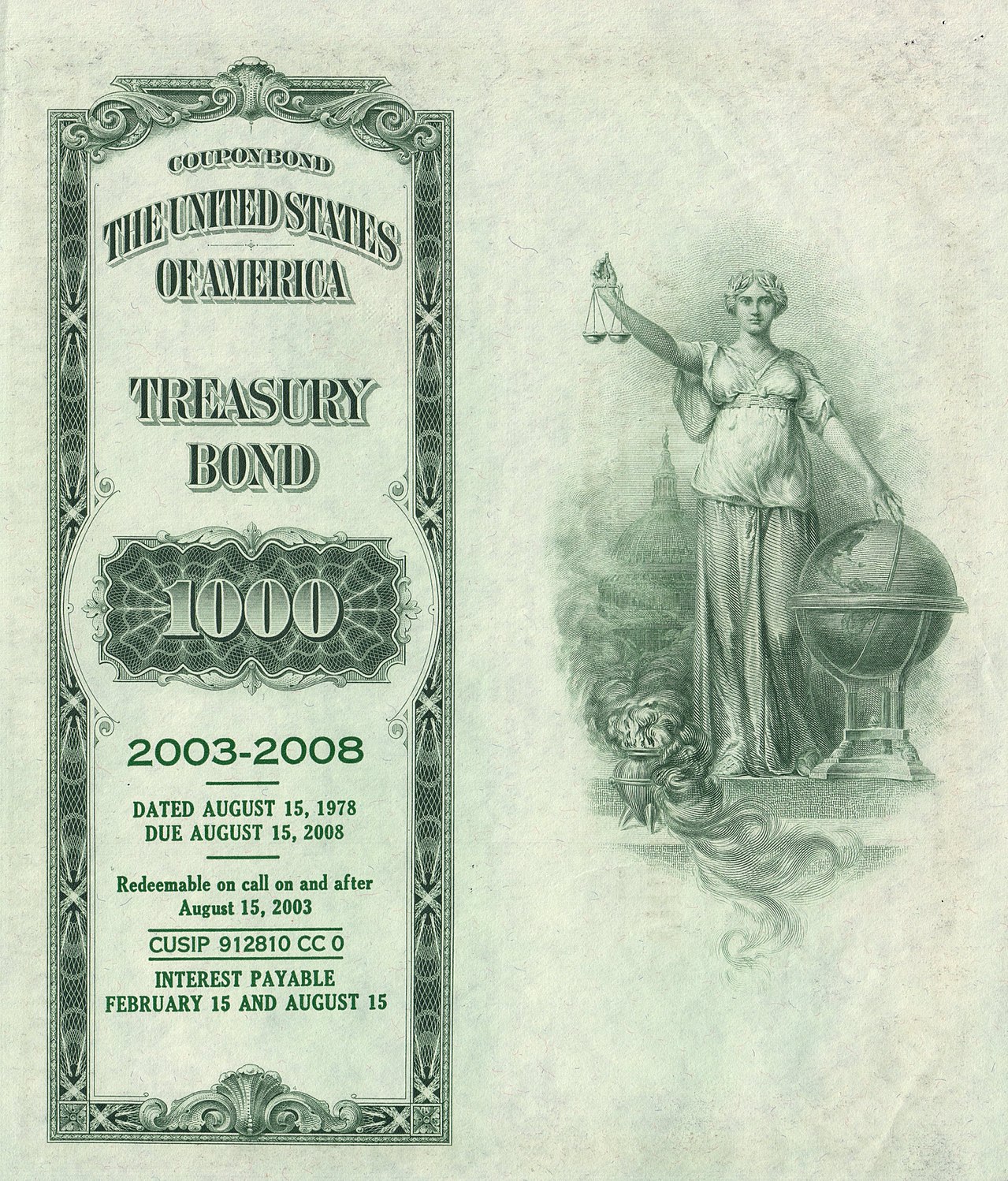}
        \end{center}
        \caption{A real-life interval debt: this 1978 US government bond is payable between 2003 and 2008.} \label{fig:bond}
    \end{wrapfigure}

    Non-proportional payments on debts are likewise nothing new to financial networks. Recent work has considered frameworks wherein \emph{priorities} are associated with each debt, with higher-priority debts paid off before lower-priority debts \cite{KKZ24,PW20}. 
    In such a setting, the priority of some debt may be either chosen by a regulatory authority or left to the individual agents. In the former case, the hardness of computing a solution which maximizes utility is of particular interest, whereas game-theoretic approaches are more relevant in the latter. Our focus in the present work is solely on questions of computational complexity from a centralized perspective. 
    
    We note that \textsc{Bailout Minimization} and its subproblem \textsc{Perfect Scheduling} remain unchanged as decision problems if bankruptcy in the IDM is redefined to require proportional payments, or immediate deletion of the bankrupt node. 
    Both problems fundamentally ask whether a perfect schedule exists; consequently, the manner in which bankruptcy and defaulting are modeled in the IDM are irrelevant. If there is a perfect schedule $\sigma$, then under $\sigma$ all debts are by definition paid on time, in full (and hence proportionally). Conversely, if no such $\sigma$ exists, then a bankruptcy (however it is modeled) must occur, and we have a no-instance of the respective problems.
    
    Lastly, we would like to comment briefly on the respective practical value of the AoN, PP and FP variants. The FP variant is quite intuitive for theoreticians, and yields our main tractable case. 
    On the other hand, the PP variant realizes the practical constraint that arbitrarily small transfers are impractical, and may be of interest where a fungible but indivisible resource needs to be exchanged.
    Personal communication \cite{LIPNE24} suggests that, perhaps unexpectedly, the AoN variant may well be the one of most interest to the finance community.
    Unlike most other models, the AoN model has the unintuitive property that a bankrupt bank may retain some assets. We note that this modelling of bankruptcy is not required for any of our hardness of tractability proofs in the AoN model. 
	
	\section{Our results}\label{section:results} 
	
	In this section we investigate the complexity of the problems presented above. We present our hardness results for \textsc{Bankruptcy Minimization}, \textsc{Perfect Scheduling} and \textsc{Bankruptcy Maximization} in Sections~\ref{subsec:hardnessBmin}, ~\ref{subsec:hardnessPS} and ~\ref{subsec:hardnessBmax}, respectively, and then in Section~\ref{subsec:poly} show that under certain constraints the problem \textsc{Bailout Minimization} and its subproblem \textsc{Perfect Scheduling} become tractable. 
	
	\subsection{Hardness results for \textsc{Bankruptcy Minimization}}\label{subsec:hardnessBmin}
	
	We begin with our core hardness result.
	
	\begin{thm}
		\label{thm:bankmin}
		For each of the AoN, PP and FP variants, the problem \textsc{Bankruptcy Minimization} is NP-complete, even when we restrict to IDM instances $(G,D,A^0)$ for which: $T=1$; $G$ is a directed acyclic graph with a longest path of length 4, has out-degree at most 2 and has in-degree at most 3; the monetary amount of any debt is at most \Euro3; and initial external assets are at most \Euro3 per bank.
	\end{thm}
	
	\begin{proof}
		We build a polynomial-time reduction from the problem \textsc{3-Sat-3} to \textsc{Bankruptcy Minimization} so that all target instances satisfy the constraints in the statement of the theorem (the problem \textsc{3-Sat-3}, defined below, was shown to be NP-complete in \cite{Tov84}).\smallskip
		
		\noindent\underline{\textsc{3-Sat-3}}
		\begin{description}
		\item Instance: a c.n.f. formula $\phi$ over $n$ Boolean variables $v_1, v_2, ... , v_n$ so that each of the $m$ clauses $c_1, c_2, ... , c_m$ has size at most 3 and where there are exactly 3 occurrences of $v_i$ or $\neg v_i$ in the clauses
		\item Yes-instance: there exists a satisfying truth assignment for $\phi$.
		\end{description}
		We may (and do, throughout the paper) restrict ourselves to those instances in which every literal appears at least once and at most twice (that is, both $v_i$ and $\neg v_i$ appear in some clause, for $1\leq i\leq n$) and where no clause contains both a literal and its negation. We define the size of an instance $\phi$ to be $n$.
		
		Suppose that we are given a \textsc{3-Sat-3} instance $\phi$ of size $n$. We construct an IDM instance $(G,D,A^0)$ as follows. For any variable $v_i$, denote by $count_{v_i}$ (resp. $count_{\lnot v_i}$) the number of occurrences of the literal $v_i$ (resp. $\lnot v_i$) in $\phi$ (of course, $count_{v_i}+count_{\lnot v_i} = 3$). We build a digraph $G$ with:
		\begin{itemize}
			\item a \emph{source node} $s_i$, for each variable $v_i$, so that this node has initial external assets \Euro3 (every other type of node will have initial external assets \Euro0)
			\item  two \emph{literal nodes} $x_i$ and $\lnot x_i$, for each variable $v_i$
			\item a \emph{clause node} $q_j$, for each clause $c_j$
			\item a \emph{sink node} $d$.
		\end{itemize}
		We then add edges and debts as follows. For every $1 \leq i \leq n$:
		\begin{itemize}
			\item we add the debt $(s_i,x_i)$ with terms $3@1$
			\item we add the debt $(s_i,\lnot x_i)$ with terms $3@1$
			\item we add the debt $(x_i, d)$ with terms $count_{\lnot v_i}@1$	(note that the monetary amount to be paid is either \Euro1 or \Euro2)
			\item we add the debt $(\lnot x_i, d)$ with terms $count_{v_i}@1$	(note that the monetary amount to be paid is either \Euro1 or \Euro2)
			\item for every $1 \leq j \leq m$:
			\begin{itemize}
				\item we add the debt $(q_j, d)$ with terms $1@1$
				\item if the literal $v_i \in c_j$ then we add the debt $(x_i, q_j)$ with terms $1@1$
				\item if the literal $\lnot v_i \in c_j$ then we add the debt $(\lnot x_i, q_j)$ with terms $1@1$.
			\end{itemize}
		\end{itemize}
		Fig.~\ref{fig:bankruptcy minimization} shows a sketch of this construction (where nodes without any depicted initial external assets start with \Euro0). The IDM instance $(G,D,A^0)$ can clearly be built from $\phi$ in polynomial-time.
		
		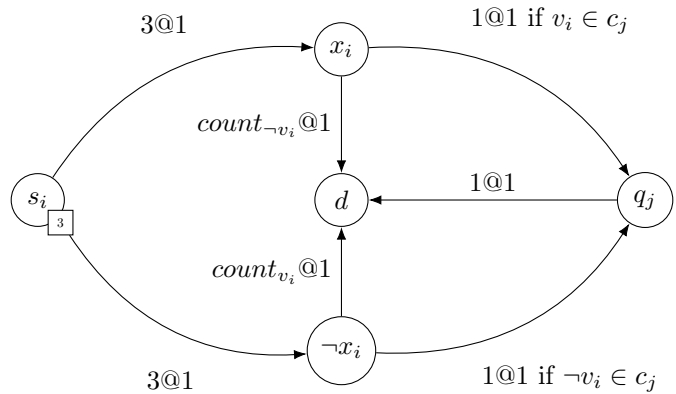
\begin{figure}[!ht]
			\centering
			\begin{tikzpicture}
				
				\myassetnode{si}{0}{0}{$s_i$}{$3$}
				\mynode{xi}{4}{2}{$x_i$}
				\mynode{nxi}{4}{-2}{$\lnot x_i$}
				\mynode{qj}{8}{0}{$q_j$}
				\mynode{d}{4}{0}{$d$}

				\path (si) edge[bend left] node[above, yshift=12pt] {$3@1$} (xi);	
				\path (siasset) edge[bend right] node[below, yshift=-12pt] {$3@1$} (nxi);	
				\path (xi) edge node[left] {$count_{\lnot v_i}@1$} (d);	
				\path (nxi) edge node[left] {$count_{v_i}@1$} (d);	
				\path (xi) edge[bend left]  node[above, yshift=12pt, xshift=12pt] {$1@1$ if $v_i \in c_j$} (qj);	
				\path (nxi) edge[bend right] node[below, yshift=-12pt, xshift=18pt] {$1@1$ if $\lnot v_i \in c_j$} (qj);	
				\path (qj) edge node[above] {$1@1$} (d);	
				
			\end{tikzpicture}
			\caption{Construction sketch for an IDM instance from a given formula $\phi$. Note that each of $x_i$ and $\lnot x_i$ owes \Euro 3 in total.}
			\label{fig:bankruptcy minimization}
		\end{figure}
		
		We claim that the instance $((G, D, A^0), 2n)$ of \textsc{Bankruptcy Minimization} as constructed above is a yes-instance of \textsc{Bankruptcy Minimization} (no matter which of the AoN, PP and FP variants we work with) iff $\phi$ is a yes-instance of \textsc{3-Sat-3}.
		
		Before we proceed, we have the following remark. Recall that for each $1\leq i \leq n$, $count_{v_i} + count_{\lnot v_i} = 3$; consequently, $count_{v_i}=2$ iff $count_{\lnot v_i}=1$, and vice versa. For each $1\leq i \leq n$, node $x_i$ (resp.~$\lnot x_i$) has a total monetary debt to the clause nodes of $count_{v_i}$ (resp.~$count_{\lnot v_i}$) and a total monetary debt to the sink node $d$ of $count_{\lnot v_i}$ (resp.~$count_{v_i}$); so, each literal node has a total monetary debt of \Euro3.
		
		\begin{clm}\label{clm:bankminclm1}
			If $\phi$ is a yes-instance of \textsc{3-Sat-3} then $((G,D,A^0),2n)$ is a yes-instance of \textsc{Bankruptcy Minimization}.
		\end{clm}
	
	\begin{proof}		
		Suppose that $\phi$ is satisfiable via some truth assignment $X$. Consider the schedule $\sigma$ for $(G, D, A^0)$ in which:
		\begin{itemize}
			\item every source node $s_i$ pays \Euro3 (at time 1, as are all payments) to the literal node $x_i$ (resp.~$\lnot x_i$) if $X(v_i) = True$ (resp.~$X(v_i) = False$) 
			\item every literal node $x_i$ (resp. $\lnot x_i$) for which $X(v_i) = True$ (resp. $X(v_i) = False$) pays all its debts in full
			\begin{itemize}
				\item as remarked above, this literal node has total monetary debt \Euro3 but, from above, it receives \Euro3 from $s_i$
			\end{itemize}
			\item every clause node pays its \Euro1 debt to the sink node $d$
			\begin{itemize}
				\item this is necessarily possible because $X$ is a satisfying truth assignment, meaning every clause node receives at least \Euro1 from some literal node corresponding to a literal in that clause set to $True$ by $X$.
			\end{itemize}
		\end{itemize}
		Note that $\sigma$ is valid in all three IDM variants. The total number of bankruptcies in $\sigma$ is $2n$: exactly $n$ bankrupt source nodes and exactly $n$ bankrupt literal nodes. Hence, if $\phi$ is satisfiable then the schedule $\sigma$ for $(G, D, A^0)$ results in at most (in fact, exactly) $2n$ bankruptcies.
		\end{proof}
	
	\begin{clm}\label{clm:bankminclm2}
		If $(G,D,A^0),2n)$ is a yes-instance of \textsc{Bankruptcy Minimization} then $\phi$ is a yes-instance of \textsc{3-Sat-3}.
	\end{clm}
		
	\begin{proof}
		Suppose that we have a schedule $\sigma$ for $(G, D, A^0)$ with at most $2n$ bankruptcies. Consider the set of all literals $L=\{v_1, \lnot v_1, v_2, \lnot v_2, ..., v_n, \lnot v_n\}$ w.r.t. $\phi$. Define the set of \emph{bankrupt literals} $B\subseteq L$ to consist of every literal whose corresponding (literal) node is bankrupt within $\sigma$. Define $X(w) = True$ iff $w\in L\setminus B$. We claim that $X$ is a (complete) truth assignment. Suppose it is not and that $X(v_i)=X(\neg v_i)=True$; so, both $x_i$ and $\neg x_i$ are bankrupt within $\sigma$. However, in any valid schedule (no matter what the IDM variant) every source node $s_i$ will necessarily go bankrupt and at least one of the literal nodes $x_i$ and $\lnot x_i$ will go bankrupt. Thus, as we have at most $2n$ bankruptcies, our supposition is incorrect. Alternatively, suppose that $X(v_i)=X(\neg v_i)=False$; so, neither $x_i$ nor $\neg x_i$ is bankrupt within $\sigma$. But, as stated, this cannot be the case. So, $X$ is a truth assignment; moreover, $\sigma$ has exactly $2n$ bankruptcies with exactly one of any pair of `oppositely-oriented' literal nodes bankrupt.
		
		Suppose, for contradiction, that $X$ is not a satisfying assignment. So, there exists at least one clause, $c_j$ say, such that every literal in the clause is made $False$ by $X$. By definition of $X$, we have that every literal node corresponding to one of these literals is a bankrupt node. Any such literal node must receive \Euro0 (as the `oppositely-oriented' literal node is not bankrupt and receives \Euro3); consequently, the clause node $q_j$ receives \Euro0 and is bankrupt. This yields a contradiction as we have exactly $2n$ bankrupt nodes (as detailed above). 
		\end{proof}
		
		Consequently, $\phi$ is satisfiable iff the IDM instance $(G,D,A^0)$ admits a schedule with at most $2n$ bankruptcies (again, this holds for each IDM variant). This concludes our proof.
	\end{proof}
	
	Our next result is perhaps rather surprising in that we restrict to IDM instances $(G,D,A^0)$ where $G$ is a fixed digraph.
	
	\begin{thm}\label{thm:bankminconstant}
		For each of the AoN, PP and FP variants, the problem \textsc{Bankruptcy Minimization} is weakly NP-complete, even when we restrict to instances $((G,D,A^0), k)$ where $G$ is a fixed, specific digraph with 32 nodes and $k=16$.
	\end{thm}
		
	\begin{proof}
		We build a polynomial-time reduction from \textsc{Equal Cardinality Partition} to \textsc{Bankruptcy Minimization} (\textsc{Equal Cardinality Partition}, defined below, was proven weakly NP-complete in \cite{Karp72}).\smallskip
		
		\noindent\underline{\textsc{Equal Cardinality Partition}}
		\begin{description}
			\item Instance: a multi-set of positive integers $S=\{a_1, a_2, ..., a_n\}$ where $n$ is even and with sum $sum(S)=2k$
			\item Yes-instance: there exist a partition of $S$ into two equal-sized sets $S_1$ and $S_2$ such that $sum(S_1)=sum(S_2)=k$.
		\end{description}
		The size of such an instance is $nb$ where $b$ is the least number of bits so that any integer $a_i$ can be represented in binary using $b$ bits.
		
		Given an instance $S=\{a_1, a_2, ..., a_n\}$ of \textsc{Equal Cardinality Partition} (where $n\geq 1$), we construct the IDM instance $(G,D,A^0)$ that is illustrated in Fig.~\ref{fig:o1 vertices}. Every appearance of the shaded node $p$ in Fig.~\ref{fig:o1 vertices} corresponds to the same single node, that we refer to as the sink, and thus this instance has 32 nodes in total. We use the symbol `$\infty$' to denote a suitably high monetary amount (though $2k+n+1$ suffices) and $T=10n+7$. Our IDM instance can trivially be constructed from $S$ in time polynomial in the size of the instance $S$. We now show that $(G,D,A^0)$ admits a valid schedule with at most 16 bankruptcies iff $S$ is a yes-instance of \textsc{Equal Cardinality Partition} (no matter whether we are in the AoN, PP or FP variant). Until further stated, we will work solely within the PP variant and return to the AoN and FP variants later. 
		
			\begin{figure}[!ht]
			\centering
			\scalebox{.62}{
				\begin{tikzpicture}
					\myassetnode{s}{0}{0}{$s$}{$2k$}
					\begin{scope}[xshift=0cm, yshift=-3cm]
						\node[state,circle,draw=red, dashed, text=red] (m1) at (0,0) {$m_1$};%
						\node[state,circle,draw=black, text=black, fill=lightgray] (t01) at (-2,0) {$p$};%
						\node[state,draw=none] (m1_out) at (-2,0) {};%
						\path[draw=red, dashed, text=red] (m1) edge[draw=red] node[above] {$\infty@1$} (m1_out);	
						\node[state,circle] (m2) at (1.5,-2) {$m_2$};%
						\node[state,circle,draw=red, dashed, text=red] (m3) at (0,-4) {$m_3$};%
						\path[draw=red, dashed, text=red] (m3) edge[draw=red] node[left] {$\infty@1$} (m1);	
						\path (m1) edge node[align=center, yshift = -10pt, xshift = 10pt] {$a_1@10,~ a_1@15$\\$...$\\$a_n@10n,~a_n@10n+5$} (m2);
						\path (m2) edge node[align=center, yshift = 10pt] {$a_1@10,~ a_1@15$\\$...$\\$a_n@10n,~a_n@10n+5$} (m3);
						\path (s) edge node[align=center] {$a_1@[10, 15]$\\$...$\\$a_n@[10n, 10n+5]$} (m1);
					\end{scope}
					\begin{scope}[xshift=-2.5cm, yshift=-9cm]
						\begin{scope}[xshift=0cm, yshift=0cm]
							\node[state,circle,draw=red, dashed, text=red] (m4a) at (0,0) {$m_4^A$};%
							\node[state,circle,draw=black, text=black, fill=lightgray] (t02) at (-2,0) {$p$};%
							\path[draw=red, dashed, text=red] (m4a) edge[draw=red] node[align=center, above] {$\infty@1$} (t02);	
						\end{scope}
						\node[state,circle] (m5a) at (1.5,-2) {$m_5^A$};%
						\node[state,circle,draw=red, dashed, text=red] (m6a) at (0,-4) {$m_6^A$};%
						
						\path[draw=red, dashed, text=red] (m6a) edge[draw=red] node[left] {$\infty@1$} (m4a);	
						\path (m4a) edge node[align=center, yshift = -5pt, xshift = 0pt] {$a_1@10$\\$...$\\$a_n@10n$} (m5a);
						\path (m5a) edge node[align=center, yshift=10pt, xshift = 2pt] {$a_1@10$\\$...$\\$a_n@10n$} (m6a);
						\path[draw=red, dashed, text=red] (m3) edge[draw=red] node[below, xshift=5pt] {$\infty@1$} (m4a);	
					\end{scope}
					
					\node[state, circle, draw=black, text=black] (m7a) at (-4, -15) {$m_7^A$};
					\path[draw=red, dashed, text=red] (m6a) edge[draw=red] node[right] {$\infty@1$} (m7a);	
					
					\begin{scope}[xshift=-2.7cm, yshift=-15cm]
						\begin{scope}[xshift=-2.5cm, yshift=3cm]
							\node[state,circle,draw=red, dashed, text=red] (m8a) at (0,0) {$m_8^A$};%
							\node[state,circle,draw=black, text=black, fill=lightgray] (t04) at (0,2) {$p$};%
							\path[draw=red, dashed, text=red] (m8a) edge[draw=red] node[align=center, right] {$\infty@1$} (t04);	
						\end{scope}
						\node[state,circle] (m9a) at (-5,3) {$m_9^A$};%
						\node[state,circle] (m10a) at (-5,0) {$m_{10}^A$};%
						\node[state,circle,draw=red, dashed, text=red] (m11a) at (-5,-3) {$m_{11}^A$};%
						\begin{scope}[xshift=-3cm, yshift=-3cm]
							\node[state,circle] (m12a) at (0.5,0) {$m_{12}^A$};
							\node[state,square,scale=0.5, fill=white!100] (m12a_asset) at (1, -.5) {$n/2$};
						\end{scope}
						
						\path (m7a) edge node[align=center, left, xshift=2pt, yshift=-6pt] {\\$1@10$\\$...$\\$1@10n$} (m8a);
						\path (m8a) edge node[align=center, above, yshift=5pt] {$1@11$\\$...$\\$1@10n+1$} (m9a);
						\path (m9a) edge node[align=center, left] {$1@11$\\$...$\\$1@10n+1$} (m10a);
						\path (m10a) edge node[align=center, left] {$1@11$\\$...$\\$1@10n+1$} (m11a);
						\path[draw=red, dashed, text=red] (m11a) edge[draw=red] node[xshift=10pt, yshift=20pt] {$\infty@1$} (m8a);	
						\path[draw=red, dashed, text=red] (m11a) edge[draw=red] node {$\infty@1$} (m12a);	
						\path (m12a) edge node[align=center, left,yshift=10pt] {$1@11$\\$...$\\$1@10n+1$\\} (m7a);
					\end{scope}

					\begin{scope}[xshift=-2.5cm, yshift=-19cm]
						
						\begin{scope}[xshift=0cm, yshift=0cm]
							\node[state,circle,draw=red, dashed, text=red] (m13a) at (0,0) {$m_{13}^A$};%
							\node[state,circle,draw=black, text=black, fill=lightgray] (t06) at (2,0) {$p$};%
							\path[draw=red, dashed, text=red] (m13a) edge[draw=red] node[align=center, below] {$\infty@1$} (t06);
						\end{scope}
						\node[state,circle] (m14a) at (-1.5,-2) {$m_{14}^A$};%
						\node[state,circle,draw=red, dashed, text=red] (m15a) at (0,-4) {$m_{15}^A$};%
						\node[state,circle] (m16a) at (0,-6) {$m_{16}^A$};
						
						\path (m7a) edge node[align=center, above, xshift=20pt] {$k@[1,T]$} (m13a);
						\path (m13a) edge node[align=center, left, yshift=10pt] {$a_1@12$\\$...$\\$a_n@10n+2$\\} (m14a);
						\path (m14a) edge node[align=center, left] {\\\\$a_1@12$\\$...$\\$a_n@10n+2$} (m15a);
						\path[draw=red, dashed, text=red] (m15a) edge[draw=red] node[align=center, right] {$\infty@1$} (m13a);	
						\path[draw=red, dashed, text=red] (m15a) edge[draw=red] node {$\infty@1$} (m16a);	
					\end{scope}
					
					\begin{scope}[xshift=2.5cm, yshift=-9cm]
						\begin{scope}[xshift=0cm, yshift=0cm]
							\node[state,circle,draw=red, dashed, text=red] (m4b) at (0,0) {$m_4^B$};%
							\node[state,circle,draw=black, text=black, fill=lightgray] (t03) at (2,0) {$p$};%
							\path[draw=red, dashed, text=red] (m4b) edge[draw=red] node[align=center, above] {$\infty@1$} (t03);	
						\end{scope}
						\node[state,circle] (m5b) at (-1.5,-2) {$m_5^B$};%
						\node[state,circle,draw=red, dashed, text=red] (m6b) at (0,-4) {$m_6^B$};%
						
						\path[draw=red, dashed, text=red] (m6b) edge[draw=red] node[right] {$\infty@1$} (m4b);	
						\path (m4b) edge node[align=center,left,yshift=13pt, xshift=2pt] {$a_1@15$\\$...$\\$a_n@10n+5$\\} (m5b);
						\path (m5b) edge node[align=center,left,yshift=-10pt, xshift=3pt] {$a_1@15$\\$...$\\$a_n@10n+5$} (m6b);
						\path[draw=red, dashed, text=red] (m3) edge[draw=red] node[below, xshift=-10pt] {$\infty@1$} (m4b);	
					\end{scope}
					
					\node[state, circle, draw=black, text=black] (m7b) at (4, -15) {$m_7^B$};
					\path[draw=red, dashed, text=red] (m6b) edge[draw=red] node[left] {$\infty@1$} (m7b);	
					
					\begin{scope}[xshift=2.2cm, yshift=-15cm]
						\begin{scope}[xshift=3cm, yshift=3cm]
							\node[state,circle,draw=red, dashed, text=red] (m8b) at (0,0) {$m_8^B$};%
							\node[state,circle,draw=black, text=black, fill=lightgray] (t05) at (0,2) {$p$};%
							\path[draw=red, dashed, text=red] (m8b) edge[draw=red] node[align=center, left] {$\infty@1$} (t05);
						\end{scope}
						\node[state,circle] (m9b) at (5.5,3) {$m_9^B$};%
						\node[state,circle] (m10b) at (5.5,0) {$m_{10}^B$};%
						\node[state,circle,draw=red, dashed, text=red] (m11b) at (5.5,-3) {$m_{11}^B$};%
						\begin{scope}[xshift=3cm, yshift=-3cm]
							\node[state,circle] (m12b) at (0,0) {$m_{12}^B$};
							\node[state,square,scale=0.5, fill=white!100] (m12b_asset) at (0.5, -.5) {$n/2$};
						\end{scope}
						
						\path (m7b) edge node[align=center, right, yshift=-5pt, xshift=-9pt] {\\$1@15$\\$...$\\$1@10n+5$} (m8b);
						\path (m8b) edge node[align=center, above, yshift=5pt] {$1@16$\\$...$\\$1@10n+6$} (m9b);
						\path (m9b) edge node[align=center, right] {$1@16$\\$...$\\$1@10n+6$} (m10b);
						\path (m10b) edge node[align=center, right] {$1@16$\\$...$\\$1@10n+6$} (m11b);
						\path[draw=red, dashed, text=red] (m11b) edge[draw=red] node[align=center, right, yshift=15pt, xshift=-5pt] {$\infty@1$} (m8b);	
						\path[draw=red, dashed, text=red] (m11b) edge[draw=red] node[align=center, above] {$\infty@1$} (m12b);	
						\path (m12b) edge node[align=center, right,yshift=0pt, xshift=3pt] {$1@16$\\$...$\\$1@10n+6$\\} (m7b);
					\end{scope}
					
					\begin{scope}[xshift=2.5cm, yshift=-19cm]
						
						\begin{scope}[xshift=0cm, yshift=0cm]
							\node[state,circle,draw=red, dashed, text=red] (m13b) at (0,0) {$m_{13}^B$};%
							\node[state,circle,draw=black, text=black, fill=lightgray] (t07) at (-2,0) {$p$};%
							\path[draw=red, dashed, text=red] (m13b) edge[draw=red] node[align=center, below] {$\infty@1$} (t07);
						\end{scope}
						\node[state,circle] (m14b) at (1.5,-2) {$m_{14}^B$};%
						\node[state,circle,draw=red, dashed, text=red] (m15b) at (0,-4) {$m_{15}^B$};%
						\node[state,circle] (m16b) at (0,-6) {$m_{16}^B$};
						
						\path (m7b) edge node[align=center, above, xshift=-20pt] {$k@[1,T]$} (m13b);
						\path (m13b) edge node[align=center, right, yshift=10pt] {$a_1@17$\\$...$\\$a_n@10n+7$\\} (m14b);
						\path (m14b) edge node[align=center, right] {\\\\$a_1@17$\\$...$\\$a_n@10n+7$} (m15b);
						\path[draw=red, dashed, text=red] (m15b) edge[draw=red] node {$\infty@1$} (m13b);	
						\path[draw=red, dashed, text=red] (m15b) edge[draw=red] node {$\infty@1$} (m16b);	
					\end{scope}
				
				\begin{scope}[xshift=0cm, yshift=-27cm]
					\node[state, circle] (d) at (0, 0) {$d$};
					\path (m16a) edge node[align=center, right, yshift=10pt] {$k@T$} (d);
					\path (m16b) edge node[align=center, left, yshift=10pt] {$k@T$} (d);
				\end{scope}

				\end{tikzpicture}
			}
			\caption{Construction of an IDM instance (with $k=16$) corresponding to the \textsc{Equal Cardinality Partition} instance $S=\{a_1, \ldots, a_n\}$. Dashed red edges are ``practically infinite'' bankrupting debts.}
			\label{fig:o1 vertices}
		\end{figure}
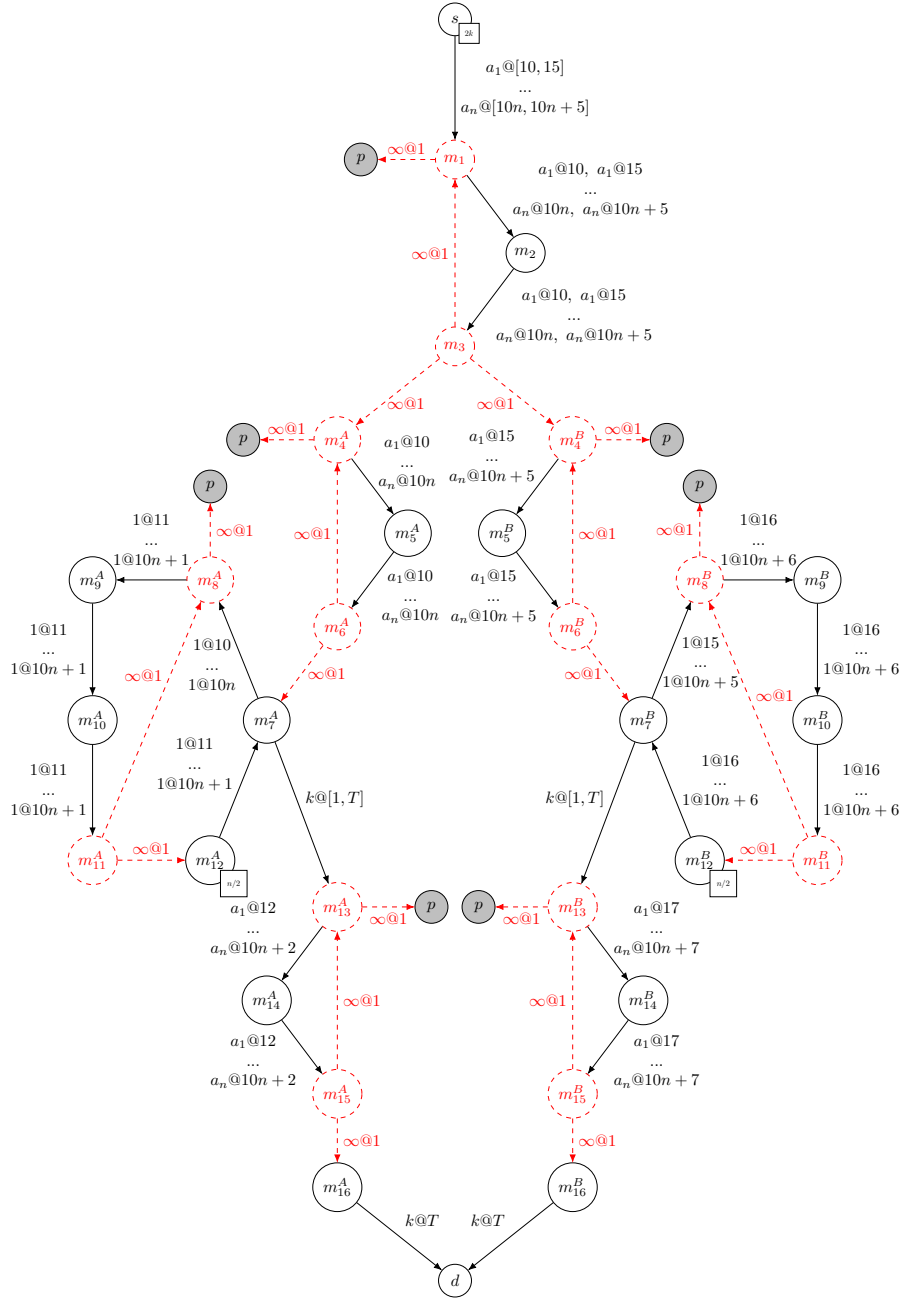
		
		\begin{clm}\label{clm:bankminconstant1}
			If the IDM instance $((G,D,A^0), 16)$ is a yes-instance of \textsc{Bankruptcy Minimization} then $S$ is a yes-instance of \textsc{Equal Cardinality Partition}.
		\end{clm}
		
		\begin{proof}
		Suppose that $(G,D,A_e^0)$ admits a valid schedule $\sigma$ with at most 16 bankruptcies. Note that the 14 nodes $m_1$, $m_3$, $m_4^A$, $m_6^A$, $m_8^A$, $m_{11}^A$, $m_{13}^A$, $m_{15}^A$, $m_4^B$, $m_6^B$, $m_8^B$, $m_{11}^B$, $m_{13}^B$ and $m_{15}^B$ are necessarily bankrupt in \emph{every} valid schedule because they are debtors of debts of monetary amount $\infty$ at time 1 and no payments are made before time 10 (these nodes are dashed and highlighted in red in Fig.~\ref{fig:o1 vertices} as are the debts at time 1 of monetary amount $\infty$). Note also that these nodes can never have any cash assets at any time and nor can the nodes $m_2$, $m_5^A$, $m_9^A$, $m_{10}^A$, $m_{14}^A$, $m_5^B$, $m_9^B$, $m_{10}^B$ and $m_{14}^B$ (as they would otherwise be withholding). Consequently, we can only have at most another 2 nodes going bankrupt within $\sigma$. We begin by showing that one of $\{m_7^A,m_{12}^A\}$ must be bankrupt and one of $\{m_7^B,m_{12}^B\}$ must be bankrupt (which will account for all bankrupt nodes).
		
		Suppose that none of the nodes $m_7^A$, $m_9^A$, $m_{10}^A$ and $m_{12}^A$ are bankrupt in $\sigma$. By considering $m_{12}^A$ at the times $t\in\{11,21,\ldots,10n+1\}$, on at least $\frac{n}{2}$ of these occasions $m_{12}^A$ must have received at least \Euro1 via payments from $m_{11}^A$ (so as to service all its debts to $m_7^A$). Consider the first occasion $t^\prime\in\{11,21,\ldots,10n+1\}$ that $m_{12}^A$ receives a non-zero payment from $m_{11}^A$ (note that all payments from $m_{11}^A$ to $m_{12}^A$ are made at some time from $\{11,21,\ldots,10n+1\}$). There must have been a payment of \Euro1 from $m_{10}^A$ to $m_{11}^A$ at time $t^\prime$ as well as payments of \Euro1 from $m_9^A$ to $m_{10}^A$ and from $m_8^A$ to $m_9^A$ at time $t^\prime$. So, $m_8^A$ must receive at least \Euro1 from $m_7^A$ and $m_{11}^A$ at time $t^\prime$. As $m_{11}^A$ makes a non-zero payment to $m_{12}^A$ at time $t^\prime$, any payment from $m_{11}^A$ to $m_7^A$ at time $t^\prime$ must be for some amount strictly less than \Euro1. Consequently, $m_8^A$ must receive a non-zero payment from $m_7^A$ at time $t^\prime$. The only way that this can happen is if there is an overdue debt from $m_7^A$ to $m_8^A$; that is, $m_7^A$ is bankrupt, which yields a contradiction. Hence, at least one of $m_7^A$, $m_9^A$, $m_{10}^A$ and $m_{12}^A$ is bankrupt. An analogous argument shows that at least one of $m_7^B$, $m_9^B$, $m_{11}^B$ and $m_{12}^B$ is bankrupt. Hence, exactly one of $m_7^A$, $m_9^A$, $m_{10}^A$ and $m_{12}^A$ is bankrupt and exactly one of $m_7^B$, $m_9^B$, $m_{11}^B$ and $m_{12}^B$ is bankrupt.
		
		Suppose that $m_9^A$ is bankrupt. So, $m_{10}^A$ is necessarily bankrupt. Conversely, if $m_{10}^A$ is bankrupt then $m_9^A$ must be bankrupt also. Consequently, neither $m_9^A$ nor $m_{10}^A$ is bankrupt and we must have that either $m_7^A$ or $m_{12}^A$ is bankrupt and analogously either $m_7^B$ or $m_{12}^B$ is bankrupt. In particular, $s$, $m_2$, $m_5^A$, $m_{13}^A$, $m_{15}^A$, $m_5^B$, $m_{13}^B$ and $m_{15}^B$ are not bankrupt.
		
		Let us turn to analysing the flow of resource via the schedule $\sigma$. As $\sigma$ is valid, we must have that both debts from $m_{16}^A$ and $m_{16}^B$ are paid on time with \Euro$2k$ in total reaching $d$. The question is: does this resource consist of the \Euro$2k$ emanating from $s$ or does it consist of resource emanating from $s$ but supplemented with resource emanating from $m_{12}^A$ or $m_{12}^B$? Let us look at the possible debt payments at time 10 (which is the earliest time that payments can be made). In particular, let us look at payments made by $m_6^A$ to $m_7^A$ at time 10. Note that all payments made from $m_6^A$ to $m_7^A$ are at a time from $\{10,20,\ldots,10n\}$.\smallskip
		
		\noindent\underline{Case (\emph{a})}: An amount of \Euro$x>0$ is paid from $m_6^A$ to $m_7^A$ at time 10.\smallskip
		
		\noindent There are two essential sub-cases at time 10:
		\begin{itemize}
			\item[(\emph{i})] $m_7^A$ services its debt to $m_8^A$, which pays \Euro 1 to the sink, with perhaps \Euro$y\geq 0$ paid from $m_7^A$ to $m_{13}^A$ and from there to the sink, so that \Euro$x-y-1\geq 0$ resides at $m_7^A$ in cash assets at time 10
			\item[(\emph{ii})] $m_7^A$ does not service its debt to $m_8^A$ and pays \Euro$x$ to $m_{13}^A$ with this payment immediately going to the sink, so that \Euro0 resides at $m_7^A$ in cash assets at time 10; hence, $m_7^A$ is bankrupt (note that no cash assets can reside at $m_7^A$ at time 10 as otherwise $m_7^A$ would be withholding).
		\end{itemize}
	
		Now consider what happens at time 11. Suppose that we are in Case (\emph{a}.\emph{i}). There must be a payment-cycle involving $m_8^A$, $m_9^A$, $m_{10}^A$ and $m_{11}^A$ and as $\frac{n}{2}\geq 1$, $m_{12}^A$ must service its debt to $m_7^A$, with perhaps $m_7^A$ paying \Euro$z\geq 0$ to $m_{13}^A$ which is immediately paid to the sink. The \Euro1 from $m_{12}^A$ does not supplement the resource emanating from $s$ but just `replaces' \Euro1 which was `lost' to the sink at time 10. Note that the cash assets of $m_{12}^A$ at time 11 are $\frac{n}{2}-1$.
		
		Suppose that we are in Case (\emph{a}.\emph{ii}). Note that as $m_7^A$ is bankrupt, $m_{12}^A$ can never become bankrupt and so must service its debts when required. There are four possibilities as regards what happens at time 11 (bearing in mind the overdue debt from $m_7^A$ to $m_8^A$):
		\begin{itemize}
			\item[(1)] $m_{12}^A$ pays \Euro1 to $m_7^A$ which pays \Euro1 to $m_8^A$ which immediately goes to the sink, with a payment-cycle involving $m_8^A$, $m_9^A$, $m_{10}^A$ and $m_{11}^A$
			\item[(2)] $m_{12}^A$ pays \Euro1 to $m_7^A$ which pays \Euro1 to $m_8^A$ which pays \Euro1 to $m_9^A$ which pays \Euro1 to $m_{10}^A$ which pays \Euro1 to $m_{11}^A$ which pays \Euro1 to $m_8^A$ which immediately goes to the sink
			\item[(3)] $m_{12}^A$ pays \Euro1 to $m_7^A$ which pays \Euro1 to $m_{13}^A$ which immediately goes to the sink, with a payment-cycle involving $m_8^A$, $m_9^A$, $m_{10}^A$ and $m_{11}^A$
			\item[(4)] $m_{12}^A$ pays \Euro1 to $m_7^A$ which pays \Euro1 to $m_8^A$ which pays \Euro1 to $m_9^A$ which pays \Euro1 to $m_{10}^A$ which pays \Euro1 to $m_{11}^A$ which pays \Euro1 to $m_{12}^A$; that is, we have a payment cycle involving $m_7^A$, $m_8^A$, $m_9^A$, $m_{10}^A$, $m_{11}^A$ and $m_{12}^A$.
		\end{itemize}
		In (1-3) above, the \Euro1 from $m_{12}^A$ does not supplement the resource emanating from $s$ but is lost to the sink (as are the \Euro$x>0$ at time 10). Note that in (3), the debt from $m_7^A$ to $m_8^A$ at time 10 is overdue and cannot be paid at time 11 as $m_7^A$ has no cash assets at time 10; so it remains overdue. In (4), again no supplement is made, although \Euro$\frac{n}{2}$ still resides at $m_{12}^A$ in cash assets at time 11, and \Euro$x>0$ emanating from $s$ has been lost to the sink. 
		
		In both Case (\emph{a}.\emph{i}) and Case (\emph{a}.\emph{ii}), at time 12, the only possible non-payment-cycle payments involve $s$, $m_1$, $m_{7}^A$, $m_{13}^A$, $m_{14}^A$ and $m_{15}^A$ but any such payments do not affect whether the resources emanating from $m_{12}^A$ or $m_{12}^B$ supplement the resource emanating from $s$. In Case (\emph{a}.\emph{ii}.3), the debt from $m_{7}^A$ to $m_8^A$ at time 10 is overdue and so must be paid from the cash assets of $m_{7}^A$ at time 12, if it has any and unless all of these assets are paid to $m_{13}^A$. All such payments by $m_7^A$ will immediately go to the sink.\smallskip
		
		\noindent \underline{Case (\emph{b})}: No payment is made from $m_6^A$ to $m_7^A$ at time 10.\smallskip
		
		\noindent Consequently, $m_7^A$ cannot pay its debt to $m_8^A$ and becomes bankrupt. At time 11, $m_{12}^A$ must necessarily service its debt to $m_7^A$ and there are four possibilities with these possibilities being exactly the possibilities (1-4) in Case (\emph{a}.\emph{ii}). As before, at time 12, the only possible non-payment-cycle payments made involve $s$, $m_1$, $m_{7}^A$, $m_{13}^A$, $m_{14}^A$ and $m_{15}^A$ but any such payments do not affect whether the resources emanating from $m_{12}^A$ or $m_{12}^B$ 
		supplement the resources emanating from $s$. Note that in (3), the debt from $m_7^A$ to $m_8^A$ at time 10 is still overdue at time 11 and cannot be paid at time 12 as $m_7^A$ has no cash assets at time 12; so it remains overdue. In (4), again no supplement is made, although \Euro$\frac{n}{2}$ still resides at $m_{12}^A$ in cash assets at time 12.\smallskip
		
		So, the resources emanating from $m_{12}^A$ cannot supplement the resources emanating from $s$ at any time $t<15$. An identical argument can be applied to time 15 so as to yield a similar conclusion as regards the resources emanating from $m_{12}^B$ at any time $t<20$. 
		
		Consider the situation at time 20. With regard to the sub-network involving $m_6^A$, $m_7^A$, $m_8^A$, $m_9^A$, $m_{10}^A$, $m_{11}^A$, $m_{12}^A$ and $m_{13}^A$ (that is, the sub-network of study above), the situation is similar to that at time 10 except that there may be additional restrictions on what can happen given: a possible existing overdue debt from $m_7^A$ to $m_8^A$ (the debt due at time 10); possible non-zero cash assets at $m_7^A$; and possibly reduced cash assets at $m_{12}^A$ of $\frac{n}{2}-1$. Note that if $m_7^A$ has cash assets prior to time 20 then this can be thought of as $m_7^A$ having acquired these assets from $m_6^A$ at time 20; that is, we are in Case (\emph{a}) above. Given the fact that the situation at time $20$ is a restricted version of the situation at time $10$ where the resources emanating from $m_{12}^A$ could not supplement those emanating from $s$, the same is true again. By analysing each time $t=25, 30, 35, \ldots$, we can see that no resources emanating from either $m_{12}^A$ or $m_{12}^B$ can supplement those emanating from $s$. Hence, in order to secure total cash assets of \Euro$2k$ at $d$ after time $T$, we need that all of the \Euro$2k$ resource emanating from $s$ reaches $d$; that is, none of it is lost to the sink en route (although some of it might have been `replaced' as per Case (\emph{a}.\emph{i})).
		
		We can now repeat the above analysis except that now we know that we cannot lose resource emanating from $s$ unless it is replaced as in Case (\emph{a}.\emph{i}). This simplifies things considerably. If $m_7^A$ has \Euro$x>0$ at time $t\in\{10,20,\ldots,10n\}$ (either as cash assets or from a payment by $m_6^A$ at time $t$) then it must be the case that $m_7^A$ services the debt to $m_8^A$ at time $t$, \Euro1 is lost to the sink and $m_7^A$ retains $x-1$ in cash assets. At time $t+1$, $m_{12}^A$ must service its debt to $m_7^A$ so as to replace the lost \Euro1, with \Euro$x$ residing at $m_7^A$ in cash assets at time $t+1$. Consequently, there can only be at most $\frac{n}{2}$ times in $\{10,20,\ldots,10n\}$ when $m_7^A$ receives a payment from $m_6^A$ (recall that $m_6^A$ only makes payments to $m_7^A$ at times from $\{10,20,\ldots,10n\}$). When $m_7^A$ either has no cash assets or receives no payment from $m_6^A$ at time $t\in\{10,20,\ldots,10n\}$, we either lose \Euro1 of cash assets from $m_{12}^A$ to the sink (and so we also lose some capacity to `replace' resource emanating from $s$ that is lost to the sink) or we are in case (4) above and have a payment cycle involving $m_7^A$, $m_8^A$, $m_9^A$, $m_{10}^A$, $m_{11}^A$ and $m_{12}^A$. Analogous comments can be made as regards the corresponding nodes superscripted $B$ and times in $\{15, 25, \ldots, 10n+5\}$. 
		
		Bearing in mind that none of the resource emanating from $s$ goes to the sink before it reaches either $m_6^A$ or $m_6^B$, at least $n$ distinct payments are made from $s$ and these payments result in at least $n$ payments in total from $m_6^A$ to $m_7^A$ or from $m_6^B$ to $m_7^B$. Thus, from above, $m_6^A$ must make exactly $\frac{n}{2}$ payments to $m_7^A$ and $m_6^B$ must make exactly $\frac{n}{2}$ payments to $m_7^B$. This means that any payment from $m_6^A$ to $m_7^A$ or from $m_6^B$ to $m_7^B$ must be for an amount from $\{a_1,a_2,\ldots,a_n\}$ and we have a partition of $\{a_1,a_2,\ldots,a_n\}$ into equal-sized sets both of whose sum is $k$; that is, our instance $S$ of \textsc{Equal Cardinality Partition} is a yes-instance and the claim follows.
		\end{proof}
		
		\begin{clm}\label{clm:bankminconstant2}
			If $S$ is a yes-instance of \textsc{Equal Cardinality Partition} then $((G,D,A^0),16)$ is a yes-instance of \textsc{Bankruptcy Minimization}.
		\end{clm}
	
		\begin{proof}
		Suppose that our instance $S=\{a_1,a_2,\ldots,a_n\}$ of \textsc{Equal Cardinality Partition} is such that $n=2m$ and $\sum_{i=1}^m a_{\alpha_i} = \sum_{i=1}^m a_{\beta_i}$, where $\{\alpha_i,\beta_i: 1\leq i\leq m\} = \{1,2,\ldots,n\}$. We need to build a valid schedule $\sigma$ for $(G,D,A^0)$ with at most 16 bankruptcies. Let $1\leq i\leq n$ and suppose that $i=\alpha_j$, where $1\leq j\leq m$. The node $s$ pays its $i$th debt to $m_1$ (that is, the debt $a_i@[10i,10i+5]$) at time $10i$ and this payment is percolated all the way down to $m_7^A$ at time $10i$. The debt $1@10i$ from $m_7^A$ to $m_8^A$ is paid at time $10i$ with this \Euro1 being replaced from $m_{12}^A$ at time $10i+1$ (see Case (\emph{a}.\emph{i}) from the proof of Claim~\ref{clm:bankminconstant1} above). We also have a payment cycle involving $m_8^A$, $m_9^A$, $m_{10}^A$ and $m_{11}^A$ at time $10i+1$. The cash assets of $a_i$ at $m_7^A$ are percolated down to $m_{16}^A$ at time $10i+2$. At time $10i+5$, we have suitable payment cycles involving: $m_1$, $m_2$ and $m_3$; $m_4$, $m_5$ and $m_6$; and $m_{13}^A$, $m_{14}^A$ and $m_{15}^A$. We also have a payment cycle involving $m_7^A$, $m_8^A$, $m_9^A$, $m_{10}^A$, $m_{11}^A$ and $m_{12}^A$ (see (4) from the proof of Claim~\ref{clm:bankminconstant1} above). An analogous course of action is taken if $i=\beta_j$, for some $1\leq j\leq m$. The resulting schedule is valid and the 16 nodes $m_1$, $m_3$, $m_4^A$, $m_6^A$, $m_7^A$, $m_{8}^A$, $m_{11}^A$, $m_{13}^A$, $m_{15}^A$, $m_4^B$, $m_6^B$, $m_7^B$, $m_{8}^B$, $m_{11}^B$, $m_{13}^B$ and $m_{15}^B$ are bankrupt. The claim follows.
		\end{proof}
		So, our main result holds for the PP variant of \textsc{Bankruptcy Minimization}. Note that everything above holds for the AoN variant too, though Theorem \ref{thm:AoN3path} is a strictly stronger result for that setting.
		
		Let us consider now the FP variant. As it happens, an argument similar to that above works within the FP variant although there are more complicated nuances. Rather than repeat the whole nuanced argument in detail, and given the above complete proof for the PP variant, we only sketch the proof for the FP variant. Henceforth, we assume that we are working within the FP variant.
		
		Consider the proof of the corresponding version of Claim~\ref{clm:bankminconstant1}. The reasoning that establishes that we must have that either $m_7^A$ or $m_{12}^A$ is bankrupt and that either $m_7^B$ or $m_{12}^B$ is bankrupt holds for the FP variant. Consider payments made from $m_6^A$ to $m_7^A$ at time 10.\smallskip
		
		\noindent\underline{Case (\emph{a})}: An amount of \Euro$x>0$ is paid from $m_6^A$ to $m_7^A$ at time 10.\smallskip
		
		\noindent There are two essential sub-cases at time 10:
		\begin{itemize}
			\item[(\emph{i})] $m_7^A$ services its debt to $m_8^A$, which pays \Euro 1 to the sink, with perhaps \Euro$y\geq 0$ paid from $m_7^A$ to $m_{13}^A$ and from there to the sink, so that \Euro$x-y-1\geq 0$ resides at $m_7^A$ in cash assets at time 10
			\item[(\emph{ii})] $m_7^A$ does not fully service its debt to $m_8^A$ but pays \Euro$w$, where $0\leq w < 1$, to $m_8^A$, which immediately goes to the sink, and \Euro$x-w$ to $m_{13}^A$, which immediately goes to the sink, so that \Euro0 resides at $m_7^A$ in cash assets at time 10; hence, $m_7^A$ is bankrupt.
		\end{itemize}
		
		Consider what happens at time 11. In Case (\emph{a}.\emph{i}), there must be a payment cycle involving $m_8^A$, $m_9^A$, $m_{10}^A$ and $m_{11}^A$ and $m_{12}^A$ services its debt to $m_7^A$. The \Euro1 from $m_{12}^A$ does not supplement the resource emanating from $s$ but just `replaces' \Euro1 which was `lost' to the sink at time 10.
		
		Suppose that we are in Case (\emph{a}.\emph{ii}). There are two scenarios:
		\begin{itemize}
			\item[(1)] $m_{12}^A$ pays \Euro1 to $m_7^A$ which pays \Euro$1-w$ to $m_8^A$ of which \Euro$w^\prime$ goes immediately to the sink and \Euro$1-w-w^\prime$ is paid to $m_9^A$; $m_7^A$ has cash assets of at most \Euro$w$ as it may be the case that $m_7^A$ also makes a payment to $m_{13}^A$ which goes straight to the sink, or
			\item[(2)] $m_{12}^A$ pays \Euro1 to $m_7^A$ which pays \Euro$y$ to $m_8^A$, where $0\leq y < 1-w$, of which \Euro$y^\prime$ goes immediately to the sink and \Euro$y-y^\prime$ is paid to $m_9^A$; also, \Euro$1-y$ is paid to $m_{13}^A$ which goes straight to the sink.
		\end{itemize}
		In (1), it must be the case that $m_8^A$ receives at least $w+w^\prime$ from $m_{11}^A$ (so as to fully service its debt to $m_9^A$); hence, $m_{11}^A$ pays at most \Euro$1-w-w^\prime$ to $m_{12}^A$. In any case, \Euro$x$ have been lost to the sink with $m_7^A$ gaining \Euro$w$ from $m_{12}^A$ (with $x \geq w$). In (2), it must be the case that $m_8^A$ receives at least \Euro$1 - (y-y^\prime)$ from $m_{11}^A$ (so as to fully service its debt to $m_9^A$); hence, $m_{11}^A$ pays at most \Euro$y-y^\prime$ to $m_{12}^A$. In any case, \Euro$x$ have been lost to the sink with $m_7^A$ gaining nothing from $m_{12}^A$.\smallskip
		
		\noindent\underline{Case (\emph{b})}: No payment is made from $m_6^A$ to $m_7^A$ at time 10.\smallskip
		
		\noindent At time 11, it must be the case that $m_{12}^A$ services its debt to $m_7^A$ and then we are essentially in Case (\emph{a}.\emph{ii}) above.\smallskip
		
		The outcome is that the resources emanating from $m_{12}^A$ cannot supplement the resources emanating form $s$ at any time $t< 15$. An identical argument can be applied to time 15 so as to yield a similar conclusion as regards the resources emanating from $m_{12}^B$ at any time $t<20$. The rest of the proof of Claim~\ref{clm:bankminconstant1} holds for the FP variant and we have that Claim~\ref{clm:bankminconstant1} holds for the FP variant. 
		
		The schedule $\sigma$ described in the proof of Claim~\ref{clm:bankminconstant2} is a valid schedule in the FP variant and so Claim~\ref{clm:bankminconstant2} also holds for the FP variant. This complete our proof of the main theorem.\end{proof}
	
	The proof of Theorem~\ref{thm:bankminconstant} clearly demonstrates the intricacies of reasoning in our financial networks. By Theorem~\ref{thm:bankminconstant}, it follows that each of the AoN, PP and FP variants of \textsc{Bankruptcy Minimization} are \emph{para-NP-hard} when parameterized by any parameter that is upper-bounded by the number of vertices, such as, e.g., the number of bankruptcies $k$ or the treewidth of the footprint. Note that Theorem~\ref{thm:bankminconstant} concerns weak completeness results (in particular, the integers in an instance of \textsc{Equal Cardinality Matching} appear explicitly as monetary debts in the corresponding instance of \textsc{Bankruptcy Minimization}).
	
	\subsection{Hardness results for \textsc{Perfect Scheduling}}\label{subsec:hardnessPS}
	
	We now turn to \textsc{Perfect Scheduling}. Since this is a subproblem of \textsc{Bankruptcy Minimization} and \textsc{Bailout Minimization}, hardness results in this section also apply to both of those problems.
	
	\begin{thm}
		\label{thm:perfschedDAG}
		The problem \textsc{Perfect Scheduling} is NP-complete for the AoN and PP variants even when we restrict to IDM instances $(G,D,A^0)$ for which: $T\leq 3$; $G$ is a directed acyclic graph with out-degree at most 3 and in-degree at most 3; the monetary amount of any debt is at most \Euro2; and any initial cash assets of a node is at most \Euro3 per bank.
	\end{thm}
	
	\begin{proof}
		Let us work within the PP variant until further notice. We introduce the \emph{multiplier gadget} shown in Fig.~\ref{fig:multiplier} and use this gadget in our main reduction (the gadget sits within the blue dotted line). We claim the following.
		\\
		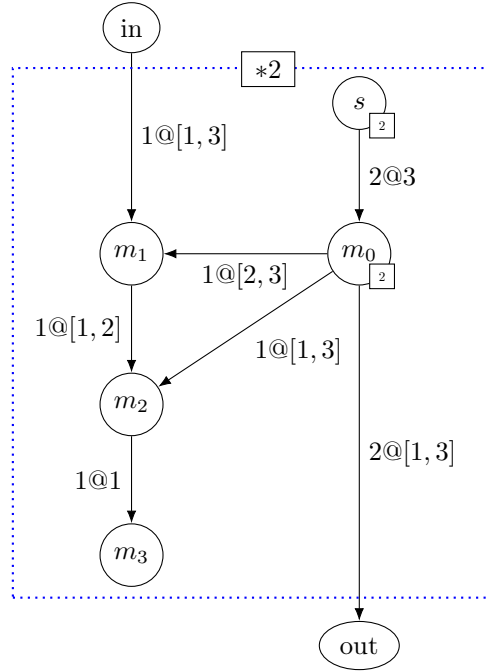
\begin{figure}[!ht]
			\centering  
			\begin{tikzpicture}
				
				\myassetnode{m0}{3}{7}{$m_0$}{$2$}
				\mynode{m1}{0}{7}{$m_1$}
				\mynode{m2}{0}{5}{$m_2$}
				\mynode{m3}{0}{3}{$m_3$}
				\myassetnode{S}{3}{9}{$s$}{$2$}
				
				\node[state] (in) at (0,10) {in};%
				\node[state] (out) at (3,1.8) {out};%

				\path (m0) edge node {$2@[1,3]$} (out);	
				\path (S) edge node {$2@3$} (m0);
				\path (m0) edge node {$1@[2,3]$} (m1);
				\path (m0) edge node {$1@[1,3]$} (m2);		
				\path (m1) edge node[left] {$1@[1,2]$} (m2);
				\path (m2) edge node[left] {$1@1$} (m3);
				\path (in) edge node {$1@[1,3]$} (m1);
				
				\node[draw=blue,dotted, thick,minimum width = 180pt, minimum height = 200pt, inner sep=1cm] at (1.6,5.95) (gadget) {};
				
				\node[state,rectangle, fill=white!100] (label) at (1.8,9.45) {$*2$};
				
			\end{tikzpicture}
			\caption{The multiplier gadget. Intuitively, the gadget ``amplifies'' payments into it at time 1 by a factor 2.}
			\label{fig:multiplier}
		\end{figure}  
		
		\begin{clm}\label{clm:multgadg}
			Assume that the node $in$ of the multiplier gadget has initial cash assets of \Euro1. 
			\begin{itemize}
				\item[(\emph{a})] In any perfect schedule for the multiplier gadget, if no payment is made by $in$ to $m_1$ at either time 1 or time 2 then no payment is made by $m_0$ to $out$ at time 1.
				\item[(\emph{b})] There is a perfect schedule $\sigma_0$ for the multiplier gadget so that a payment is made by $in$ at time 3.
				\item[(\emph{c})] There is a perfect schedule $\sigma_1$ for the multiplier gadget so that a payment is made by $in$ to $m_1$ at time 1 and a payment of \Euro2 is made from $m_0$ to $out$ at time 1. 
			\end{itemize}
		\end{clm}
		
		\begin{proof}
			Suppose that no payment is made by $in$ to $m_1$ at times 1 and 2; so $m_1$ receives no payment at time 1 and makes no payment at time 1. As there is a payment of \Euro1 from $m_2$ to $m_3$ at time 1, there must be a payment of \Euro1 from $m_0$ to $m_2$ at time 1. Suppose that a payment of \Euro1 is made from $m_0$ to $out$ at time 1. If so then $m_0$ can make no payment to $m_1$ at time 2 and $m_1$ is bankrupt which yields a contradiction. Hence, there is no payment from $m_0$ to $out$ at time 1. The statement (\emph{a}) follows.
			
			Consider the schedule whereby:
			\begin{itemize}
				\item at time 1: $m_0$ pays \Euro1 to $m_2$; $m_2$ pays \Euro1 to $m_3$
				\item at time 2: $m_0$ pays \Euro1 to $m_1$; $m_1$ pays \Euro1 to $m_2$
				\item at time 3: $s$ pays \Euro2 to $m_0$; $m_0$ pays \Euro2 to $out$; $in$ pays \Euro1 to $m_1$.
			\end{itemize}
			This yields a perfect schedule and statement (\emph{b}) follows.
			
			Suppose that there is a payment of \Euro1 made from $in$ to $m_1$ at time 1; so, we can also make payments of \Euro1 from $m_1$ to $m_2$ and from $m_2$ to $m_3$ at time 1. Additionally, we can make a payment of \Euro2 from $m_0$ to $out$ at time 1. At time 2, no payments are made. At time 3, we can make payments of: \Euro2 from $s$ to $m_0$; \Euro1 from $m_0$ to $m_2$; and \Euro1 from $m_0$ to $m_1$. This yields a perfect schedule and statement (\emph{c}) follows.
		\end{proof}
		
		Our reduction is from \textsc{3-Sat-3} (again) to \textsc{Perfect Scheduling}. As before, we assume that all \textsc{3-Sat-3} instances are such that every literal appears at least once and at most twice and that no clause contains both a literal and its negation. Given a \textsc{3-Sat-3} instance $\phi$ involving $n$ Boolean variables and $m$ clauses, we construct an IDM instance $(G,D,A^0)$ as portrayed in Fig.~\ref{fig:perfscheddag} (we omit the formal description of $(G,D,A^0)$ as it can be immediately derived from Fig.~\ref{fig:perfscheddag}; moreover, we proceed similarly with other IDM instances that we construct later on). The types of nodes (source, literal, clause and sink) are as in the proof of Theorem~\ref{thm:bankmin}, although we have additional so-called $a$-type nodes, and we abbreviate our multiplier gadget as a square box with $*2$ inside (note that we have $2n$ distinct copies of our multiplier gadget where the node $in$ is taken as $a_i$ and the node $out$ as the literal node $x_i$ or the literal node $\neg x_i$, for $1\leq i\leq n$; of course, we have $m$ clause nodes, $n$ source nodes, $n$ $a$-type nodes and one sink node).
		
		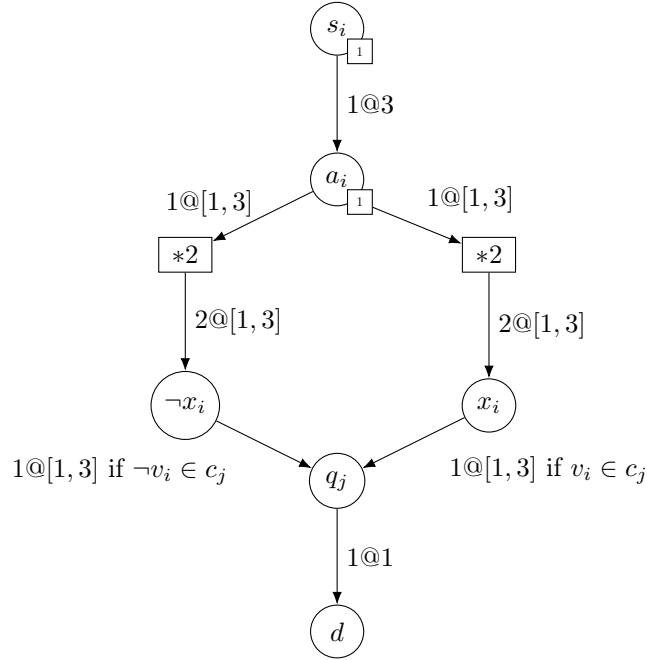
\begin{figure}[!ht]
			\centering
			\begin{tikzpicture}
				\myassetnode{si}{0}{8}{$s_i$}{$1$}
				\myassetnode{ai}{0}{6}{$a_i$}{$1$}
				\mynode{xi}{2}{3}{$x_i$}
				\mynode{nxi}{-2}{3}{$\lnot x_i$}
				\mynode{qj}{0}{2}{$q_j$}
				\mynode{d}{0}{0}{$d$}
				\node[state,rectangle, fill=white!100] (Tmul) at (2,5) {$*2$};
				\node[state,rectangle, fill=white!100] (Fmul) at (-2,5) {$*2$};
				\path (si) edge node[right] {$1@3$} (ai);
				\path (aiasset) edge node[right, yshift=10pt] {$1@[1,3]$} (Tmul);	
				\path (ai) edge node[left, yshift=5pt] {$1@[1,3]$} (Fmul);	
				\path (Tmul) edge node {$2@[1,3]$} (xi);
				\path (Fmul) edge node {$2@[1,3]$} (nxi);
				
				\path (xi) edge node[right, xshift=10pt, yshift=-10pt] {$1@[1,3]$ if $v_i \in c_j$} (qj);	
				\path (nxi) edge node[left, xshift=-10pt, yshift=-10pt] {$1@[1,3]$ if $\lnot v_i \in c_j$} (qj);	
				\path (qj) edge node[right] {$1@1$} (d);	
			\end{tikzpicture}
			\caption{Illustration of the reduction from \textsc{3-SAT 3} to \textsc{Perfect Scheduling} restricted to Directed Acyclic Graphs (DAGs).}
			\label{fig:perfscheddag}
		\end{figure}
				
		\begin{clm}\label{prop:psdag to sat}
			If $(G,D,A^0)$ is a yes-instance of \textsc{Perfect Scheduling} then $\phi$ is a yes-instance of \textsc{3-Sat-3}.
		\end{clm}
		
		\begin{proof}
			Define the truth assignment $X$ via: if, within $\sigma$, $x_i$ receives a payment of at least \Euro1 at time 1 then $X(v_i)=True$; otherwise $X(v_i)=False$.
			
			Fix $1\leq j\leq m$. We must have that $q_j$ pays \Euro1 to $d$ at time 1 and so $q_j$ must receive \Euro1 from some node $x_i$ at time 1 or \Euro1 from some node $\neg x_i$ at time 1.
			
			Suppose that $q_j$ receives \Euro1 from $x_{\tau_j}$ at time 1; in particular, $v_{\tau_j}\in c_j$. Hence, $x_{\tau_j}$ receives at least \Euro1 from its corresponding multiplier gadget at time 1 and so, by definition, $X(v_{\tau_j}) = True$ with the clause $c_j$ satisfied by $X$.
			
			Alternatively, suppose that $q_j$ receives \Euro1 from $\neg x_{\tau_j}$ at time 1; in particular, $\neg v_{\tau_j}\in c_j$. Hence, $\neg x_{\tau_j}$ receives at least \Euro1 from its corresponding multiplier gadget at time 1. By Claim~\ref{clm:multgadg}.\emph{a}, there must be a payment of \Euro1 made from $a_{\tau_j}$ to this multiplier gadget at either time 1 or time 2. Consequently, no payment is made by $a_{\tau_j}$ to the complementary multiplier gadget (that is, the one with a debt to $x_{\tau_j}$) at either time 1 or time 2. By Claim~\ref{clm:multgadg}.\emph{a}, no payment is received by $x_i$ at time 1 and so, by definition, $X(v_{\tau_j})=False$ with the clause $c_j$ satisfied by $X$. The claim follows.
		\end{proof}
		
		\begin{clm}\label{prop:sat to psdag}
			If $\phi$ is a yes-instance of \textsc{3-Sat-3} then $(G,D,A^0)$ is a yes-instance of \textsc{Perfect Scheduling}.
		\end{clm}
		
		\begin{proof}
			Let $X$ be a satisfying truth assignment for $\phi$. For each clause $c_j$, let $v_{\tau_j}$ be a Boolean variable whose occurrence in $c_j$, either via the literal $v_{\tau_j}$ or the literal $\neg v_{\tau_j}$, leads to $c_j$ being satisfiable. So, we get a list $L = v_{\tau_1},v_{\tau_2},\ldots,v_{\tau_m}$ of `satisfying' Boolean variables, possibly with repetitions although no variable appears in the list more than twice and if a Boolean variable $v$ does appear twice then the corresponding literals in the two corresponding clauses are both positive or both negative (of course, this stems from the format of $\phi$ as an instance of \textsc{3-Sat-3}): if the occurrences are positive then we say that $v$ has \emph{positive polarity}, with \emph{negative polarity} defined analogously. Note that any debt from a literal node to a clause node in our IDM instance exists solely because of the occurrence of the corresponding literal in the corresponding clause; in particular, we can never have debts from both literal nodes corresponding to some variable to the same clause node.
						
			Consider the following schedule $\sigma$. At time 1, for each $1\leq j\leq m$:
			\begin{itemize}
				\item $q_j$ pays \Euro1 to $d$
				\item if $v_{\tau_j}$ appears in $L$ with positive (resp. negative) polarity then $x_{\tau_j}$ (resp. $\neg x_{\tau_j}$) pays \Euro1 to $q_j$
				\item if $v_{\tau_j}$ appears in $L$ with positive (resp. negative) polarity then $a_{\tau_j}$ pays \Euro1 to the multiplier gadget which has a debt to $x_{\tau_j}$ (resp. $\neg x_{\tau_j}$).
			\end{itemize} 
			In addition, for each $1\leq j\leq m$, if $v_{\tau_j}$ appears in $L$ with positive (resp. negative) polarity then:
			\begin{itemize}
				\item within the multiplier gadget that has a debt to $x_{\tau_j}$ (resp. $\neg x_{\tau_j}$), we include the schedule $\sigma_1$ from Claim~\ref{clm:multgadg}.\emph{c}
				\item within the multiplier gadget that has a debt to $\neg x_{\tau_j}$ (resp. $x_{\tau_j}$), we include the schedule $\sigma_0$ from Claim~\ref{clm:multgadg}.\emph{b}.
			\end{itemize}
			At time 3, for each $1\leq j\leq m$:
			\begin{itemize}
				\item $s_{\tau_j}$ pays \Euro1 to $a_{\tau_j}$
				\item if $v_{\tau_j}$ appears in $L$ with positive (resp. negative) polarity then $a_{\tau_j}$ pays \Euro1 to the multiplier gadget which has a debt to $\neg x_{\tau_j}$ (resp. $x_{\tau_j}$).
			\end{itemize}
			Finally, having done the above, add any Boolean variable $v$ not appearing in $L$ to $L$ and proceed as above with these Boolean variables and assuming that they have positive polarity (note that the restriction of $X$ to these Boolean variables has no effect on whether $X$ satisfies $\phi$). What results is a perfect schedule. 	
		\end{proof}
	
		Given that our IDM instance $(G,D,A^0)$ of \textsc{Perfect Scheduling} can be built from the instance $\phi$ of \textsc{3-Sat-3} in polynomial time, we obtain our result for the PP variant.
		
		Consider now the AoN variant. As it happens, all of the above schedules are perfect schedules within the AoN variant and all associated reasoning still holds. Hence, we have our result for the AoN variant too.	
	\end{proof}

	\begin{thm}\label{thm:perfsched multiditree}
		The problem \textsc{Perfect Scheduling} is NP-complete for the PP and AoN variants even when we restrict to IDM instances $(G,D,A^0)$ for which: $G$ is a multiditree with diameter 6; all debts have monetary amount $\Euro1$; and there is a maximum of 6 debts between any pair of nodes. 
	\end{thm}
	
	\begin{proof}
		We show that, given an instance $\phi$ of \textsc{3-Sat-3}, involving $n$ Boolean variables and $m$ clauses, we can construct in polynomial-time an IDM instance $(G,D,A^0)$ where $G$ satisfies the criteria stated in the theorem so that $(G,D,A^0)$ admits a perfect schedule iff $\phi$ has a satisfying truth assignment. As usual, we restrict ourselves to those instances of \textsc{3-Sat-3} in which every literal appears at least once and at most twice and where no clause contains both a literal and its negation. In particular, we can label any appearance of any literal in $\phi$ as the first appearance or the second appearance.
		
		Our reduction is portrayed in Fig.~\ref{fig:treehard_node}. There is a distinct \emph{variable gadget} for each Boolean variable $v_i$, with $1\leq i\leq n$, and a distinct \emph{clause gadget}, for each clause $c_j$, with $1\leq j\leq m$. There is one node $r$. The debts in any variable gadget or involving the node $r$ are self-evident whereas the debts in a clause gadget are more involved.
		\begin{itemize}
			\item If the literal $v_i$ is in the clause $c_j$ and this appearance is the first (resp. second) appearance of $v_i$ in $\phi$ then there are:
			\begin{itemize}
				\item debts $1@10(i-1)+1$ (resp. $1@10(i-1)+3$) from $b_j$ to $a_j$ and from $e_j$ to $d_j$
				\item debts $1@10(i-1)+2$ (resp. $1@10(i-1)+4$) from $a_j$ to $b_j$ and from $d_j$ to $e_j$.
			\end{itemize}
			\item If the literal $\neg v_i$ is in the clause $c_j$ and this appearance is the first (resp. second) appearance of $\lnot v_i$ in $\phi$ then there are:
			\begin{itemize}
				\item debts $1@10(i-1)+6$ (resp. $1@10(i-1)+8$) from $b_j$ to $a_j$ and from $e_j$ to $d_j$
				\item debts $1@10(i-1)+7$ (resp. $1@10(i-1)+9$) from $a_j$ to $b_j$ and from $d_j$ to $e_j$.
			\end{itemize}
			\item If the clause $c_j$ has 3 literals (resp. 2 literals) then there are:
			\begin{itemize}
				\item two separate\footnote{By having only unit debts in the instance we have that every PP schedule is also an AoN schedule, and vice versa; for the PP variant we could instead have a single debt $2@[1, T]$.} debts $1@[1,T]$ (resp. a single debt $1@[1,T]$) from $a_j$ to $e_j$ and from $e_j$ to $a_j$.
			\end{itemize}
		\end{itemize}
		The legend in Fig.~\ref{fig:treehard_node} shows the time intervals corresponding to each literal, which we call the \emph{active windows} of the literals, and the occurrence of $x_1$ (resp. $\neg x_2$, $x_6$) in $c_j$ in Fig.~\ref{fig:treehard_node} is the first (resp. second, first) occurrence.
		
	\begin{figure}[!ht]
		\centering
		\begin{tikzpicture}
			\mynode{r}{10}{-2}{$r$}
			
			\begin{scope}
				\myassetnode{u}{0}{0}{$u_i$}{1}
				\mynode{y}{3}{0}{$y_i$}
				\mynode{w}{6}{0}{$w_i$}
				\path (u) edge[bend left] node[align=center, yshift=0pt] {$1@10i+1$\\$1@10i+6$} (y);
				\path (y) edge[bend left] node[align=center, xshift=-28pt] {$1@10i+4$\\$1@10i+9$} (uasset);	
				\path (y) edge[bend left] node[above] {$1@[10i+1,10i+9]$} (w);	
				\path (w) edge[bend left] node[below] {$1@[10i+1,10i+9]$} (y);	
				\node[rectangle, draw=blue, dotted, thick, inner sep=1cm, minimum width = 200pt, minimum height = 100pt] (gadget) at (3,-0.1) {};
				\node[state,rectangle, fill=white!100] (label) at (3,-1.85) {variable gadget for $v_i$};
			\end{scope}
			
			\begin{scope}[yshift=-10.5cm, xshift=0cm]
				\mynode{a}{6}{0.5}{$a_j$}
				\myassetnode{b}{0}{0.5}{$b_j$}{1}
				\mynode{e}{6}{5.5}{$e_j$}
				\mynode{d}{0}{5.5}{$d_j$}
				\path (b) edge[bend left] node[above, align=left] {1@1~1@18~1@51} (a);	
				\path (a) edge[bend left] node[below, align=left] {1@2~1@19~1@52} (basset);

				\path (a) edge[bend left, looseness=.2, out=10, in=170, align=center] node[left, yshift=15pt] {$1@[1,T]$\\$1@[1,T]$} (e);
				\path (e) edge[bend left, looseness=.2, out=10, in=170, align=center] node[right, yshift=15pt] {$1@[1,T]$\\$1@[1,T]$} (a);

				\path (e) edge[bend left] node[below, align=left] {1@1~1@18~1@51} (d);	
				\path (d) edge[bend left] node[above, align=left] {1@2~1@19~1@52} (e);
				\node[rectangle, draw=blue, dotted, thick, inner sep=1cm, minimum width = 200pt, minimum height = 250pt] (gadget) at (3,3.1) {};
				\node[state,rectangle, fill=white!100] (label) at (3,7.5) {clause gadget for $c_j= (x_1, \lnot x_2, x_6)$};
			\end{scope}

			
			\path (r) edge[bend left] node[right,yshift=-.2cm] {1@[1,T]} (e);
			\path (e) edge[bend left] node[right,yshift=-.2cm] {1@[1,T]} (r);
			
			\path (r) edge[bend left] node[align=left,yshift=30pt,xshift=35pt] {$1@[10i+1,$\\\hspace{20pt}$10i+9]$} (w);
			\path (w) edge[bend left] node[align=left,yshift=-10pt] {$1@[10i+1,$\\\hspace{20pt}$10i+9$]} (r);
			
			\node (species1) [shape=rectangle,draw=none, anchor=north west] at (6.8,-7.5) {
				\begin{tabular}{
						| >{\centering\arraybackslash}m{.9cm} 
						| >{\centering\arraybackslash}m{2.2cm} |}
					\hline
					literal & active window \\
					\hline
					$x_1$ & [1,4] \\
					\hline
					$\lnot x_1$ & [6,9] \\
					\hline
					$x_2$ & [11,14] \\
					\hline
					$\lnot x_2$ & [16,19] \\
					\hline
					\ldots & \ldots \\
					\hline
					$x_{i+1}$ & $[10i+1,10i+4]$ \\
					\hline
					$\lnot x_{i+1}$ & $[10i+6,10i+9]$ \\
					\hline
					\ldots & \ldots \\
					\hline
				\end{tabular}
			};
			
		\end{tikzpicture}
		\caption{A reduction from \textsc{3-Sat-3} to \textsc{Perfect Scheduling} restricted to multiditrees.}
		\label{fig:treehard_node}
	\end{figure}
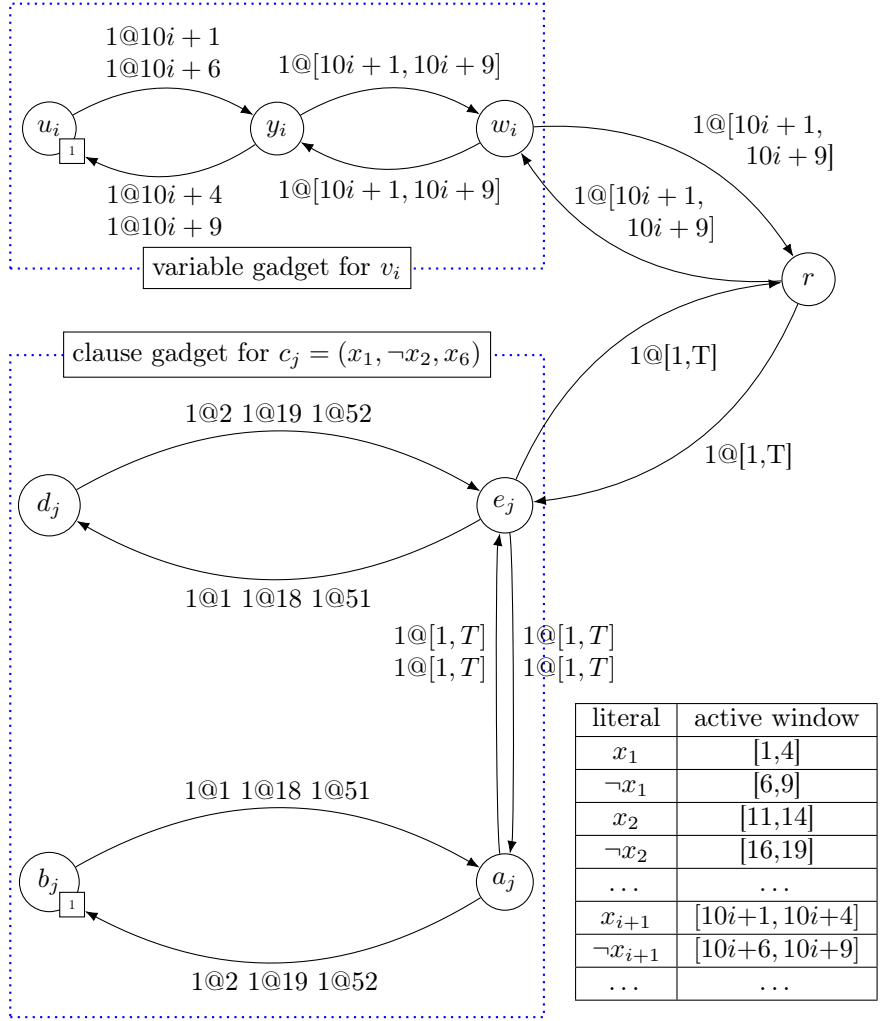

	\begin{clm}
		If $(G,D,A^0)$ is a yes-instance of \textsc{Perfect Scheduling} then  $\phi$ is a yes-instance of \textsc{3-Sat-3}.
	\end{clm}
	
	\begin{proof}
		Suppose that there is a perfect schedule $\sigma$ for $(G,D,A^0)$. Consider a clause gadget, corresponding to the clause $c_j$. There are 4 debts due within each active window corresponding to a literal in the clause. Suppose that the active windows are $[\alpha_1,\beta_1]$, $[\alpha_2,\beta_2]$ and $[\alpha_3,\beta_3]$, with $\beta_1<\alpha_2$ and $\beta_2<\alpha_3$ (we are assuming that our clause has 3 literals but the arguments for clauses with only 2 literals run analogously). If no payment has been received by $e_j$ from $r$ by time $\beta_1+1$ then: the \Euro1 at $b_j$ within the clause gadget must have been used in $\sigma$:
		\begin{itemize}
			\item to pay the debts of $b_j$ to $a_j$, $a_j$ to $e_j$ and $e_j$ to $d_j$ at time $\alpha_1$
			\item to pay the debts of $d_j$ to $e_j$, $e_j$ to $a_j$ and $a_j$ to $b_j$ at time $\alpha_1+1$,
		\end{itemize}
		if the appearance of the corresponding literal is the first; or
		\begin{itemize}
			\item to pay the debts of $b_j$ to $a_j$, $a_j$ to $e_j$ and $e_j$ to $d_j$ at time $\alpha_1+2$
			\item to pay the debts of $d_j$ to $e_j$, $e_j$ to $a_j$ and $a_j$ to $b_j$ at time $\alpha_1+3 = \beta_1$,
		\end{itemize}
		if the appearance of the corresponding literal is the second (note that the payments made towards the debts of $a_j$ to $e_j$ and $e_j$ to $a_j$ are only partial payments). We have an analogous situation as regards the interval $[\alpha_2,\beta_2]$ when no payment has been received by $e_j$ from $r$ by time $\beta_2+1$. However, if no payment has been received by $e_j$ from $r$ by time $\beta_3+1$ then we obtain a contradiction as the debts from $a_j$ to $e_j$ and from $e_j$ to $a_j$ will have been fully paid and consequently $e_j$ will be bankrupt at time $\beta_3$. Hence, within $\sigma$, there must be a payment of \Euro1 received by $e_j$ from $r$ and this is the only payment made from $r$ to $e_j$. Furthermore, there can be no payment from $e_j$ to $r$ strictly before the time at which a payment is made from $r$ to $e_j$ and if a payment is made from $e_j$ to $r$ at the same time that a payment is made from $r$ to $e_j$ then this can be interpreted as no resource leaving or entering the clause gadget, with the external resource satisfying both debts involving $r$ and $e_j$, which cannot be the case.
		
		Consider now a variable gadget, corresponding to the variable $v_i$. At times $t \in [0,10i] \cup [10i+4, 10i+5] \cup [10i+9, T]$ there must be cash assets of (at least) \Euro1 at node $u_i$. So, outside the active windows of the literals associated with the variable $v_i$, there must be at least \Euro1 of cash assets at the node $u_i$. Also, note that the \Euro1 originating at node $u_i$ can only leave and return to the variable gadget at some times in $[10i+1,10i+5]$ or in $[10i+6,10i+9]$ but not both. Consequently, at any time, there is at most \Euro1 of resource that originated within a variable gadget outside that particular variable gadget. 
		
		Consider the time interval $[1,9]$. Within this time interval, the \Euro1 originating at $u_1$ is the only euro that might be possibly `outside' the variable gadget corresponding to $v_1$. Call this euro $E$. Suppose $E$ leaves its variable gadget at some time in $[1,4]$. It needs to be back at $y_1$ by time 4 (and will never leave the variable gadget again). Suppose that $E$ is paid from $r$ to $e_j$, for some $j$, within the time interval $[1,4]$. If the literal $x_1$ is not in $c_j$ then all debts involving only $a_j$ and $b_j$ and all debts involving only $d_j$ and $e_j$ will be due at some time outside $[1,4]$ and so $E$ is of no use to the clause gadget for $c_j$. If $E$ is paid from $r$ to $e_j$ at time 1 (resp. by time 3) and the literal $x_1$ is in clause $c_j$ as the first (resp. second) appearance then $E$ can be used to pay the debts from $e_j$ to $d_j$ at time 1 (resp. time 3) and from $d_j$ to $e_j$ at time 2 (resp. time 4). Note that $E$ cannot be used to pay the debt from $a_j$ to $b_j$ at time 2 or time 4 as it would not get back to its variable gadget by time 4. Given that $E$ leaves the clause gadget by time 4, no more resource either enters or leaves the clause gadget. Alternatively, suppose that $E$ leaves the variable gadget corresponding to $v_1$ at some time in $[6,9]$. Exactly the same argument can be made as that above except with respect to the literal $\neg x_1$ appearing in some clause or other.
		
		Let us continue with the time interval $[11,19]$ and the \Euro1 originating at $u_2$. An analogous argument to that above holds. Note that all clause gadgets that were previously `visited' by $E$, above, are now `closed' in that they accept or eject no further resource. Indeed, an analogous argument to that above holds for every euro originating at some node $u_i$. As $\sigma$ is a perfect schedule, every clause gadget must be visited by some euro originating in some variable gadget and the particular literal corresponding to the active window during which the euro left its variable gadget must appear in the clause. Define the truth assignment $X$ via: $X(v_i) = True$ (resp. $False$) if the euro from the variable gadget corresponding to $v_i$ leaves its variable gadget during the active window $[10i+1,10i+4]$ (resp. $[10i+6,10i+9]$) and visits a clause gadget, with any Boolean variables $v_i$ for which $X(v_i)$ has not been defined such that $X(v_i)$ is defined arbitrarily. Given the above discussion, it should be clear that $X$ satisfies $\phi$.
	\end{proof}
		
	\begin{clm}
		If $\phi$ is a yes-instance of \textsc{3-Sat-3} then $(G,D,A^0)$ is a yes-instance of \textsc{Perfect Scheduling}.
	\end{clm}
	
	\begin{proof}
		Suppose that $X$ is a satisfying truth assignment for $\phi$. Define the schedule $\sigma$ as follows.
		
		If $X(v_i) = True$ then at time $10(i-1)+1$, the euro originating at $u_i$ is paid from $u_i$ to $y_i$ to $w_i$ to $r$.
		\begin{itemize}
			\item If the clause $c_j$ containing the first appearance of the literal $x_i$ exists and has not been visited by some euro originating within a variable gadget then at time $10i+1$, the euro is paid from $r$ to $e_j$ and on to $d_j$. At time $10i+2$, the same euro is paid from $d_j$ to $b_j$ and on to $r$.
			\item If there is no clause containing the literal $x_i$ or if the clause gadget corresponding to the first appearance of $x_i$ has already been visited in the schedule $\sigma$ by some euro originating within a variable gadget, then do nothing.
			\item If the clause $c_j$ containing the second appearance of the literal $x_i$ exists and has not been visited by some euro originating within a variable gadget then at time $10i+3$, the euro is paid from $r$ to $e_j$ and on to $d_j$. At time $10i+4$, the same euro is paid from $d_j$ to $b_j$ and on to $r$.
			\item If there is no clause containing the literal $x_i$ or if there is no second appearance of the literal $x_i$ or if the clause gadget corresponding to the second appearance of $x_i$ has already been visited in the schedule $\sigma$ by some euro originating within a variable gadget, then do nothing.
			\item Our euro at $r$ is paid at time $10i+4$ from $r$ to $w_i$ to $y_i$ to $u_i$.
		\end{itemize}
		If $X(v_i) = False$ then at time $10i+6$, the euro originating at $u_i$ is paid from $u_i$ to $y_i$ to $w_i$ to $r$.
		\begin{itemize}
			\item If the clause $c_j$ containing the first appearance of the literal $\neg x_i$ exists and has not been visited by some euro originating within a variable gadget then at time $10i+6$, the euro is paid from $r$ to $e_j$ and on to $d_j$. At time $10i+7$, the same euro is paid from $d_j$ to $e_j$ and on to $r$.
			\item If there is no clause containing the literal $\neg x_i$ or if the clause gadget corresponding to the first appearance of $\neg x_i$ has already been visited in the schedule $\sigma$ by some euro originating within a variable gadget, then do nothing.
			\item If the clause $c_j$ containing the second appearance of the literal $\neg x_i$ exists and has not been visited by some euro originating within a variable gadget then at time $10i+8$, the euro is paid from $r$ to $e_j$ and on to $d_j$. At time $10i+9$, the same euro is paid from $d_j$ to $e_j$ and on to $r$.
			\item If there is no clause containing the literal $\neg x_i$ or if there is no second appearance of the literal $\neg x_i$ or if the clause gadget corresponding to the second appearance of $\neg x_i$ has already been visited in the schedule $\sigma$ by some euro originating within a variable gadget, then do nothing.
			\item Our euro at $r$ is paid at time $10i+9$ from $r$ to $w_i$ to $y_i$ to $u_i$.
		\end{itemize}
		Within any clause gadget corresponding to $c_j$, the euro originating at $b_j$ is used to pay the debts corresponding to the literal not addressed by the euro from a variable gadget. It can easily be seen that $\sigma$ is valid and a perfect schedule.
	\end{proof}
	
 \noindent The result follows, given that the construction of $(G,D,A^0)$ can clearly be undertaken in polynomial-time.
\end{proof}

	In all the above results, the input IDM instance is allowed to have unlimited (i.e.,~unbounded) total initial assets which might be unrealistic in practically relevant financial systems. We now show that even in the highly restricted case where there is just \Euro1 in initial external assets in total, \textsc{Perfect Scheduling} still remains NP-complete in the AoN and PP variants.

\begin{thm}\label{thm:perfschedhampath}
	The problem \textsc{Perfect Scheduling} is NP-complete in the AoN and PP variants even when the total value of all initial external assets in any instance is \Euro1.
\end{thm}

\begin{proof}
	The following proof applies to both the AoN and PP variants. Our reduction is a reduction from the problem \textsc{Sourced Hamiltonian Path} defined as follows.\smallskip
	
	\noindent\underline{\textsc{Sourced Hamiltonian Path}}
	\begin{description}
		\item Instance: a digraph $H$ and a vertex $x$
		\item Yes-instance: there exists a Hamiltonian path in $H$ with source $x$.
	\end{description}
	This problem can be trivially shown to be NP-complete by reducing from the standard problem of deciding whether a digraph has a Hamiltonian path \cite{Karp72}.
	
	Let $H$ be some digraph on the $n$ vertices $\{x_i:1\leq i\leq n\}$ and w.l.o.g. let $x=x_1$. In order to describe our reduction, we first describe a gadget, namely the \emph{at-least-once} gadget. We have one of these gadgets for each vertex of $H$ and we refer to the gadget corresponding to the vertex $x_i$ of $H$ as at-least-once$(i)$. Our at-least-once gadget can be defined as in Fig.~\ref{fig:atleastonce} where the value $T$ is defined as $2n+1$. Note that the gadget is exactly the nodes and debts within the blue dotted box and so contains its own copies of nodes $v_L$, $v_C$ and $v_R$ and the $4n-1$ debts involving them. The nodes $v^\prime_R$ and $v^{\prime\prime}_L$ are not part of the gadget but are nodes in other gadgets as we now explain.

	\begin{figure}[!ht]
	\centering
	\scalebox{0.95}{
		\begin{tikzpicture}
			\begin{scope}
				\mynode{vl}{0.3}{0}{$v_L$}
				\mynode{vc}{4}{0}{$v_C$}
				\mynode{vr}{7.7}{0}{$v_R$}
				
				\mynode{vrprime}{-1.8}{0}{$v^\prime_R$}
				\mynode{vlprimeprime}{9.8}{0}{$v^{\prime\prime}_L$}

				\path (vrprime) edge[bend left] node[above] {$1@[1,T]$} (vl);
				\path (vr) edge[bend left] node[above] {$1@[1,T]$} (vlprimeprime);
				
				\path (vlprimeprime) edge[bend left] node[below] {$1@T$} (vr);
				\path (vl) edge[bend left] node[below] {$1@T$} (vrprime);
				
				\path (vl) edge[bend left] node[above] {$1@1~1@3~\ldots~1@T-2$} (vc);	
				\path (vc) edge[bend left] node[align=center, below] {$1@[1,T-1]$\\ $\ldots$\\$1@[1,T-1]$\\ \scriptsize($n-1$ times)\normalsize} (vl);	
				
				\path (vc) edge[bend left] node[above, xshift=-3pt] {$1@2~1@4~\ldots~1@T-1$} (vr);	
				\path (vr) edge[bend left] node[align=center, below] {$1@[1,T-1]$\\ $\ldots$\\$1@[1,T-1]$\\ \scriptsize($n-1$ times)\normalsize} (vc);	
				
				\path (vr) edge[bend right, looseness=.8,out=-90,in=-90] node[above] {$1@T$} (vl);

				\node[draw=blue,dotted, thick,inner sep=0cm, minimum height=5.2cm, minimum width=8.3cm,yshift=5pt] at (4,0) (gadget) {};
				
				\node[state,rectangle, fill=white!100] (label) at (2,2.8) {at-least-once($i$)};
				
			\end{scope}
			
	\end{tikzpicture}}
	\caption{An at-least-once gadget. Note that as $T=2n+1$ there are, $n$ \Euro 1 debts owed by $v_L$ to $v_C$ and by $v_C$ to $v_R$.}
	\label{fig:atleastonce}
	\end{figure}
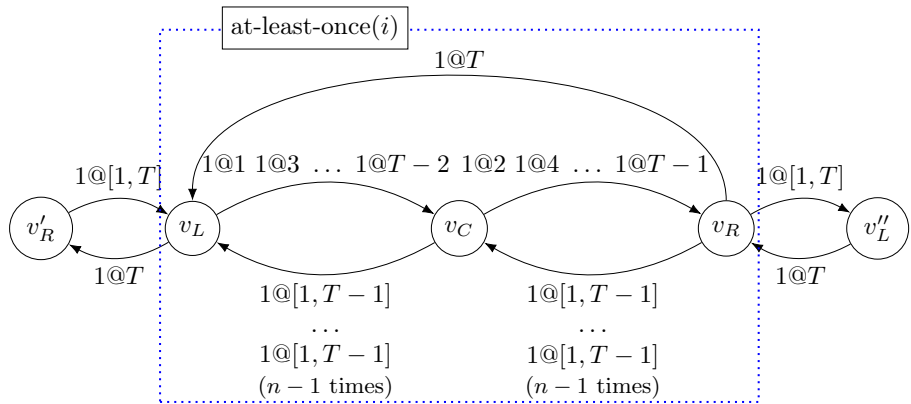

	Set $T=2n+1$. Our IDM instance $(G,D,A^0)$ can be defined as follows:
	\begin{itemize}
		\item there is the at-least-once$(i)$ gadget, for $1\leq i\leq n$
		\item for every edge $(x_i,x_j)$ of $H$, there is a debt of \Euro1 from $v_R$ of at-least-once$(i)$ to $v_L$ of at-least-once$(j)$ to be paid in the interval $[1,T]$ and a debt of \Euro1 from $v_L$ of at-least-once$(j)$ to $v_R$ of at-least-once$(i)$ to be paid at time $T$
		\item all nodes have initial external assets of 0 except for node $v_L$ of at-least-once$(1)$ which has initial external assets of 1.
	\end{itemize}
	We refer to the single euro of initial external assets as the \emph{initial euro}. The IDM instance $(G,D,A^0)$ can clearly be constructed in time polynomial in $n$.

	\begin{clm}\label{clm1:perfschedhampath}
		If $(G,D,A^0)$ is a yes-instance of \textsc{Perfect Scheduling} then $(H,x)$ is a yes-instance of \textsc{Sourced Hamiltonian Path}.
	\end{clm}
		
	\begin{proof}
		Note that strictly prior to time $T$, the only payment-cycles that can exist within $G$ involve the two nodes $v_L$ and $v_C$ of some at-least-once gadget or the two nodes $v_C$ and $v_R$ of some at-least-once gadget. Note also that within some gadget there are $n$ debts of \Euro1 from $v_L$ to $v_C$ needing to be satisfied and $n-1$ debts of \Euro1 from $v_C$ to $v_L$. An analogous statement can be made as regards $v_C$ and $v_R$. Consequently, in order to satisfy all debts involving $v_L$ and $v_C$ within some gadget, at some odd time in $[1,T-1]$, the initial euro must be within that gadget so as to satisfy some debt from $v_L$ to $v_C$ (by moving from $v_L$ to $v_C$). Similarly, at some even time in $[1,T-1]$, the initial euro must be within the gadget so as to satisfy some debt from $v_C$ to $v_L$ (by moving from $v_C$ to $v_R$). Moreover, at any time in $[1,T-1]$, the initial euro can only satisfy at most one of the debts mentioned above. Hence, given that $T=2n+1$, at any time in $[1,T-1]$, the initial euro must be satisfying exactly one of the above debts.
		
		Suppose that the initial euro satisfies some debt from $v_L$ to $v_C$ in some at-least-once gadget at time $t$. As the initial euro needs to satisfy one of the above debts at time $t+1$, we need that at time $t+1$ the initial euro satisfies a debt from $v_C$ to $v_R$ in the same gadget. Also, it cannot be the case that a debt from $v_C$ to $v_R$ in some gadget is satisfied by the initial euro before the euro satisfies some debt from $v_L$ to $v_R$ in the same gadget. As the initial euro starts at $v_L$ in at-least-once$(1)$, our schedule must be such that the initial euro `moves' through the at-least-once gadgets corresponding to every node, entering at the node $v_L$ and exiting at the node $v_R$. Consequently, its path within $G$ corresponds to a path $x=x_1, x_2, \ldots, x_n$ in $H$ where every node on this path is distinct and where there is a directed edge from node $x_i$ to $x_{i+1}$, for $1\leq i\leq n-1$; that is, a Hamiltonian path in $H$ with source $x$.
	\end{proof}

	\begin{clm}\label{clm2:perfschedhampath}
	If $(H,x)$ is a yes-instance of \textsc{Sourced Hamiltonian Path} then $(G,D,A^0)$ is a yes-instance of \textsc{Perfect Scheduling}.
	\end{clm}

	\begin{proof}
		Let $x=x_1,x_2,\ldots,x_n$ be a Hamiltonian path $P$ in the digraph $H$. Consider the following schedule $\sigma$:
		\begin{itemize}
			\item the initial euro is used so as to pay the following debts:
			\begin{itemize}
				\item \Euro1 at time $2i-1$ from $v_L$ to $v_C$ in at-least-once$(x_i)$ and \Euro1 at time $2i$ from $v_C$ to $v_R$ in at-least-once$(x_i)$, for $1\leq i\leq n$
				\item \Euro1 at time $2i$ from $v_R$ in at-least-once$(x_i)$ to $v_L$ in at-least-once $(x_{i+1})$, for $1\leq i\leq n-1$
			\end{itemize}
			\item for any $v_L$ and $v_C$ within some at-least-once gadget and at any odd time $t<T$ when a debt is not being paid from $v_L$ to $v_C$ using the initial euro, there is a payment-cycle consisting of payments of \Euro1 from $v_L$ to $v_C$ and of \Euro1 from $v_C$ to $v_L$
			\item for any $v_C$ and $v_R$ within some at-least-once gadget and at any even time $t< T$ when a debt is not being paid from $v_C$ to $v_R$ using the initial euro, there is a payment-cycle consisting of payments of \Euro1 from $v_C$ to $v_R$ and of \Euro1 from $v_R$ to $v_C$
			\item for any edge $(x_i,x_j)$ of $H$ that does not feature in the Hamiltonian path $P$, there is a payment-cycle consisting of payments at time $T$ of \Euro1 from $v_R$ in at-least-once$(x_i)$ to $v_L$ in at-least-once$(x_j)$ and of \Euro1 from $v_L$ in at-least-once$(x_j)$ to $v_R$ in at-least-once$(x_i)$
			\item the initial euro is used so as to pay the following debts:
			\begin{itemize}
				\item \Euro1 at time $T$ from $v_R$ to $v_L$ in at-least-once$(x_i)$, for $1\leq i\leq n$
				\item \Euro1 at time $T$ from $v_L$ in at-least-once$(x_i)$ to $v_R$ in at-least-once$(x_{i-1})$, for $2\leq i\leq n$.
			\end{itemize}
		\end{itemize}
		The `path' taken by the initial euro within $(G,D,A^0)$ can be visualized as in Fig.~\ref{fig:chain} where the debt arrows are tagged with the time of payment. It is clear that $\sigma$ is valid and a perfect schedule.
	\end{proof}
	Our main result follows.
\end{proof}

		\begin{figure}[!ht]
			\centering
			\scalebox{0.85}{%
				\begin{tikzpicture}
					
					\node[draw,circle]       (v1L) at ({360/24 * 23}:4cm) {$v_{L}$};
					\node[draw,circle]       (v1C) at ({360/24 * 22}:4cm) {$v_{C}$};
					\node[draw,circle,thick] (v1R) at ({360/24 * 21}:4cm) {$v_{R}$};
					
					\node (dummy1)                  at ({360/24 * 20}:4cm) {};
					
					\node[draw,circle]       (v2L) at ({360/24 * 19}:4cm) {$v_{L}$};
					\node[draw,circle]       (v2C) at ({360/24 * 18}:4cm) {$v_{C}$};
					\node[draw,circle]       (v2R) at ({360/24 * 17}:4cm) {$v_{R}$};
					
					\node (dummy2)                  at ({360/24 * 16}:4cm) {};
					
					\node[draw,circle]       (v3L) at ({360/24 * 15}:4cm) {$v_{L}$};
					\node[draw,circle]       (v3C) at ({360/24 * 14}:4cm) {$v_{C}$};
					\node[draw,circle]       (v3R) at ({360/24 * 13}:4cm) {$v_{R}$};
					
					\node (dummy3)                  at ({360/24 * 12}:4cm) {};
					
					\node[draw,circle]       (v4L) at ({360/24 * 11}:4cm) {$v_{L}$};
					\node[draw,circle]       (v4C) at ({360/24 * 10}:4cm) {$v_{C}$};
					\node[draw,circle]       (v4R) at ({360/24 *  9}:4cm) {$v_{R}$};
					
					\node (dummy4)                  at ({360/24 *  8}:4cm) {};
					
					\node[draw,circle]       (v5L) at ({360/24 *  7}:4cm) {$v_{L}$};
					\node[draw,circle]       (v5C) at ({360/24 *  6}:4cm) {$v_{C}$};
					\node[draw,circle]       (v5R) at ({360/24 *  5}:4cm) {$v_{R}$};
					
					\node (dummy5)                  at ({360/24 *  4}:4cm) {$\cdots$};
					
					\node[draw,circle]       (vnL) at ({360/24 *  3}:4cm) {$v_{L}$};
					\node[draw,circle]       (vnC) at ({360/24 *  2}:4cm) {$v_{C}$};
					\node[draw,circle]       (vnR) at ({360/24 *  1}:4cm) {$v_{R}$};
					
					\node (dummyn)                  at ({360/24 *  0 }:4cm) {};
					
					\node[draw=blue,dotted, thick,fit=(v1L)(v1R), inner sep=0cm] (gadget) {};
					\node[draw=blue,dotted, thick,minimum height = 60pt, minimum width = 80pt, inner sep=0cm] at (0,-4) (gadget) {};
					\node[draw=blue,dotted, thick,fit=(v3L)(v3R), inner sep=0cm] (gadget) {};
					\node[draw=blue,dotted, thick,fit=(v4L)(v4R), inner sep=0cm] (gadget) {};
					\node[draw=blue,dotted, thick,minimum height = 60pt, minimum width = 80pt, inner sep=0cm] at (0,4) (gadget) {};
					\node[draw=blue,dotted, thick,fit=(vnL)(vnR), inner sep=0cm] (gadget) {};
					
					\node[state,rectangle, fill=white!100] (label) at (5.4,-2.7) {at-least-once($1$)};
					\node[state,rectangle, fill=white!100] (label) at (0,-5.1) {at-least-once($2$)};
					\node[state,rectangle, fill=white!100] (label) at (-5.4,-2.7) {at-least-once($3$)};
					\node[state,rectangle, fill=white!100] (label) at (-5.4,2.7) {at-least-once($4$)};
					\node[state,rectangle, fill=white!100] (label) at (0,5.1) {at-least-once($5$)};
					\node[state,rectangle, fill=white!100] (label) at (5.4,2.3) {at-least-once($n$)};

					
					\path (v1R) edge[bend left] node {$2$} (v2L);
					\path (v2L) edge[bend left] node {$3$} (v2C);
					\path (v2C) edge[bend left] node {$4$} (v2R);
					
					\path (v2R) edge[bend left] node {$4$} (v3L);
					\path (v3L) edge[bend left] node {$5$} (v3C);
					\path (v3C) edge[bend left] node {$6$} (v3R);
					
					\path (v3R) edge[bend left] node {$6$} (v4L);
					\path (v4L) edge[bend left] node {$7$} (v4C);
					\path (v4C) edge[bend left] node {$8$} (v4R);
					
					\path (v4R) edge[bend left] node {$8$} (v5L);
					\path (v5L) edge[bend left] node {$9$} (v5C);
					\path (v5C) edge[bend left] node {$10$} (v5R);
					
					\path (v5R) edge[bend left] node {$10$} (dummy5);
					
					\path (dummy5) edge[bend left] node {$2n-2$} (vnL);
					\path (vnL) edge[bend left] node {$2n-1$} (vnC);
					\path (vnC) edge[bend left] node[xshift=5pt, yshift=-5pt] {$2n$} (vnR);
					
					\path (v1L) edge[bend left] node {$1$} (v1C);
					\path (v1C) edge[bend left] node {$2$} (v1R);
					
					
					\path (v1R) edge[bend left] node {$2n+1$} (v1L);
					\path (vnR) edge[bend left] node {$2n+1$} (vnL);
					\path (vnL) edge[bend left] node[yshift=-5pt, xshift=5pt] {$2n+1$} (dummy5);
					\path (dummy5) edge[bend left] node[yshift=-5pt, xshift=15pt] {$2n+1$} (v5R);
					\path (v5R) edge[bend left] node {$2n+1$} (v5L);
					\path (v5L) edge[bend left] node {$2n+1$} (v4R);
					\path (v4R) edge[bend left] node {$2n+1$} (v4L);
					\path (v4L) edge[bend left] node {$2n+1$} (v3R);
					\path (v3R) edge[bend left] node {$2n+1$} (v3L);
					\path (v3L) edge[bend left] node {$2n+1$} (v2R);
					\path (v2R) edge[bend left] node {$2n+1$} (v2L);
					\path (v2L) edge[bend left] node {$2n+1$} (v1R);
					
				\end{tikzpicture}
			}
			\caption{The path taken by the initial euro straightforwardly corresponds to a sourced Hamiltonian path in the original graph.}\label{fig:chain}
		\end{figure}
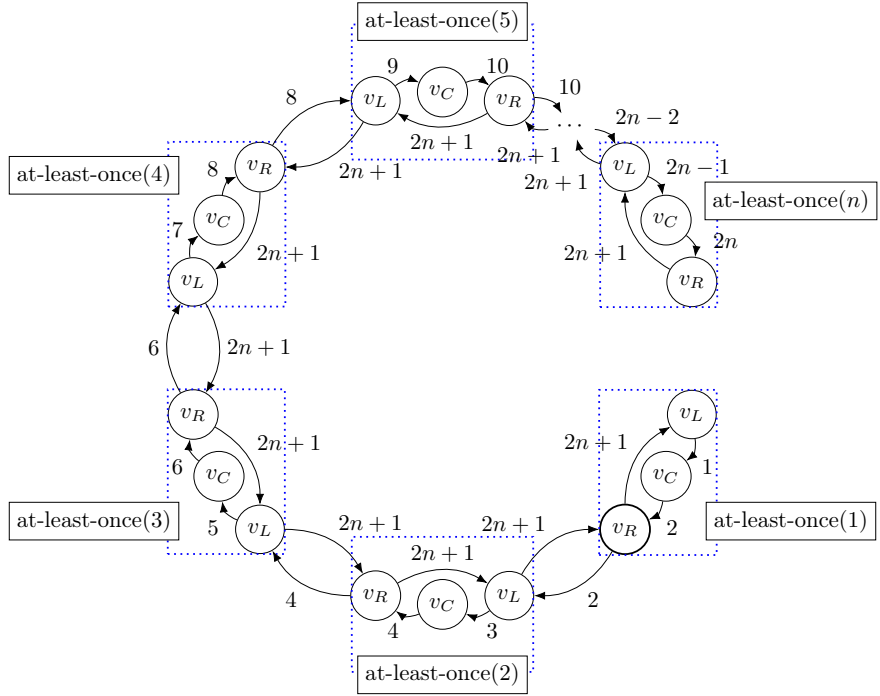
		
		Of course, one can obtain additional restrictions on the structure of the IDM instances for \textsc{Perfect Scheduling} in Theorem~\ref{thm:perfschedhampath} by looking at NP-completeness results relating to \textsc{Sourced Hamiltonian Path} on restricted digraphs; however, we have refrained from doing so (as nothing of any significance emerges).
		
		We can constraint the digraph $G$ of an instance $(G,D,A^0)$ of \textsc{Perfect Scheduling} even further in the AoN variant; indeed, so that it is always a directed path of length 3. The price we pay is that the initial external assets are potentially large.
		
		\begin{thm}\label{thm:AoN3path}
			Consider the problem \textsc{Perfect Scheduling} restricted so that every instance $(G,D,A^0)$ is such that $G$ is a directed path of bounded length.\begin{itemize}
				\item[(\emph{a})] If, further, $T$ is restricted to be 2 then the resulting problem is weakly NP-complete in the AoN variant.
				\item[(\emph{b})] If there are no restrictions on $T$ then the resulting problem is strongly NP-complete in the AoN variant.
			\end{itemize}
		\end{thm}
		
		\begin{proof}
			Consider (\emph{a}). We reduce from the problem \textsc{Partition} defined as follows (and proven in \cite{Karp72} to be weakly NP-complete).\smallskip
			
			\noindent\underline{\textsc{Partition}}
			\begin{description}
				\item Instance: a multi-set of integers $S=\{a_1,a_2,\ldots,a_n\}$ with $sum(S)=2k$
				\item Yes-instance: there exists a partition of $S$ into two subsets $S_1$ and $S_2$ such that $sum(S_1)=sum(S_2)=k$.
			\end{description}
			In general, an instance $S=\{a_1,a_2,\ldots,a_n\}$ has size $nb$ where $b$ is the least number of bits required to express any of the integers of $S$ in binary.
			
			Let $S$ be an instance of \text{Partition} of size $nb$. Consider the IDM instance $(G,D,A^0)$ in Fig.~\ref{fig:AoN_partition}. Note that the time taken to construct $(G,D,A^0)$ from $S$ is polynomial in $nb$.
						
			\begin{figure}[!ht]
				\centering
				\begin{tikzpicture}
					\myassetnode{s}{0}{0}{$s$}{$2k$};
					\mynode{v}{3}{0}{$v$};
					\mynode{w}{6}{0}{$w$};
					\mynode{x}{9}{0}{$x$};
					
					\path (v) edge node[above, align=center] {$a_1@[1,2]$\\$a_2@[1,2]$\\...\\$a_{n}@[1,2]$} (w);
					\path (s) edge node[align=center] {$k@1$\\$k@2$} (v);
					\path (w) edge node[align=center] {$k@1$\\$k@2$} (x);
					
				\end{tikzpicture}
				
				\caption{An IDM instance encoding an instance $\{a_1, \ldots, a_n\}$ of \textsc{Partition} with sum $2k$.}
				\label{fig:AoN_partition}
			\end{figure}
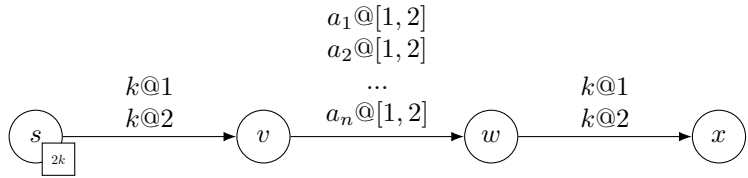
			
		\begin{clm}
		If $(G,D,A^0)$ is a yes-instance of \textsc{Perfect Scheduling} then $S$ is a yes-instance of \textsc{Partition}.
		\end{clm}
	
		\begin{proof}
		Suppose that there is a valid schedule $\sigma$ for $(G,D,A^0)$ that is a perfect schedule. It must be the case that the total amount paid by $s$ at time 1 is \Euro$k$ and that this payment is immediately paid by $v$ to $w$ and from $w$ to $x$ at time 1. As we are in the AoN variant, it must be the case that the sum of a subset of integers of $S$ amounts to $k$. An analogous argument applies to the payments made at time 2 and the remainder of the integers in $S$. the claim follows.
		\end{proof}
	
		\begin{clm}
		 	If $S$ is a yes-instance of \textsc{Partition} then $(G,D,A^0)$ is a yes-instance of \textsc{Perfect Scheduling}.
		\end{clm}
		
		\begin{proof}
			Suppose that $S_1$ and $S_2$ is a partition of $S$ such that $\sumop(S_1)=\sumop(S_2)=k$. Define the schedule $\sigma$ so that:
			\begin{itemize}
				\item at time 1: $s$ pays \Euro$k$ to $v$; if $a_i\in S_1$, for $1\leq i\leq n$, then $v$ pays \Euro$a_i$ to $w$; and $w$ pays \Euro$k$ to $x$
				\item at time 2: $s$ pays \Euro$k$ to $v$; if $a_i\in S_2$, for $1\leq i\leq n$, then $v$ pays \Euro$a_i$ to $w$; and $w$ pays \Euro$k$ to $x$.
			\end{itemize}
			The schedule $\sigma$ is a valid perfect schedule.
		\end{proof}
		\noindent The proof of (\emph{a}) follows.
		Now consider (\emph{b}). We reduce from the strongly NP-complete problem \textsc{3-Partition} defined as follows (see \cite{Garey-Johnson79}).\smallskip
		
		\noindent\underline{\textsc{3-Partition}}
		\begin{description}
			\item Instance: a multi-set of integers $S=\{a_1,a_2,\ldots,a_{3m}\}$, for some $m\geq 1$, with $sum(S)=mk$
			\item Yes-instance: there exists a partition of $S$ into $m$ triplets $S_1,S_2,\ldots S_m$ such that $sum(S_i)=k$, for each $1\leq i \leq m$.
		\end{description}
		In general, an instance $S=\{a_1,a_2,\ldots,a_{3m}\}$ has size $mb$ where $b$ is the least number of bits required to express any of the integers of $S$ in binary.
		
		Let $S$ be an instance of \textsc{3-Partition} of size $mb$. By multiplying all integers by 4 if necessary, we may assume that every integer of $S$ is divisible by 4 as is $k$. Consider the IDM instance $(G,D,A^0)$ in Fig.~\ref{fig:AoN_partition}. Note that the time taken to construct $(G,D,A^0)$ from $S$ is polynomial in $mb$.
			
		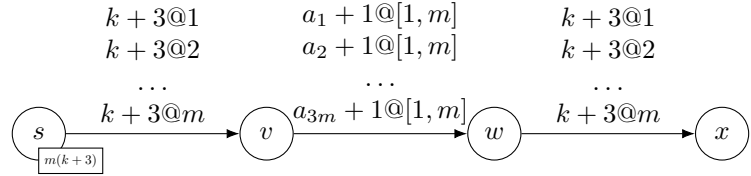
\begin{figure}[!ht]
			\centering
			\begin{tikzpicture}
				\myassetnodenew{s}{0}{0}{$s$}{$m(k+3)$};
				\mynode{v}{3}{0}{$v$};
				\mynode{w}{6}{0}{$w$};
				\mynode{x}{9}{0}{$x$};
				
				\path (v) edge node[above, align=center] {$a_1+1@[1,m]$\\$a_2+1@[1,m]$\\$\ldots$\\$a_{3m}+1@[1,m]$} (w);
				\path (s) edge node[align=center] {$k+3@1$\\$k+3@2$\\$\ldots$\\$k+3@m$} (v);
				\path (w) edge node[align=center] {$k+3@1$\\$k+3@2$\\$\ldots$\\$k+3@m$} (x);
				
			\end{tikzpicture}
			
			\caption{An IDM instance showing the reduction from \textsc{3-Partition} to \textsc{AoN Perfect Scheduling}.}
			\label{fig:AoN_3partition}
		\end{figure}
			
		\begin{clm}
			If $(G,D,A^0)$ is a yes-instance of \textsc{Perfect Scheduling} then $S$ is a yes-instance of \textsc{3-Partition}.
		\end{clm}
		
		\begin{proof}
			Suppose that there is a valid schedule $\sigma$ for $(G,D,A^0)$ that is a perfect schedule. It must be the case that the total amount paid by $s$ at time $i$, for every $1\leq i\leq m$, is \Euro$k+3$ and that this payment is immediately paid by $v$ to $w$ and by $w$ to $x$. Suppose that the payment at time $i$ by $v$ to $w$ pays at least 4 of the debts due. So, there exists another time $j$, say, where the payments made by $v$ to $w$ pay at most 2 debts. So, we have that either $a_\alpha+a_\beta+2=k+3$ or $a_\alpha+1=k+3$, for some $1\leq \alpha\neq \beta \leq 3m$. This yields a contradiction as the right-hand sides of these equations are equivalent to 3 modulo 4 whereas the left-hand sides are not. So, at any time $i$, for $1\leq i\leq m$, exactly three debts are paid by $v$ to $w$ at time $i$. If $1\leq \alpha,\beta,\gamma\leq 3m$ are distinct so that debts of monetary amounts $a_\alpha+1$, $a_\beta+1$ and $a_\gamma+1$ are paid by $v$ to $w$ at some time then $a_\alpha+1+a_\beta+1+a_\gamma+1 = k+3$; that is, $a_\alpha+a_\beta+a_\gamma = k$. So, we have a yes-instance of \textsc{3-Partition}. The claim follows.
		\end{proof}
		
		\begin{clm}
			If $S$ is a yes-instance of \textsc{3-Partition} then $(G,D,A^0)$ is a yes-instance of \textsc{Perfect Scheduling}.
		\end{clm}
		
		\begin{proof}
			Suppose that $S$ can be partitioned into triplets so that the sum of the integers in each triplet is $k$; so, suppose that $a_{\alpha_i}+a_{\beta_i}+a_{\gamma_i} = k$, for each $1\leq i\leq m$, where $S = \{a_{\alpha_i},a_{\beta_i}, a_{\gamma_i} : 1\leq i\leq m\}$. Define the following schedule $\sigma$: at time $i$, for each $1\leq i\leq m$, $s$ pays \Euro$k+3$ to $v$; $v$ pays \Euro$a_{\alpha_i}+a_{\beta_i}+a_{\gamma_i}+3 = k+3$ to $w$; and $w$ pays \Euro$k+3$ to $x$. The schedule $\sigma$ is valid and a perfect schedule. The claim follows.
		\end{proof}
		\noindent The main result follows.\end{proof}

	\subsection{Hardness results for \textsc{Bankruptcy Maximization}}\label{subsec:hardnessBmax}
		
		We now turn to \textsc{Bankruptcy Maximization}.

	\begin{thm}
		\label{thm:bankmax}
		The problem \textsc{Bankruptcy Maximization} is NP-complete in the AoN, PP and FP variants even when for an instance $(G,D,A^0)$: $T=2$; $G$ is a directed acyclic graph with out-degree at most $2$, in-degree at most 3; all monetary debts are at most \Euro2 per edge; and initial external assets are at most \Euro3 per bank.
	\end{thm}
	
	\begin{proof}
        We build a polynomial-time reduction from the problem \textsc{3-Sat-3}, with the usual restrictions on instances (see the proof of Theorem~\ref{thm:bankmin}). Suppose that we have some instance $\phi$ of \textsc{3-Sat-3} where there are $n$ Boolean variables and $m$ clauses. We start with the \emph{chain gadget} chain$(l)$, where $l\geq 1$, as portrayed in Fig.~\ref{fig:chain gadget} (note that the gadget is the path of $l$ nodes within the blue dotted box). The key point about any chain gadget is that in some schedule: if at time 1, node $u$ does not make a payment to node $m_1$ then $u$ and all the nodes of the chain gadget are bankrupt; and if at time 1, $u$ pays its debt to $m_1$ then none of the nodes of the chain gadget is bankrupt. As in our proof of Theorem \ref{thm:bankminconstant}, we first work in the PP variant unless otherwise stated, though the reasoning will apply to the AoN and FP variants as well.
		
	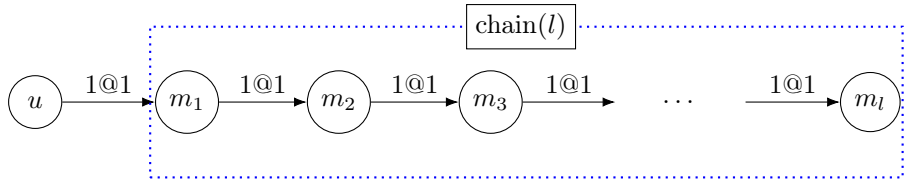
\begin{figure}[!ht]
			\centering
			\begin{tikzpicture}
				\mynode{u}{1}{0}{$u$}
				\mynode{m1}{3}{0}{$m_1$}
				\mynode{m2}{5}{0}{$m_2$}
				\mynode{m3}{7}{0}{$m_3$}
				
				\node[state,draw=none] (m4) at (9,0) {};%
				\node[state,draw=none] (m5) at (10,0) {};%

				\mynode{ml}{12}{0}{$m_l$}
				
				\path (u) edge node[above] {$1@1$} (m1);
				\path (m1) edge node[above] {$1@1$} (m2);
				\path (m2) edge node[above] {$1@1$} (m3);	
				\path (m3) edge node[above] {$1@1$} (m4);	
				\path (m5) edge node[above] {$1@1$} (ml);	
				\path (m4.east) -- node[auto=false]{\ldots} (m5.west);
				
				\node[draw=blue,dotted, thick, inner sep=1cm, minimum height = 35pt, minimum width=282pt] at (7.47,0) (gadget) {};
				
				\node[state,rectangle, fill=white!100] (label) at (7.4,1) {chain($l$)};

			\end{tikzpicture}
			\caption{The chain gadget. In any valid schedule, either $p_{u,m_1}^1<1$ and all vertices $m_i$ are bankrupt, or $p_{u,m_1}^1=1$ and no vertices $m_i$ are bankrupt.}
			\label{fig:chain gadget}
	\end{figure}

	We now define variable nodes $\{s_i:1\leq  i \leq n\}$, literal nodes $\{x_i, \neg x_i: 1\leq i \leq n\}$, clause nodes $\{q_j:1 \leq j \leq m\}$ and a sink node $d$ analogously to as in the proof of Theorem~\ref{thm:bankmin} and include the debts as depicted in Fig.~\ref{fig:bankruptcy maximization} so as to obtain our IDM instance $(G,D,A^0)$. Note that the chain gadgets corresponding to the different literal nodes are all distinct and $|c_j|$ denotes the number of literals in the clause $c_j$ of $\phi$.

	\begin{figure}[!ht]
	\centering
	\begin{tikzpicture}
		
		\myassetnode{si}{0}{0}{$s_i$}{$3$}
		\mynode{xi}{2}{2}{$x_i$}
		\mynode{nxi}{2}{-2}{$\lnot x_i$}
		\mynode{qj}{4}{0}{$q_j$}
		\mynode{d}{6}{0}{$d$}
		\node[state,rectangle, fill=white!100] (Tchain) at (6,2) {chain($m+1$)};
		\node[state,rectangle, fill=white!100] (Fchain) at (6,-2) {chain($m+1$)};

		\path (si) edge node[left, xshift=-5pt] {$3@1$} (xi);	
		\path (siasset) edge node[left, xshift=-5pt] {$3@1$} (nxi);
		\path (xi) edge  node[right, yshift=5pt] {$1@2$ if $\lnot v_i \in c_j$} (qj);	
		\path (nxi) edge node[right, yshift=-5pt] {$1@2$ if $v_i \in c_j$} (qj);	
		\path (qj) edge node[above] {$|c_j|@2$} (d);	
		\path (xi) edge node[above] {$1@1$} (Tchain);	
		\path (nxi) edge node[above] {$1@1$} (Fchain);	
		
	\end{tikzpicture}
	\caption{An IDM instance illustrating the reduction from \textsc{3-SAT 3} to \textsc{Bankruptcy Maximization}, using chain gadgets.}
	\label{fig:bankruptcy maximization}
\end{figure}
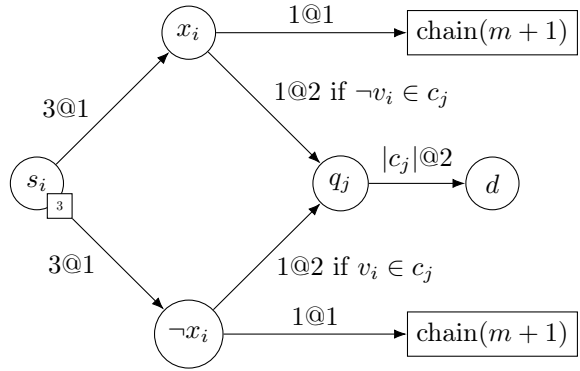

	Note that because all debts of $(G,D,A^0)$ are due at an exact time, rather than over an interval, reasoning in the AoN variant is identical to reasoning in the PP variant.

	\begin{clm}\label{clm:bankruptnodes}
		In any valid schedule $\sigma$ for $(G,D,A^0)$ in which $c\geq 0$ variable nodes $s_i$ either pay \Euro3 to $x_i$ or \Euro3 to $\neg x_i$:
		\begin{itemize}
			\item all $n$ variable nodes are bankrupt
			\item exactly $n$ literal nodes are bankrupt
			\item exactly $c(m+1)$ chain nodes are bankrupt.
		\end{itemize}
	Consequently, this amounts to exactly $2n+c(m+1)$ bankrupt nodes with any other bankrupt nodes necessarily being clause nodes.
	\end{clm}
	
	\begin{proof}
		Suppose that in the valid schedule $\sigma$, $s_i$, for some $1\leq i\leq n$, pays \Euro1 to a node from $\{x_i,\neg x_i\}$ at time 1 and \Euro2 to the other node from $\{x_i,\neg x_i\}$ at time 1. Since both $x_i$ and $\neg x_i$ have \Euro1 at time 1, both must pay \Euro1 to their corresponding chain gadget; so, none of the nodes of either of these chain gadgets is bankrupt. As any variable and its negation both appear at least once in some clause of $\phi$, exactly one of the nodes $x_i$ and $\neg x_i$ is bankrupt at time 2.
		
		Alternatively, suppose that $s_j$, for $1\leq j\leq n$, pays \Euro3 to either $x_j$ or $\neg x_j$ at time 1. So, the literal node to which $s_j$ makes no payment is bankrupt at time 1 as are all the nodes of its corresponding chain gadget. As any variable and its negation both appear at least once in some clause of $\phi$, exactly one of the nodes $x_j$ and $\neg x_j$ is bankrupt at time 2.
		
		In any case, we have: $n$ variable nodes that are bankrupt; $n$ literal nodes that are bankrupt; and $c(m+1)$ chain nodes that are bankrupt. This results in $2n+c(m+1)$ bankrupt nodes. As the sink node $d$ cannot be bankrupt, the claim follows.		
	\end{proof}

	\begin{clm}\label{clm:chainvalidimpliestruth}
		If $((G,D,A^0),2n+n(m+1)+m)$ is a yes-instance of \textsc{Bankruptcy Maximization} then $\phi$ is a yes-instance of \textsc{3-Sat-3}.
	\end{clm}
		
	\begin{proof}
	Suppose that there exists a valid schedule $\sigma$ that results in at least $2n+n(m+1)+m$ bankruptcies. So, by Claim~\ref{clm:bankruptnodes}, every variable node $s_i$ pays \Euro3 to either $x_i$ or $\neg x_i$ and also every clause node $q_j$ is bankrupt in $\sigma$. The reason a clause node $q_j$ is bankrupt is because there is some literal $v_i$ or $\neg v_i$ in clause $c_j$ but where $\neg x_i$ or $x_i$, respectively, receives no payment from $s_i$. Define the truth assignment $X$ on the variables of $\phi$ so that $X(v_i)=True$ iff $x_i$ receives a payment of \Euro3 from $s_i$, for $1\leq i\leq n$. This truth assignment satisfies every clause of $\phi$.
	\end{proof}

	\begin{clm}\label{clm:chaintruthimpliesvalid}
		If $\phi$ is a yes-instance of \textsc{3-Sat-3} then $((G,D,A^0),2n+n(m+1)+m)$ is a yes-instance of \textsc{Bankruptcy Maximization}.
	\end{clm}
	
	\begin{proof}
	Suppose that there is a satisfying truth assignment $X$ for $\phi$. Consider the following schedule $\sigma$:
	\begin{itemize}
		\item at time 1, every $s_i$ pays: \Euro3 to $x_i$ if $X(v_i)=True$; and \Euro3 to $\neg x_i$ if $X(v_i)=False$
		\item if $x_i$ (resp. $\neg x_i$) received \Euro3 from $s_i$ at time 1 then:
		\begin{itemize}
			\item at time 1, it pays \Euro1 to its corresponding chain gadget so as to satisfy all debts in the gadget
			\item at time 2, it pays \Euro1 to each clause node $q_j$ for which the literal $\neg v_i\in c_j$ (resp. $v_i\in c_j$)
		\end{itemize}
		\item if $x_i$ (resp. $\neg x_i$) received no payment from $s_i$ at time 1 then at times 1 and 2 then it can make no payments
		\item each $q_j$ makes a payment of however many euros it has to $d$ at time $2$ (note that it never received more than \Euro$|c_j|$).
	\end{itemize}
	The schedule $\sigma$ is clearly valid. By Claim~\ref{clm:bankruptnodes}, we have at least $2n+n(m+1)$ bankrupt nodes with any additional bankrupt nodes necessarily clause nodes. Consider some clause node $q_j$ containing some literal $v_i$ so that $X(v_i)=True$. By definition, $\neg x_i$ receives no payment from $s_i$ at time 1 and so the debt of \Euro1 at time 2 from $\neg x_i$ to $q_j$ is not paid. Consequently, $q_j$ is bankrupt. Hence, we have exactly $2n+n(m+1)+m$ bankrupt nodes and the claim follows.
	\end{proof}
		
	Note that the above reasoning clearly holds in the AoN variant (since in the schedules we consider all debts are paid either in full or not at all) as well as in the FP variant (since in order for $s_i$ to bankrupt one of the chains attached to $x_i$ and $\lnot x_i$ it must pay \emph{strictly} less than £1 to the bankrupt node and hence \emph{at least} £2 to the ``surviving'' node). 
    As the construction of $((G,D,A^0),2n+n(m+1)+m)$ can be completed in time polynomial in $n$, the result follows.\end{proof}		

	Just as we did with \textsc{Perfect Scheduling} in Theorem~\ref{thm:AoN3path}, we can restrict \textsc{Bankruptcy Maximization} in the AoN variant so that any IDM $(G,D,A^0)$ in any instance is such that $G$ is a directed path of bounded length (here 2).
		
	\begin{thm}\label{thm:AoN3pathBankMax}
			Consider the problem \textsc{Bankruptcy Maximization} restricted so that every instance $((G,D,A^0),k)$ is such that $G$ is a directed path of length 3. If, further, $T$ is restricted to be 2 then the resulting problem is weakly NP-complete in the AoN variant.
	\end{thm}
	
	\begin{proof}
		We reduce from the weakly NP-complete problem \textsc{Subset Sum} defined as follows (see \cite{Garey-Johnson79}).\newpage

        \noindent\underline{\textsc{Subset Sum}}
		\begin{description}
			\item Instance: a multi-set of integers $S=\{a_1,a_2,\ldots,a_{n}\}$ and an integer $k$
			\item Yes-instance: there exists a subset $S_1$ of $S$ so that the sum of the numbers in $S_1$ is $k$.
		\end{description}
		
		In general, an instance $S=\{a_1,a_2,\ldots,a_{n}\}$ has size $nb$ where $b$ is the least number of bits required to express any of the integers of $S$ in binary. 	
		Let $S$ be an instance of \textsc{Subset Sum} of size $nb$. By doubling all integers if necessary, we may assume that every integer of $S$ is at least 2. Consider the IDM instance $(G,D,A^0)$ in Fig.~\ref{fig:AoN bankmax} (for which $T=2$). The value $A$ in Fig.~\ref{fig:AoN bankmax} is the sum of all integers in $S$. Note that the time taken to construct $(G,D,A^0)$ from $S$ is polynomial in $nb$.
		
			\begin{figure}[!ht]
			\centering
			\begin{tikzpicture}
				\myassetnode{u}{0}{0}{$u$}{$A$}
				\myassetnode{v}{3}{0}{$v$}{$k$}
				\mynode{w}{6}{0}{$w$}
				
				\path (v) edge node[align=center,above] {$a_1@[1,2]$\\\ldots\\$a_{n}@[1,2]$} (w);
				\path (v) edge node[align=center,below] {$1@1$} (w);
				\path (u) edge node[align=center] {$A@2$} (v);
			\end{tikzpicture}
			
			\caption{An IDM instance corresponding to an instance of \textsc{Subset Sum}.}
			\label{fig:AoN bankmax}
		\end{figure}
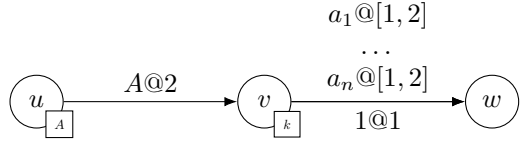
		
		\begin{clm}
			If $((G,D,A^0),1)$ is a yes-instance of \textsc{Bankruptcy Maximization} then $S$ is a yes-instance of \textsc{Subset Sum}.
		\end{clm}
		
		\begin{proof}
			Suppose that $\sigma$ is a valid schedule within which there is at least 1 bankruptcy in the IDM instance $(G,D,A^0)$. The nodes $u$ and $w$ are never bankrupt; so, $v$ must be bankrupt within $\sigma$. At time 2, the node $v$ necessarily pays off all of the unpaid debts to $w$ of monetary amount greater than \Euro1, as it receives sufficient funds from $u$ at time 2 to do this. Hence, $v$ must be bankrupt at time 1; that is, $v$ does not pay its debt to $w$ of monetary amount \Euro1 at time 1. As $\sigma$ is valid, $v$ is not withholding at time 1 and the only way for this to happen is for $v$ to pay debts to $w$ amounting to \Euro$k$. That is, we have a subset of integers of $S$ whose total sum is $k$; that is, $S$ is a yes-instance of \textsc{Subset Sum}. The claim follows.	
		\end{proof}
		
		\begin{clm}
			If $S$ is a yes-instance of \textsc{Subset Sum} then $((G,D,A^0),1)$ is a yes-instance of \textsc{Bankruptcy Maximization}.
		\end{clm}
		
		\begin{proof}
			Suppose that the subset $S_1$ of $S$ is such that $sum(S_1)=k$. W.l.o.g. let $S_1=\{a_1,a_2,\ldots,a_r\}$. Define the schedule $\sigma$ as: at time 1, $v$ pays the debts $a_1@[1,2], \ldots,a_r@[1,2]$; and at time $2$, $u$ pays its debt to $v$ and $v$ pays the debts $a_{r+1}@[1,2],\ldots,a_n@[1,2]$. Note that this is a valid schedule, as no node is withholding at any time, within which $v$ is bankrupt. The claim follows.
		\end{proof}
		The main result follows.
	\end{proof}

	\subsection{Polynomial-time algorithms}\label{subsec:poly}

	In this section we show that \textsc{Bailout Minimization} in the FP variant is solvable in polynomial-time and also that \textsc{Bailout Minimization} in the PP variant is solvable in polynomial-time when our IDM instances are restricted to out-trees. We begin with the FP variant result.
	
	\begin{thm}\label{thm:ptime fractional}
		The problem \textsc{Bailout Minimization} in the FP variant is solvable in polynomial time.
	\end{thm}
	
     \begin{proof}
         
    A solution to \textsc{FP Bailout Minimization} is a bailout vector $B$ of size $|V|$ together with a schedule $\sigma$ consisting of $|E|T$ payment values $p_e^t$. We describe below how an instance $G, D, A^0, b$ of \textsc{Bailout Minimization} can be encoded as a linear program (LP), which can then be solved in polynomial time. 
    
    Our variables are:
    \begin{itemize}
        \item Bailout variables $\{B[v]| v \in V\}$ (altogether $|V|$ variables), 
        \item Payment variables $\{p_e^t| e \in e, t \in [T]\}$ (altogether $|E|T$ variables), and
        \item Income variables $\{I_v^t|v \in V, t \in [T]\}$, outgoing variables $\{O_v^t|v \in V, t \in [T]\}$, and cash asset variables $\{c_v^t|v \in V, t \in [0,T]\}$ (altogether $3\cdot|V|\cdot T$ variables). 
    \end{itemize}
    
    
    In the below, for $a,b \in \mathbb N_0$, $[a,b]$ denotes the set $\{a, a+1, \ldots, b\}$, and we write $[b]$ as shorthand for $[1,b]$. Our constraints are:
    \begin{itemize}
        \item The total bailout is at most $b$:
            $$\sum_v B[v] \le b$$
        \item The starting cash assets of a node (at time $0$) are its external assets (specified by $A^0$) plus any bailout it receives. For each $v \in V$:
            $$c_v^0 = A^0[v] + B[v]$$
        \item No debt is paid early, and all payments are non-negative. For each $e \in E$ and $t \in [0, T]$:
        $$
            p_e^t~\begin{cases}
           = 0, & \text{if}\ t < D_{t_1}(e) \\
          \ge 0 , & \text{otherwise}
        \end{cases}
        $$
        \item The income (resp. outgoings) of a node at some time are obtained by summing over payments into (resp. out of) that node at each time. These then can be used to compute external (cash) assets at all nodes and times. For each $v \in V$ and $t \in [T]$:
            $$I_v^t = \sum_{e \in E_{in}(v)} p_e^t ~\textrm{ (resp. }~ O_v^t = \sum_{e \in E_{out}(v)} p_e^t)$$
            $$e_v^t=e_v^{t-1} + I_v^t - O_v^t$$
        \item No bank has negative assets at any point. For each $v \in V, t \in [T]$:
            $$e_v^t \ge 0$$
        \item Each debt is paid in full within its interval. This guarantees that there are no bankruptcies in the schedule (and that no banks are withholding, a validity constraint). For each $e\in E$ with $D(e)=(a, t_1,t_2)$:
            $$\sum_{t\in [t_1, t_2]}p_e^t=a$$
    \end{itemize}
    
    Recall from our discussion of canonical instances in Section 2.4 that we may assume $T$ is at most $2|E|$. Then we have $O(nm+m^2)$ variables and $O(nm+m^2)$ constraints. If the largest integer in the input instance $G,D,A^0,b$ required $\beta$ bits to encode, then our constructed LP has size polynomial in $n+m+\beta$. Any assignment to $B$ and to the payment variables $p_e^t$ satisfying the above is necessarily a perfect valid schedule for $((G,D,A^0),b)$. As linear programs can be solved in polynomial-time, our result follows.
    \end{proof}
        
	Note the limitations of the use of linear programming for other problems. For \textsc{Bailout Minimization} in the PP variant, proceeding as in the proof of Theorem~\ref{thm:ptime fractional} results is an integer linear program, the solution of which is NP-complete in general. Moreover, we have already proven \textsc{Perfect Scheduling}, the special case of \textsc{Bailout Minimization} with $b$ fixed to $0$, to be NP-complete in the PP variant through the proofs in Theorems~\ref{thm:perfschedDAG}, \ref{thm:perfsched multiditree} and \ref{thm:perfschedhampath}. As regards trying to use linear programming for \textsc{Bankruptcy Minimization} in the FP variant, it is not possible to express a constraint on the number of bankruptcies through a linear combination of the payment variables; indeed, we have already proven \textsc{Bankruptcy Minimization} in the FP variant to be NP-complete in Theorems~\ref{thm:bankmin} and \ref{thm:bankminconstant}.
	
	For the AoN and PP variants, by restricting the temporal properties of the IDM instances considered, we obtain tractability of \textsc{Bailout Minimization}, namely when all debts are due at an exact time.
    
    \begin{thm}\label{thm:ptime-t1eqt2}
    The problem \textsc{(AoN/PP/FP) Bailout Minimization} is solvable in polynomial time when restricted to inputs $(G,D,A^0)$ such that $D_{t_1}=D_{t_2}$. 
    \end{thm}
    \begin{proof}
    Let $(G,D,A^0)$ be an IDM instance satisfying the above. In such an instance, all debts are due at an exact point in time, rather than an interval. For convenience, we use $D_t$ as shorthand for either of $D_{t_1}$ or $D_{t_2}$.
    By definition, for any bailout vector $B$ (including the all-zero vector) a perfect schedule for $(G,D,A^0+B)$ is one in which every debt is paid in full and on time. Let $\sigma$ be the schedule defined by $p_e^{D_t(e)} = D_a(e)$ for each edge $e$, with all other payment variables equal to zero. 
    Clearly, for any vector $B$, a perfect schedule for $(G,D,A^0+B)$ exists if and only if $\sigma$ is a valid schedule (and hence a perfect schedule).
    
    Moreover, we can efficiently compute a vector $B$ of minimum sum such that $\sigma$ is a perfect schedule for $(G,D,A^0+B)$. 
    For each vertex $v$ and time $t$, compute $c_v^t$ under $\sigma$ for the instance $(G,D,A^0)$. Note that if $(G,D,A^0)$ does not admit a perfect schedule then $c_v^t$ will be negative for some $v$ and $t$, and $\sigma$ is not a valid schedule for that instance (without a bailout).
    Denote the minimum (again, possibly negative) cash assets of $v$ at any time by $c_v^{\min{}}$.
    Compute $b_v := \max(-1 \cdot c_v^{\min{}},0)$ for each $v$, and let $B=(b_v | v \in V)$. By construction, $\sigma$ is a perfect schedule for $(G,D,A^0+B)$, and $\sigma$ is not a perfect schedule for $(G,D,A^0+B')$ for any $B'$ with $\sumop(B')< \sumop(B)$.
    
    All of our arguments hold in all three variants (AoN, PP, and FP), and the result follows.
    \end{proof}

    Interestingly, Theorem \ref{thm:ptime-t1eqt2} is the only positive result we derive for the All-or-Nothing setting. We also obtain tractability results for the problem \textsc{PP Bailout Minimization} if we restrict the structure of IDM instances.
	
	\begin{thm}\label{thm:ptime out-tree}
		The problem \textsc{Bailout Minimization} in the PP variant is solvable in polynomial-time when our IDM instances are restricted to out-trees.  
	\end{thm}
	
	\begin{proof}
		Let $((G,D,A^0),b)$ be an instance of \textsc{Bailout Minimization} so that $G$ is an out-tree. Suppose that $G$ has node set $\{u_i:1\leq i\leq n\}$. We need to decide whether we can increase the initial external assets of each node $u_i$ by $b_i$ so that $\sum_i b_i \leq b$ and $(G,D,A^0+B)$ has a perfect schedule, where $B=(b_1,b_2,\ldots,b_n)$; that is, whether $(G,D,A^0)$ is `$b$-bailoutable' via a \emph{$b$-bailout vector} $B$. Our intention is to repeatedly amend $(G,D,A^0)$ so that $(G,D,A^0)$ is `$b$-bailoutable' iff the resulting IDM instance is `$b^\prime$-bailoutable', for some amended $b^\prime$; in such a case, we say that the two problem instances are \emph{equivalent}. We will then work with the (simplified) amended instance.
		
		We proceed as follows. First, identify nodes $v$ for which, at any time $t$, the sum of all debts $v$ must pay by time $t$ minus the sum of all debts which could be paid to $v$ by time $t$ exceeds $v$'s initial external assets $c_v^0$. We call these nodes \emph{prefix-insolvent}. Note that if a node $v$ is prefix-insolvent then under any perfect schedule $\sigma$, we would have $I_v^{[t]} + c_v^0 < O_v^{[t]}$, violating a validity constraint, and hence there is no such perfect schedule. Also note that every insolvent node is prefix-insolvent (namely by taking $t=T$). For any node that is prefix-insolvent, increase the initial external assets by the minimal amount that causes the node to cease to be prefix-insolvent and simultaneously decrease the bailout amount by this value. Our new instance is clearly equivalent with our initial instance. If, in doing this, the bailout amount becomes less than 0 then we answer `no' and we are done. So, we may assume that none of our nodes is prefix-insolvent.
		
		Before we start, we make a simple amendment to the debts $D$: we replace any debt from node $u$ to node $v$ of the form $a@[t_1,t_2]$, where $a > 1$, with $a$ distinct debts $1@[t_1,t_2]$. Our resulting instance, with bailout $b$, is equivalent to our initial instance, with bailout $b$, as we are working within the PP variant. This amendment simplifies some of the reasoning coming up. Note that it may be necessary to simulate this operation rather than actually performing it (since if $a$ is exponential in the instance size then the operation takes exponential time), but that the reasoning which follows can easily be ``scaled up'' to deal with non-unit amounts. 
		
		Consider a leaf node $v$ and its parent $u$ in the out-tree $G$. Replace every debt of the form $1@[t_1,t_2]$ from $u$ to $v$ by the debt $1@t_2$ and denote the revised instance by $(G,D^\prime,A^0)$. Let $\sigma$ be a perfect valid schedule for $(G,D,A^0+B)$ (here, and throughout, we write $B$ to denote some $b$-bailout vector; that is, some assignment of resource to the nodes of $G$ so that the total bailout amount does not exceed the total available $b$). Define the schedule $\sigma^\prime$ for $(G,D^\prime,A^0+B)$ by amending any payment from $u$ to $v$ of some debt $1@[t_1,t_2]$ so that the payment is made at time $t_2$. The schedule $\sigma^\prime$ is clearly a perfect valid schedule for $(G,D^\prime,A^0+B)$. Furthermore, any perfect valid schedule for $(G,D^\prime,A^0+B)$ is a perfect valid schedule for $(G,D,A^0+B)$. Hence, we can replace $((G,D,A^0),b)$ by $((G,D^\prime,A^0),b)$, as these instances are equivalent. We can proceed as above for every leaf node and its parent and so assume that all debts from a parent to a leaf are due at some specific time only; that is, have a singular time-stamp and are of the form $1@t$. Note also that no node of $G$ is insolvent.
		
		Suppose that we have two leaf nodes $v$ and $w$ with the same parent $u$. We can replace $v$ and $w$ with a `merged' node $vw$ so that all debts from $u$ to $v$ or from $u$ to $w$ are now from $u$ to $vw$. Our initial problem instance is clearly equivalent to our amended problem instance (note that we never assign bailout resource to either $v$ or $w$ as this is pointless). We can proceed likewise for all such triples $(u,v,w)$. Hence, we may assume that our digraph $G$ is such that: no two leaves have a common parent; all debts to a leaf node have monetary amount \Euro1 and have a singular time-stamp; and no node in $G$ is prefix-insolvent.
		
		Suppose that we have a leaf node $w$ that is the only child of its parent node $v$ whose parent is $u$ (such a node $w$ exists: take a leaf of the tree that is furthest away from the root). We may assume that $v$ has no initial external assets as we would simply use these assets to pay as many debts to $w$ as possible (in increasing order of time-stamp); that is, we could remove all these debts from $D$ along with the corresponding amount from the initial external assets of $v$. If the result of doing this is that there are no debts from $v$ to $w$ then we remove $w$ from $G$ and any remaining initial external assets from $v$. We would then repeat all of the above amendments until it is the case that our nodes $u$, $v$ and $w$ are such that $v$ has no initial external assets.
		
		By the above, every debt from $v$ to $w$ is of the form $1@t$. Let $t^\prime$ be the minimum time-stamp for all debts from $v$ to $w$ and let $d_v$ be a debt from $v$ to $w$ of the form $1@t^\prime$. Consider the debts from $u$ to $v$: these have the form $1@[t_1,t_2]$. There are various cases:
		\begin{itemize}
			\item[(\emph{a})] there is a debt $d_u$ from $u$ to $v$ of the form $1@[t_1,t_2]$ where $t_2\leq t^\prime$
			\item[(\emph{b})] there is no debt from $u$ to $v$ of the form $1@[t_1,t_2]$ where $t_2\leq t^\prime$ but there is a debt from $u$ to $v$ of the form $1@[t_1,t_2]$ where $t_1\leq t^\prime \leq t_2$
			\item[(\emph{c})] there is no debt from $u$ to $v$ of the form $1@[t_1,t_2]$ where $t_1\leq t^\prime$.
		\end{itemize}
		The nodes $u$, $v$ and $w$ can be visualized as in Fig.~\ref{fig:bankruptcy minim}(d) and the three cases above in Fig.~\ref{fig:bankruptcy minim}(a)--(c).
		
		\begin{figure*}[!ht]
			\centering
			\subfloat[\label{fig:bankruptcy minima}]{
			\begin{tikzpicture}
				\node(t1) at (0,2) {$t_1[$};
				\node(t2) at (1,2) {$]t_2$};
				\node(t3) at (1,1.5) {$t_3[$};
				\node(t4) at (2.5,1.5) {$]t_4$};
				\node(t5) at (-0.25,1) {$t_5[$};
				\node(t6) at (2.7,1) {$]t_6$};
				\node(t7) at (0,2.5) {$t_7[$};
				\node(t8) at (2,2.5) {$]t_8$};
				\node(tprime) at (1.5,3) {$t^\prime$};
				\node(bottprime) at (1.5,0.5) { };
				\path (t1) edge[-] (t2);
				\path (t3) edge[-] (t4);
				\path (t5) edge[-] (t6);
				\path (t7) edge[-] (t8);
				\path (tprime) edge[-, dashed] (bottprime);
			\end{tikzpicture}
			}	
			\hfill
			\subfloat[\label{fig:bankruptcy minimb}]{
				\begin{tikzpicture}
					\node(t1) at (0,2) {$t_1[$};
					\node(t2) at (2,2) {$]t_2$};
					\node(t3) at (1,1.5) {$t_3[$};
					\node(t4) at (2.5,1.5) {$]t_4$};
					\node(t5) at (-0.25,1) {$t_5[$};
					\node(t6) at (2.7,1) {$]t_6$};
					\node(t7) at (2,2.5) {$t_7[$};
					\node(t8) at (3,2.5) {$]t_8$};
					\node(tprime) at (1.5,3) {$t^\prime$};
					\node(bottprime) at (1.5,0.5) { };
					\path (t1) edge[-] (t2);
					\path (t3) edge[-] (t4);
					\path (t5) edge[-] (t6);
					\path (t7) edge[-] (t8);
					\path (tprime) edge[-, dashed] (bottprime);
				\end{tikzpicture}
			}
			\hfill
			\subfloat[\label{fig:bankruptcy minimc}]{
				\begin{tikzpicture}
					\node(t1) at (0,2) {$t_1[$};
					\node(t2) at (2,2) {$]t_2$};
					\node(t3) at (0.5,1.5) {$t_3[$};
					\node(t4) at (2,1.5) {$]t_4$};
					\node(t5) at (-0.25,1) {$t_5[$};
					\node(t6) at (2.7,1) {$]t_6$};
					\node(t7) at (1,2.5) {$t_7[$};
					\node(t8) at (2.25,2.5) {$]t_8$};
					\node(tprime) at (-0.5,3) {$t^\prime$};
					\node(bottprime) at (-0.5,0.5) { };
					\path (t1) edge[-] (t2);
					\path (t3) edge[-] (t4);
					\path (t5) edge[-] (t6);
					\path (t7) edge[-] (t8);
					\path (tprime) edge[-, dashed] (bottprime);	
				\end{tikzpicture}
			}
			\\
			\subfloat[\label{fig:bankruptcy minimd}]{
			\begin{tikzpicture}
				\mynode{w}{1.2}{0}{$w$}
				\mynode{v}{1.2}{2}{$v$}
				\mynode{u}{1.2}{4}{$u$}
				\path (u) edge node[right, align=center] {$1@[t_1,t_2]$\\$1@[t_3,t_4]$\\$\ldots$} (v);	
				\path (v) edge node[right, align=center] {$1@t^\prime$\\$1@t^{\prime\prime}$\\$\ldots$} (w);		
			\end{tikzpicture}
			}
		\hspace{1in}
		\subfloat[\label{fig:bankruptcy minime}]{
			\begin{tikzpicture}
				\mynode{w}{1.2}{0}{$w$}
				\mynode{v}{1.2}{2}{$v$}
				\mynode{u}{1.2}{4}{$u$}
				\mynode{vprime}{-0.8}{2}{$v^\prime$}
				\path (u) edge node[right, align=center] {\sout{$1@[t_1,t_2]$}\\$1@[t_3,t_4]$\\$\ldots$} (v);	
				\path (v) edge node[right, align=center] {\sout{$1@t^\prime$}\\$1@t^{\prime\prime}$\\$\ldots$} (w);
				\path (u) edge node[left, align=center] {$1@t^\prime$} (vprime);
			\end{tikzpicture}
		}
			\caption{Cases when a leaf has a parent with one child in our algorithm for \textsc{PP Bailout Minimization} on out-trees.}
			\label{fig:bankruptcy minim}
		\end{figure*}
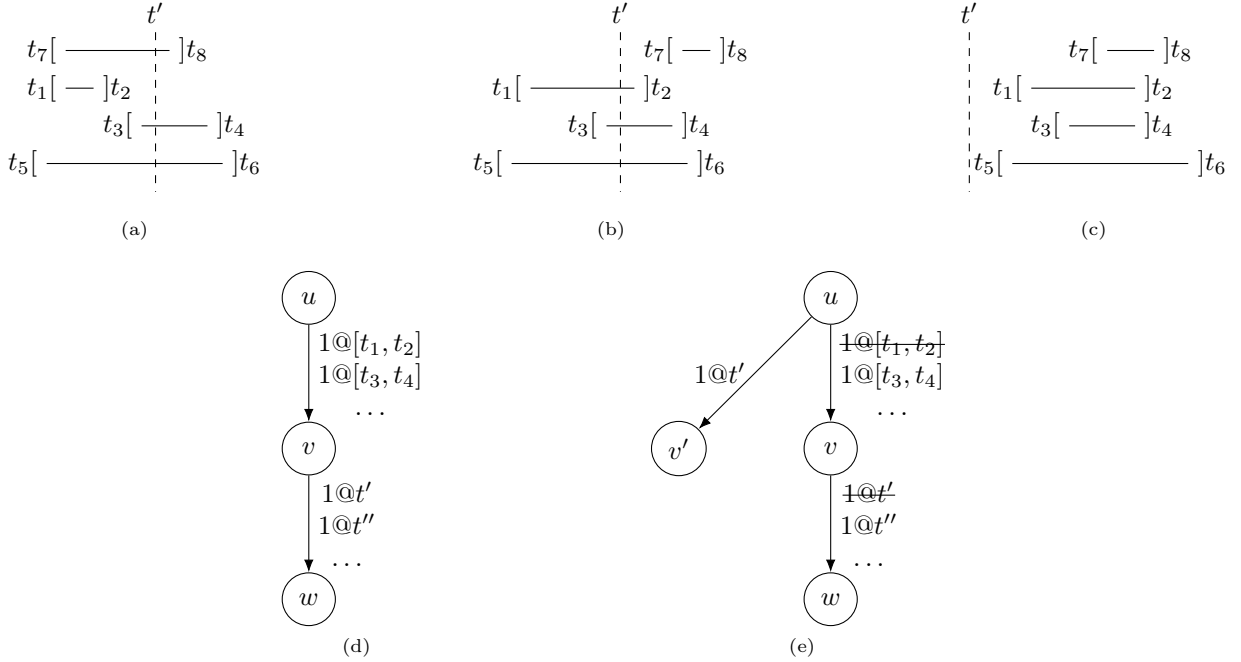
		
		\noindent\underline{Case (\emph{a})}: We amend $G$ by: introducing a new node $v^\prime$ whose parent is $u$ and a new debt $d^\prime$ from $u$ to $v^\prime$ of the form $1@[t_1,t_2]$; removing the debt $d_u$ from $u$ to $v$; and removing the debt $d_v$ from $v$ to $w$. Denote this revised IDM instance by $(G^\prime,D^\prime,A^0)$. Suppose that there exists a perfect valid schedule $\sigma$ for $(G,D,A^0+B)$. Define the schedule $\sigma^\prime$ for $(G^\prime,D^\prime,A^0+B)$ from $\sigma$ by: changing the payment from $u$ to $v$, at time $t$ where $t_1\leq t\leq t_2\leq t^\prime$ and covering the debt $d_u$, so that the payment is made from $u$ to $v^\prime$ at time $t$ (so as to cover the new debt $d^\prime$); and dropping the payment, at time $t^\prime$, that covers the debt $d_v$. The resulting schedule $\sigma^\prime$ is clearly valid and perfect. Conversely, suppose that we have a perfect valid schedule $\sigma^\prime$ for $(G^\prime,D^\prime,A^0+B)$. Define the schedule $\sigma$ for $(G,D,A^0+B)$ from $\sigma^\prime$ by: changing the payment from $u$ to $v^\prime$ at time $t$ where $t_1\leq t\leq t_2\leq t^\prime$ and covering the debt $d^\prime$, so that the payment is made from $u$ to $v$ at time $t$, so as to cover the debt $d_u$; and using the \Euro1 received by $v$ so as to cover the debt $d_v$ from $v$ to $w$. The resulting schedule $\sigma$ is clearly a perfect valid schedule for $(G,D,A^0+B)$. Consequently, $((G,D,A^0),b)$ and $((G^\prime,D^\prime,A^0), b)$ are equivalent and we can work with $((G^\prime,D^\prime,A^0), b)$.\smallskip	
		
		\noindent\underline{Case (\emph{b})}: Let $D_u$ be the set of debts from $u$ to $v$ and let $D_v$ be the set of debts from $v$ to $w$. Order the $k$ debts of $D_v$ in increasing order of time-stamp as $d_v=d_1, d_2, \ldots, d_k$ where the corresponding time-stamps are $t^\prime = \bar{t}_1, \bar{t}_2, \ldots, \bar{t}_k$ (there may be repetitions). Suppose that for some $1\leq i\leq k$, the number of debts in $D_u$ of the form $1@[t_1,t_2]$ with $t_1\leq \bar{t}_i$ is strictly less than $i$. Consequently, at time $\bar{t}_i$, the debt $d_i$ cannot be paid and we necessarily need to give $v$ some bailout amount, \Euro$c>1$ say, to cover the $c$ debts that cannot be paid by $v$ at time $\bar{t}_i$. We do this and reduce the overall bailout amount by \Euro$c$. We then delete debts $d_1,d_2,\ldots, d_c$ from $G$ and remove the bailout amount of \Euro$c$ from $v$. If doing this results in there being no remaining debts from $v$ to $w$ then we delete $w$ from $G$. Irrespective of this, the resulting instance $((G^\prime, D^\prime, A^0), b^\prime)$, where $b^\prime = b-c$, is equivalent to $((G,D,A^0),b)$. We would then repeat all of the amendments above and so w.l.o.g. we may assume that we are in Case (\emph{b}) and for every $1\leq i\leq k$, the number of debts in $D_u$ of the form $1@[t_1,t_2]$ with $t_1\leq \bar{t}_i$ is at least $i$. In particular, there are at least $k$ debts in $N_u$.
		
		Suppose that $\sigma$ is a valid perfect schedule for $(G,D,A^0+B)$, for some $B$ where $v$ receives a bailout amount of \Euro$c>0$. Suppose further that there is no $B^\prime$ where there is a valid perfect schedule for $(G,D,A^0+B^\prime)$ with $v$ receiving a bailout of less than \Euro$c$. The reason that $v$ receives the bailout amount of \Euro$c$ is that in the schedule $\sigma$, if we ignore the payments by $v$ that use the bailout amount at $v$ then there are $c$ debts from $d_1, d_2, \ldots, d_k$ that are not paid on time; let us call these debts the `bad' debts. Note that for each bad debt, there is a debt from $u$ to $v$ that \emph{might} have been paid at a time early enough to cover the debt but wasn't. Let $e_1,e_2,\ldots,e_c$ be distinct debts from $D_u$ that might have been paid earlier so as to enable the payment of the bad debts. Amend the bailout $B$ so that the \Euro$c$ formerly given as bailout to $v$ is now given to $u$ and denote this revised bailout by $B^\prime$. Revise the schedule $\sigma$ so that for each $1\leq i\leq c$, \Euro1 of bailout at $u$ is used to pay the debt $e_i$ at the earliest time possible. Doing so results in us being able to pay all bad debts on time; hence, we have a perfect schedule for $(G,D,A^0+B^\prime)$ where $v$ receives no bailout. This yields a contradiction and so if there is a bailout $B$ and a valid perfect schedule for $(G,D,A^0+B)$ then there is a bailout $B^\prime$ and a valid perfect schedule for $(G,D,A^0+B^\prime)$ where $v$ receives no bailout funds. We shall return to this comment in a moment.
		
		From the debts of $D_u$, we choose a debt $d_u = 1@[t_1,t_2]$, where $t_1\leq t^\prime \leq t_2$, so that from amongst all of the debts of $D_u$ of the form $1@[t_3,t_4]$, where $t_3\leq t^\prime \leq t_4$, we have that $t_2\leq t_4$; that is, from all of the debts of $D_u$ that `straddle' $t^\prime$, $d_u$ is a debt whose right-most time-stamp is smallest. We amend $G$ by: introducing a new node $v^\prime$ and a new debt $d^\prime$ from $u$ to $v^\prime$ of the form $1@t^\prime$; removing the debt $d_u$ from $u$ to $v$; and removing the debt $d_v$ from $v$ to $w$. Denote this revised IDM instance by $(G^\prime, D^\prime, A^0)$; it can be visualized as in Fig.~\ref{fig:bankruptcy minim}(e).
		
		Suppose that $\sigma$ is a valid perfect schedule for $(G,D,A^0+B)$, for some $B$. From above, we may assume that there is no bailout to node $v$ in $B$. Consider the payment by $v$ to $w$ of the debt $d_v$. If the actual \Euro1 that pays this debt came from the payment of a debt from $D_u\setminus\{d_u\}$ of the form $1@[t_3,t_4]$ (where $t_3\leq t^\prime\leq t_4$), then we can amend $\sigma$ so that we use this \Euro1 to pay the debt $d_u$ at the time $t^\prime$ and use the \Euro1 that paid the debt $d_u$ to pay the debt $1@[t_3,t_4]$ (at whatever time $d_u$ was paid); that is, we swap the times of the payment of the debts $d_u$ and $1@[t_3,t_4]$ in $\sigma$ except that we now pay $d_u$ at time $t^\prime$. If the actual \Euro1 that pays $d_v$ came from the payment of $d_u$ then we can amend the payment time of $d_u$ to $t^\prime$ (if necessary).
		
		Build a schedule $\sigma^\prime$ in $(G^\prime, D^\prime, A^0+B)$ from $\sigma$ by: instead of paying $d_u$ (at time $t^\prime$), we pay the new debt $d^\prime$ from $u$ to $v^\prime$; and we remove the payment of the debt $d_v$. The schedule $\sigma^\prime$ is clearly a valid perfect schedule of $(G^\prime, D^\prime, A^0+B)$. Conversely, if $\sigma^\prime$ is a valid perfect schedule of $(G^\prime, D^\prime, A^0+B)$, we can build a schedule $\sigma$ for $(G,D,A^0+B)$ from $\sigma^\prime$ by: instead of paying the debt $d^\prime$ (at time $t^\prime$), we pay the debt $d_u$ at time $t^\prime$; and we use this \Euro1 to pay immediately the debt $d_u$. The schedule $\sigma$ is clearly a valid perfect schedule of $(G,D,A^0+B)$. Hence, $((G,D,A^0),b)$ and $((G^\prime, D^\prime, A^0), b)$ are equivalent and we can work with $((G^\prime, D^\prime, A^0),b)$.\smallskip
		
		\noindent\underline{Case (\emph{c})}: Suppose that there is no debt from $u$ to $v$ of the form $1@[t_1,t_2]$ where $t_1\leq t^\prime$. This case cannot happen as we have ensured that no node of $G$ is insolvent.\smallskip
	
		By iteratively applying all of the amendments to the instance $((G,D,A^0),b)$, as laid out above, we reduce $((G,D,A^0),b)$ to an equivalent instance $((G^\prime,D^\prime,$\linebreak$(A^\prime)^0),b^\prime)$ where $G$ consists of a solitary directed edge and all debts have a singular time-stamp. The process of reduction can clearly by undertaken in time polynomial in the size of the initial instance $((G,D,A^0),b)$ and the resulting instance $((G^\prime,D^\prime,(A^\prime)^0),b^\prime)$ can clearly be solved in time polynomial in the size of the initial instance $((G,D,A^0),b)$. Hence, \textsc{Bailout Minimization} is solvable in polynomial-time.
		\end{proof}
		
		Our polynomial-time algorithm for \textsc{Bailout Minimization}, and so \textsc{Perfect Scheduling}, in the PP variant in Theorem~\ref{thm:ptime out-tree} when we restrict to out-trees contrasts with the NP-completeness of \textsc{Perfect Scheduling} when we restrict to directed acyclic graphs or multiditrees, as proven in Theorem~\ref{thm:perfschedDAG} and Theorem~\ref{thm:perfsched multiditree}, respectively.
	
	\section{Conclusion and open problems}\label{section:conclusion}

		This paper introduces the \emph{Interval Debt Model (IDM)}, a new model seeking to capture the temporal aspects of debts in financial networks. We investigate the computational complexity of various problems involving debt scheduling, bankruptcy and bailout with different payment options (All-or-nothing (AoN), Partial (PP), Fractional (FP)) in this setting. We prove that many variants are hard even on very restricted inputs but certain special cases are tractable. For example, we present a polynomial time algorithm for \textsc{PP Bailout Minimization} where the IDM graph is an out-tree. However, for a number of other classes (DAGs, multitrees, total assets are $\Euro1$), we show that the problem remains NP-hard. This leaves open the intriguing question of the complexity status of problems which are combinations of two or more of these constraints, most naturally on multitrees which are also DAGs, an immediate superclass of our known tractable case. 

        An interesting result of ours is the (weak) NP-completeness of \textsc{Bankruptcy Minimization} on a fixed, 32-node footprint graph (with edge multiplicity unbounded) in Theorem \ref{thm:bankminconstant}. It is noteworthy that constantly many nodes suffice to express the complexity of any problem in NP, and this leads to several open questions. Does the same hold when integers must be encoded in unary? We know this is true for the AoN case (as shown in Theorem \ref{thm:AoN3path}).
        What is the smallest number $n$ such a family of $n$-node \textsc{(FP/PP) Bankruptcy Minimization} instances is NP-complete? From the other side, what is the largest number $n$ such that any $n$-node \textsc{(FP/PP) Bankruptcy Minimization} instance may be solved in polynomial time, and with what techniques?

		We prove that \textsc{FP Bailout Minimization} is polynomial-time solvable by expressing it as a Linear Program. 
		Can a similar argument be applied to some restricted version of FP Bankruptcy Minimization (which is NP-Complete, in general)? A natural generalization is simultaneous Bailout and Bankruptcy minimization i.e. can we allocate $\Euro b$ in bailouts such that a schedule with at most $k$ bankruptcies becomes possible. 
		Variations of this would be of practical interest. For example, if regulatory authorities can allocate bailouts as they see fit, but not impose specific payment times, it would be useful to consider the problem of allocation of $\Euro b$ in bailouts such that the maximum number of bankruptcies in any valid schedule is at most $k$. Conversely, where financial authorities can impose specific payment times, the combination of the problems Bankruptcy Minimization and Bailout Minimization would be more applicable.

		Finally, can we make our models more realistic and practical? How well do our approaches perform on real-world financial networks? Can we identify topological and other properties of financial networks that may be leveraged in designing improved algorithms? What hardness or tractability results hold for variants in which the objective is, instead of the number of bankruptcies, the total amount of unpaid debt (or any other objective, for that matter)?
  
\section*{Declarations}
This work was partially supported by Engineering and Physical Sciences Research Council grant EP/P020372/1.

The authors have no relevant financial or non-financial interests to disclose.
The authors have no conflicts of interest to declare that are relevant to the content of this article.
All authors certify that they have no affiliations with or involvement in any organization or entity with any financial interest or non-financial interest in the subject matter or materials discussed in this manuscript.
The authors have no financial or proprietary interests in any material discussed in this article.

\backmatter 

\bibliography{tcs_idm}

\end{document}